\documentclass[12pt]{article}
 
\usepackage[capposition=top]{floatrow}
\usepackage[margin=1in]{geometry} 
\usepackage{amsmath,amsthm,amssymb}
\usepackage{natbib}
\usepackage{setspace}
\usepackage{hyperref}
\usepackage{graphicx}
\usepackage[utf8]{inputenc}
\usepackage{booktabs,caption}
\usepackage[flushleft]{threeparttable}
\usepackage[affil-it]{authblk}
\usepackage{comment}
\usepackage{float}
\usepackage[toc,page]{appendix}
\usepackage{multirow}
\DeclareMathOperator{\interior}{int}

\usepackage{longtable}
\usepackage{threeparttablex}
\usepackage[justification=centering]{caption}
\usepackage{subcaption}
\usepackage{enumerate}

\newtheorem{definition}{Definition}
\newtheorem{assumption}{Assumption}
\newtheorem{theorem}{Theorem}

\newtheorem{lemma}[theorem]{Lemma}
\newtheorem{proposition}{Proposition}
\newtheorem{example}{Example}

\DeclareMathOperator*{\argmax}{argmax}

\DeclareMathOperator{\diag}{diag} 
\DeclareMathOperator{\Normal}{N}
\DeclareMathOperator{\Uniform}{U}
\DeclareMathOperator{\var}{Var}

\newcommand{\conv}[1]{\overset{{\mathrm{#1}}}{\to}}
\newcommand{\asto}{\conv{as}}                     % almost sure
\newcommand{\dto}{\conv{d}}                       % distribution (d)

\newcommand{\E}{\mathrm{E}}

\def\npliter{l}
\def\Npliter{L}
\def\beliefspace{\Xi}

\defcitealias{Arcidiacono2016}{ABBE (2016)}

\title{Nested Pseudo Likelihood Estimation of Continuous-Time Dynamic Discrete Games}

\date{\vspace{-0.2cm}
January 4, 2023% \thanks{We thank seminar and conference participants at Ohio State, the 2019 Midwest Econometrics Group meeting, and the 2020 European Winter Meeting of the Econometric Society for useful comments and discussions.
}
\author{ 
Jason Blevins\thanks{Email: \texttt{blevins.141@osu.edu}}\\
    \vspace{-4mm} The Ohio State University
\and  Minhae Kim\thanks{Email: \texttt{minhae.kim@okstate.edu}} \\
    \vspace{-4mm} Oklahoma State University
}

\begin{document}

\maketitle

% \fbox{
%   \parbox{\textwidth}{
% \tableofcontents
%   }
% }

\begin{abstract} \onehalfspacing \noindent We introduce a sequential estimator for continuous time dynamic discrete choice models (single-agent models and games) by adapting the nested pseudo likelihood (NPL) estimator of \cite{AM02,AM07}, developed for discrete time models with discrete time data, to the continuous time case with data sampled either discretely (i.e., uniformly-spaced snapshot data) or continuously.  We establish conditions for consistency and asymptotic normality of the estimator, a local convergence condition, and, for single agent models, a zero Jacobian property assuring local convergence.  We carry out a series of Monte Carlo experiments using an entry-exit game with five heterogeneous firms to confirm the large-sample properties and demonstrate finite-sample bias reduction via iteration.  In our simulations we show that the convergence issues documented for the NPL estimator in discrete time models are less likely to affect comparable continuous-time models.  We also show that there can be large bias in economically-relevant parameters, such as the competitive effect and entry cost, from estimating a misspecified discrete time model when in fact the data generating process is a continuous time model. \\

\end{abstract}

\newpage 

\section{Introduction}

This paper introduces a new sequential estimator for continuous-time dynamic discrete choice models. In \emph{dynamic} models, as compared to static models, forward-looking agents make decisions based on expected future payoffs each period. \emph{Discrete choice} models focus on agents' discrete choices such as a firm's entry-exit decision or a worker's retirement. Dynamic discrete choice models have been rigorously studied in labor economics and empirical industrial organization since they were pioneered by \cite{Miller84}, \cite{wolpin}, \cite{pakes1986}, and \cite{rust}. They have been used subsequently in the areas of consumer inventory \citep{Song2003, Hendel2006}, dynamic demand of durable goods \citep{Schiraldi2011, Gowrisankaran2012}, and firm entry and exit \citep{benkard, holmes}.\footnote{\cite{dorapakes} and \cite{AM10} provide the survey of applied works in dynamic discrete choice models in discrete time.} These examples, however, are all based on discrete time models.

More recently, authors have worked with \emph{continuous time} dynamic discrete choice models due to their flexibility in separating the timing of actions in the model from the timing of data sampling and to their computational advantages in large-scale games.
In such models, the state variables change in a stochastic, sequential manner which helps avoid the ``curse of dimensionality'' arising from the need to compute players' expectations over future states \citep{Dora2012}.\footnote{In discrete time models all state variables change at the same time, so if there are $n$ state variables (i.e., one for each firm) and each state variable can change to one of $\kappa$ alternative values, the number of possible future states is $\kappa^n$. In contrast, in continuous time models only one state variable can change at any given instant, meaning that there are only $\kappa (n-1)$ possible future states for a player to consider.}
\citet*{Arcidiacono2016}---henceforth \citetalias{Arcidiacono2016}---showed that continuous time models remain empirically tractable, can be estimated with continuous or discrete time data, can better approximate the economic reality of certain industries, and that the computational advantages can allow researchers to solve and simulate counterfactuals in large-scale models.
\cite{Blevins16} presents theoretical, computational and econometric properties of continuous time dynamic discrete choice games and extends results from \citetalias{Arcidiacono2016}. Since then, CCP estimation in dynamic discrete choice models has been applied in various areas, including the transportation industry \citep{mazur, qin}, supermarket industry \citep{Schiraldi2012}, TV advertising \citep{deng}, nightlife \citep{cosman}, online games \citep{nevskaya}, and job search \citep{Arcidiacono2019}.

The goal of this paper is to introduce a nested pseudo likelihood (NPL) estimator for continuous time models, which we refer to as the CTNPL estimator.
\citetalias{Arcidiacono2016} developed a two-step estimator for continuous time discrete choice models, but here we extend this estimator using insights from the discrete time NPL estimator of \cite{AM02,AM07} to iteratively impose the equilibrium conditions from the game to improve estimates of the structural parameters.
\cite{AM02} originally developed the NPL algorithm and estimator in a single agent, discrete time setting setting and \cite{AM07} extended the algorithm and estimator to discrete time dynamic discrete choice games.
\cite{bugni} introduced the $K$-MD estimator, a sequential version of the minimum-distance estimator proposed by \cite{PS08}, based on the NPL mapping.

In this paper we introduce the CTNPL estimator, derive the large sample properties, and evaluate its performance in the setting of a familiar dynamic oligopoly model.
Compared to the two-step pseudo maximum likelihood (PML) estimator, the NPL estimator has a number of advantages.
First, it does not require consistent estimates for the initial conditional choice probabilities.
Second, the estimator performs better in finite sample as nonparametric estimates of choice probabilities (CCPs) can be imprecise, leading two-step estimators to be biased.
This paper shows that these advantages also carry over to the CTNPL estimator for continuous time models.
We also compare continuous time and discrete time models and (i) show that the CTNPL estimator has advantages in terms of convergence\footnote{Despite its advantages, \cite{PS10} showed that the NPL algorithm may fail to converge to the NPL estimator when the NPL mapping is unstable.
They provided a simple two-firm example model illustrating the issue and subsequent authors such as \cite{Egesdal} and \cite{Aguirregabiria2019} have studied the problem in more complex models. % other cases where the NPL algorithm fails to converge or converges to an inconsistent estimate 
\cite{KS12} provided a local convergence condition for the NPL estimator in discrete time models, however, \cite{Aguirregabiria2019} show that even when the population mapping is stable the algorithm can fail to converge when the sample counterpart mapping is unstable.
We derive a local convergence condition for continuous time models, which is motivated by that of \cite{KS12}, and in our simulation results the condition is more likely to be satisfied in continuous time.} and (ii) evaluate the potential for biased estimates and misleading economic conclusions when estimating a discrete time model using discrete time snapshot data that was actually generated by a continuous time model.

% Aguirregabiria and Mira (2002, 2007) studies the asymptotic properties for the NPL estimator. We show similar properties for the NPL estimator in continuous time. In \cite{AM07}, they also argue that the NPL estimator performs better in finite sample than two-step pseudo maximum likelihood (PML) estimator using Monte Carlo experiments. We develop a continuous time version of their five-player game in retail industry and present the finite sample performance of the NPL estimator in continuous time. 

% This paper is based on the literature on the NPL estimator and dynamic discrete choice models in continuous time. We examine the properties of the NPL estimator in continuous time and develop a continuous time version of the five-player game example in \cite{AM07}. \cite{bugni} introduce $K$-ML estimator which nests the NPL estimator as it becomes equivalent to the NPL estimator by letting the number of iterations $K \to \infty$. We rely on \cite{bugni} to present asymptotic properties for the policy iteration estimator where the number of iterations for the NPL algorithm is finite.

The rest of the paper is organized as follows. Section 2 develops a dynamic discrete game in continuous time with an example. Section 3 introduces the CTNPL algorithm to estimate the parameters in the model. In Section~\ref{sec:mc}, we simulate a five-player game in retail industry and present the performance of the CTNPL estimator. Section 5 concludes.

\section{Model}

In this section, we briefly review the continuous time dynamic discrete game framework of \citetalias{Arcidiacono2016}.
To provide a concrete example, we also introduce a continuous-time version of the five-player entry/exit model of \cite{AM07}.
This model serves not only as a running example but also as the basis for our Monte Carlo experiments in Section~\ref{sec:mc}.

\subsection{Setting}

There are $N$ agents in the model indexed by $i=1, \ldots, N$ who choose actions $j$ from a choice set $\mathcal{A}=\{0, 1, \ldots, J-1\}$ at decision times that occur in continuous time indexed by $t \in [0, \infty)$. Such decisions by agents result in endogenous state changes. The model also allows exogenous state variables. Decision times and exogenous state changes occur according to Poisson processes as detailed below.

\textbf{State space.} The states are discrete and finite, so every state at time $t$ can be represented by a state vector $x_k$ in a finite state space $\mathcal{X}=\{x_1, \ldots, x_K\}$ where each state is indexed by $k=1, 2, \ldots, K$.

\begin{example}
  Suppose that there are $N=5$ firms competing in a small local retail market where firms have at most one store.
  Based on incumbency status $a_i \in \{0,1\}$ each firm $i$ can either choose to enter/exit ($j=1$) or do nothing ($j=0$).
  Exit is not permanent: firms that exit have the opportunity to re-enter in the future.
  Market size is denoted by an exogenous state variable $s_k \in \{1,2,\ldots, 5\}$.
  Therefore, there are $K=5 \times 2^5 = 160$ distinct states in the model.
  Each state can be represented by a $6 \times 1$ vector $x_k$:
  \begin{equation*}
    x_k = (s_k, a_{1k}, a_{2k}, a_{3k}, a_{4k}, a_{5k}).
  \end{equation*}
\end{example}

\textbf{State changes.} There are two types of state changes in the model---endogenous state changes which are relying on each agent's decision and exogenous state changes driven by ``nature.'' As endogenous state changes only when an agent makes a choice, we can define a state continuation function $l(i,j,k)$ which maps an agent $i$'s choice $j$ and current state $k$ to next state $l$. We assume that choice $j=0$, where $l(i,0,k)=k$, is costless and that all choices $j$ are meaningfully distinct $l(i,j,k) \neq l(i,j',k)$ for all states $k=1, \ldots, K$.
On the other hand, although exogenous state changes are not the result of an action by a player per se, we can think of them as actions by nature.  When treated as a player in this way we will index nature by $i = 0$.

\addtocounter{example}{-1}
\begin{example}[continued] If firm 2 ($i=2$) decides to enter a market ($j=1$) while the market is in state $k$, then the continuation state index will be denoted $l(2,1,k)$.  The new state in vector form is $x_{l(2,1,k)} = (s_k, a_{1k}, 1, \ldots, a_{5k})$. To give an example of an exogenous state change, when market size increases by 1 in state $k$, the state moves from $k$ to $l$ which can be written in vector form as $x' = (s_k + 1, a_{1k}, a_{2k}, \ldots, a_{5k})$. Note that only one state variable moves at any single instant and other state variables remain the same as in the previous state. \qed
\end{example}

\textbf{Poisson and Markov jump processes.\footnote{Definitions and notations in this section are based on \cite{Schuette}.}}
% Consider a situation where the state of the model $k$ (and therefore some component of the state vector representation $x_k$) changes at time $T_n$ and the next change occurs at $T_{n+1}$. State changes can be due a firm's entry into the market or changes in other state variables such as an increase in market size. We call $T_n$ as the $n$-th event time and the difference between two absolute times as inter-event time denoted as $\tau_n$. Specifically, we consider a following process:
% \begin{equation*}
% T_{n+1} = T_n + \tau_n
% \end{equation*}
% where $n \geq 1$ and $T_0=0$.
%
% We assume that the decision times of firms and the times of exogenous state changes occur according to independent Poisson processes (possibly with different rate parameters).
% Under this assumption, the sequence of inter-event times $\{\tau_n\}_{n \in \mathbb{N}}$ is an independent and identically distributed sequence of exponential random variables with parameter $\lambda > 0$. Then, the number of events up to some time $t$, denoted as $N(t)$, follows Poisson distribution:
% \[P[N(t)=s] = \frac{(\lambda t)^s}{s!}e^{-\lambda t}.\]
%
Corresponding to the two types of state changes in the model there are also two types of Poisson processes governing the times of those state changes. First, in each state $k$ each player $i$ makes decisions about changing its endogenous state according to a Poisson processes with a rate parameter $\lambda_{ik}$. Recall that for such a Poisson process, the expected number of events during a one-unit time interval is $1/\lambda_{ik}$. For example, if the unit of time is one calendar year and we assume $\lambda_{ik}=\lambda=1$ for all $i$ and $k$, then each agent moves on average once per year. This would correspond closely to the usual discrete-time assumption of exactly one move per year. However, the Poisson assumption is more general since it allows, stochastically, for each player to make more or fewer moves per period. Because the probabilities of each choice vary, in general, across states, the state-to-state transition rates are endogenous and are state-and-firm specific.  Second, in each state $k$ the exogenous state variables change according to another state-specific Poisson process.
We assume that these Poisson processes are independent so that the probability that more than one component of the state vector $x_k$ changes at a single instant is zero. However, the processes are not identically distributed since they have different rate parameters.

\textbf{Intensity matrices.} Since Poisson processes are continuous-time Markov processes and the state space of the model is finite, the overall state transition process is a finite Markov jump process. Such a process can be characterized by a $K \times K$ intensity matrix (also known as a transition rate matrix or infinitesimal generator matrix) which contains the rates of all possible state-to-state transitions.  This matrix is the instantaneous counterpart of the one-step transition probability matrix in a discrete time model.

Each component $q_{kl}$ of the intensity matrix $Q = (q_{kl})$ represents the rate of leaving state $k$ and transitioning to state $l$. The difference is that the intensity rate gives the transition rate for an instant instead of over one period of time.  So, when we define transition probability matrix over some small time interval $h$ as $P(h)$, we can write $P(h) = I + Qh$.  Note that as $h \to 0$, $P(h)$ approaches the identity matrix (under which no transitions occur).
Then, for states $k$, $l$ with $k \neq l$, we have
\begin{equation*}
\Pr(x_{t+h}=l \mid x_t=k) =
\begin{cases} q_{kl}h & l \neq k, \\
1 - \sum_{l\neq k}q_{kl}h & l = k.
\end{cases}
\end{equation*}

We can formally define the intensity matrix $Q$:
\begin{equation*}
Q= \begin{bmatrix}
q_{11} & q_{12} & q_{13} & \cdots & q_{1K} \\
q_{21} & q_{22} & q_{23} & \cdots & q_{2K} \\
\vdots & \vdots & \vdots & \ddots & \vdots \\
q_{K1} & q_{K2} & q_{K3} & \cdots & q_{KK}
\end{bmatrix},
\end{equation*}
where
\begin{equation*}
q_{kl} =
\begin{cases} 
\lim_{h \to 0} \frac{\Pr(x_{t+h} = l \mid x_t = k)}{h} & \text{for} \ l \neq k,\\
- \sum_{l \neq k} q_{kl} & \text{otherwise}.
\end{cases}
\end{equation*}

For endogenous state changes, we can write $q_{kl}$ as the product of the decision rate in state $k$ and a choice probability. First, each agent decides whether or not to change the state according to Poisson process at rate $\lambda_{ik}$.  To determine the choice probability, we must look closer at the player's problem. Players are forward-looking and discount future payoffs at rate $\rho \in (0, \infty)$.  They observe other players' moves and optimally choose their own action $j\in \mathcal{A}$ based on their beliefs of other agents' actions, including state changes by nature.  We denote the conditional choice probability of agent $i$ choosing action $j$ in state $k$ as $\sigma_{ijk}$. So, if at some state $k$ a player makes decisions at rate $\lambda_{ik}$ and chooses action $j$ by probability $\sigma_{ijk}$, then the rate of state changes due to player $i$ taking action $j$ in state $k$ is $\lambda_{ik} \sigma_{ijk}$.

Under the assumptions of the model so far the overall intensity matrix $Q$ can be written as a sum of individual player-specific intensity matrices: $Q = Q_0 + Q_1 + \dots + Q_N$, where $Q_0 = (q_{kl})$ denotes the intensity matrix for exogenous state changes by nature and the elements of $Q_0$ are denoted $q_{kl}$ (i.e., we drop the $i = 0$ subscript for nature).
For players $i= 1, \dots, N$, each element of $Q_i$, denoted $q_{ikl}$, can be written as follows:
\begin{equation*}
q_{ikl} =
\begin{cases}
\lambda_{ik} \sigma_{ijk} & \text{if} \ l \neq k \ \text{and} \ l(i,j,k)=l, \\
- \sum_{l \neq k} q_{ikl} & l(i,j,k) = k, \\
0 & \text{otherwise}.
\end{cases}
\end{equation*}
The first case represents state transitions out of the current state.
The rate of a transition from $k$ to $l$ is the product of the decision rate in $k$ and the probability of choosing action $j$ which results a transition to state $l$.
The second case is the net rate of leaving state $k$, and therefore is the negative of the sum of the off-diagonal elements in row $k$.
Because firms may find it optimal to remain in the same state (choosing $j = 0$) we must allow for self-transitions, so the total rate of leaving state $k$ due to action by player $i$ need not equal the rate of decisions $\lambda_{ik}$.
The third case covers states $l$ for which transitions from $k$ to $l$ cannot occur in a single instant.

\addtocounter{example}{-1}
\begin{example}[continued] In the previous example, nature changes the market size according to Markov jump process characterized by intensity matrix $Q_0$. We impose a natural restriction that the market size can only change by 1 level at an instant.
  Firms make decisions at rate $\lambda$ and will then choose actions according to the choice probabilities $\sigma_{ijk}$.
  So, $q_{ikl}$ for firm $i$'s intensity matrix can be written as:
  \begin{equation*}
    q_{ikl} = \begin{cases} \lambda \sigma_{i1k} & \text{if firm $i$ enters or exits and $l(i,1,k)=l$}, \\
      - \lambda \sigma_{i1k} & \text{if firm $i$ does nothing}, \\
      0 & \text{otherwise}.\end{cases}
  \end{equation*}
  \qed
\end{example}

\textbf{Payoffs.} In the continuous time setting we distinguish agents' flow payoffs from instantaneous choice-specific payoffs. First, agents receive flow payoffs at each state $k$ regardless of their choices. We denote the flow payoffs by $u_{ik}$ and assume $|u_{ik}|< \infty$ for $i=1, \ldots, N$ and $k=1, \ldots, K$. On the other hand, at a decision time when player $i$ chooses action $j$ in state $k$, the player receives a lump-sum instantaneous choice-specific payoff which is the sum of a deterministic component $\psi_{ijk}$, where $|\psi_{ijk}|<\infty$, and a stochastic component $\varepsilon_{ijk} \in \mathbb{R}$. The deterministic component $\psi_{ijk}$ is common knowledge to other players but the second part $\varepsilon_{ijk}$ is only observed by player $i$. However, players have common knowledge of the distribution of the stochastic components.

\addtocounter{example}{-1}
\begin{example}[continued] We consider a log-linear profit function for the flow payoffs including competition effects and market size. 
\begin{equation}\label{eq:profit}u_{ik} =\theta_{\text{RS}}\ln(s_k) - \theta_{\text{RN}}\ln\big(1+\sum_{m\neq i}a_{mk}\big)- \theta_{FC,i} \end{equation}
We estimate the fixed costs $\theta_{FC,i}$ for $i=1, \ldots, 5$ and other profit function parameters $\theta_{\text{RS}}$ and $\theta_{\text{RN}}$.
When a firm decides an action $j$, they receive an instantaneous payoff $\psi_{ijk}+\varepsilon_{ijk}$ where
\begin{equation}
\label{eq:switching}
\psi_{ijk} = \begin{cases} -\theta_{\text{EC}} & \text{if} \ j = 1 \  \text{and} \  a_{ik} = 0,
\\ 0 & \text{otherwise}.
\end{cases}
\end{equation}
Note that there is a distinction between payoffs received while in a particular state and payoffs associated with actions.
For example, the entry cost is paid as a lump sum in \eqref{eq:switching} instead of appearing in the flow profit equation \eqref{eq:profit}.
By comparison, the entry cost parameter appears together with product market profits in the aggregated discrete time period profit equation given in (48) in \cite{AM07}.
\qed
\end{example}

\subsection{Bellman Optimality and CCP Representation}

In a Markovian game, every agent plays a stationary Markov strategy based on their beliefs regarding other agents' choice probabilities. Let $\varsigma_i$ denote the beliefs held by player $i$ about each of the $(N-1)$ other players playing each of their $J$ choices in each of $K$ states. Therefore, $\varsigma_i$ is a collection of $(N-1) \times J \times K$ probabilities, denoted $\varsigma_{imjk}$ for each player $m \neq i$, choice $j$, and state $k$.

We now introduce the value functions and CCPs. Agents establish the value function based on their expectations about endogenous state changes due to other agents' moves, exogenous state changes due to nature, and their own future decisions.  We follow the derivation of instantaneous Bellman equation in \cite{Blevins16}.  The probability that agent $i$ makes a decision during a small increment of time of length $h$ while in state $k$ is $\lambda_{ik} h$ under the Poisson assumption.  Denoting each agent $i$'s discount rate as $\rho_i$, the discount factor for the time increment $h$ is $1/(1+\rho_i h)$. Given agent $i$'s beliefs, $\varsigma_i$, we can express the Bellman equation as below:
\begin{multline*}
  V_{ik} = \frac{1}{1+\rho_i h} \left[u_{ik} h + \sum_{l \neq k} q_{kl} h V_{il} 
  + \sum_{m \neq i} \lambda_{mk}h\sum_{j=0}^J \varsigma_{imjk}V_{i,l(m,j,k)} \right. \\ \left.
  + \lambda_{ik} h \E \left( \max_j \{\psi_{ijk} + \varepsilon_{ijk}+V_{l(i,j,k)}\} \right)
  + \left(1-\sum_{i=1}^N\lambda_{ik} h - \sum_{l \neq k} q_{0kl} h\right) V_{ik} + o(h)\right].
\end{multline*}
We now turn to each of the terms inside square brackets in the previous equation.
First, over the small amount of time $h$ that the model remains in state $k$ player $i$ receives an integrated flow payoff of $u_{ik}h$.
Second, if the state changes from $k$ to $l$ due to nature, which occurs with probability proportional $q_{kl}$ over the interval $h$, then player $i$ receives the continuation value $V_{il}$.
Third, each of player $i$'s rival players, indexed by $m \neq i$, may make a decision with probability proportional to $\lambda_{mk}$ over the time interval $h$.
When this occurs, player $i$ believes that player $m$ plays action $j$ in state $k$ with probability $\varsigma_{imjk}$.
When this happens, player $i$ receives a continuation value $V_{i, l(m,j,k)}$ where $l(m,j,k)$ denotes the index of the continuation state following the action of player $m$.
Fourth, at rate $\lambda_{ik}$ player $i$ makes a decision and optimally chooses an action $j$. However, the values of the stochastic choice-specific shocks have not yet been realized, so player $i$ receives receives the expected maximum value of $\psi_{ijk}+\varepsilon_{ijk} + V_{l(i,j,k)}$ where $j$ is chosen optimally at the future decision time.
Finally, the last (non-negligible) term accounts for the case where no state changes occur over the time interval $h$ and so the continuation value is simply $V_{ik}$.

By rearranging and letting $h \to 0$, we arrive at a simpler expression for the Bellman equation:
\begin{equation*}
V_{ik} = \frac{u_{ik}+\sum_{l \neq k}q_{kl} V_{il} + \sum_{m \neq i} \lambda_{mk} \sum_j \varsigma_{imjk}V_{i, l(m,j,k)} + \lambda_{ik} E\max_j\{\psi_{ijk} + \epsilon_{ijk} + V_{i, l(i,j,k)}\}}{\rho + \sum_{l \neq k} q_{kl} + \sum_{i=1}^N\lambda_{ik}}.
\end{equation*}
Now we define the best response function for each agent.
A Markov strategy for agent $i$ is a mapping $\delta_{i}$ from each state $(k, \epsilon_{ik}) \in \mathcal{X} \times \mathbb{R}^J$ to an action $j \in \mathcal{A}$. Then, the best response for an agent $i$ is defined as:
\begin{equation*}
\delta_i(k, \epsilon_{ik}; \varsigma_i) = j \Longleftrightarrow \psi_{ijk} + \epsilon_{ijk} + V_{i,l(i,j,k)} \geq \psi_{ij'k} + \epsilon_{ij'k} + V_{i, l(i,j',k)} \quad \forall j' \in \mathcal{A}.
\end{equation*}

Then, the choice probability for each agent becomes:
\begin{equation*}
\sigma_{ijk} = \Pr[\delta_i(k, \epsilon_{ik}; \varsigma_i) =j|k].
\end{equation*}

We now collect assumptions needed for the existence of an equilibrium. 

\begin{assumption}[Discrete states]
\label{a:state}
The state space is finite: $K=|\mathcal{X}|<\infty$.
\end{assumption}

\begin{assumption}[Bounded rates and payoffs]
\label{a:bounded} The discount rate $\rho_i$, move arrival rate, rates of state changes due to nature, and payoffs are all bounded for all $i=1, \ldots, N$, $j=0, \ldots, J-1$, $k=1, \ldots, K$, $l=1, \ldots, K$ with $l \neq k$: (a) $0<\rho_i<\infty$, (b) $0<\lambda_{ik}<\infty$, (c) $0 \leq q_{kl}<\infty$, (d) $|u_{ik}|<\infty$, and (e) $|\psi_{ijk}|<\infty$. 
\end{assumption}

Assumptions \ref{a:state} and \ref{a:bounded} are needed to define a continuous-time finite Markov jump process.
In Assumption \ref{a:bounded}, we rule out infinite rates and payoffs to make it possible to estimate by restricting the value to a finite number.

\begin{assumption}[Additive separability]
\label{a:additive}
In each state $k$, the instantaneous payoff associated with choice $j$ is additively separable as $\psi_{ijk}+\varepsilon_{ijk}$.
\end{assumption}

\begin{assumption}[Distinct actions]
\label{a:distinct}
For all $i=1, \ldots, N$ and $k=1, \ldots, K$, the continuation state function $l(i,j,k)$ and the choice-specific payoffs $\psi_{ijk}$ satisfy the following two properties: (a) choice $j=0$ is a costless continuation choice with $l(i,j,k)=k$ and $\psi_{ijk}=0$, and (b) all choices $j$ are meaningfully distinct in the sense that the continuation states differ: $l(i,j,k) \neq l(i,j',k)$ for all $j=0, \ldots, J-1$ and $j'\neq j$. 
\end{assumption}

\begin{assumption}[Private information]
\label{a:private}
The errors $\varepsilon_{ik}$ are i.i.d. over time with joint distribution $F$ which is absolutely continuous with respect to Lebesgue measure (with joint density $f$), has finite first moments, and has support equal to $\mathbb{R}^J$. 
\end{assumption}

Assumption \ref{a:additive} is common in the literature, as in the case of Assumption 1 of \cite{AM02} and \cite{AM07}. Assumption \ref{a:distinct} requires that different choices should result in observationally distinct state changes. For example, choices $j=1$ and $j=2$ should not be defined separately if they in fact always result in the same state changes. It also involves a normalization on the baseline payoff/cost associated with inaction. Finally, the distributional assumption on $\varepsilon$ is widely in the previous literature. For example, see Assumption 2 of \cite{AM02} and \cite{AM07}.

We define Markov perfect equilibrium (MPE) as in \cite{AM07} and \citetalias{Arcidiacono2016}.
% FIXME: We need to add something about beliefs here.
\begin{definition}[Markov perfect equilibrium] A collection of Markov strategies $\{\delta_1, \ldots, \delta_N\}$ and beliefs $\{\varsigma_1, \ldots, \varsigma_N\}$ is a Markov perfect equilibrium if for all $i$:
\begin{enumerate}
    \item $\delta_i(k, \epsilon_i)$ is a best response given beliefs $\varsigma_i$, for all $k$ and almost every $\epsilon_i$;
    \item the beliefs $\varsigma_{im}$ for all players $m\neq i$ are consistent with the best response probabilities implied by the best response mapping $\delta_m$.
\end{enumerate}
\end{definition}

By Proposition 5 in \citetalias{Arcidiacono2016}, under Assumptions~1--5, a Markov perfect equilibrium exists. Moreover, under Assumptions~1--5 and assuming that $\lambda_{ik}$ is constant across all agents and states denoted by $\lambda$, \citetalias{Arcidiacono2016} established the following linear representation of the value function, where $V_i(\theta,\sigma)$ is a $K$-vector:
\begin{equation}
\label{eq:value}
V_i(\theta, \sigma)=\Big[(\rho+N\lambda)I-\lambda \sum_{m=1}^N \Sigma_m(\sigma_m) -Q_0\Big]^{-1}[u_i(\theta)+\lambda_i E_i(\theta, \sigma)],
\end{equation}
where $\Sigma_m(\sigma_m)$ is the $K\times K$ state transition matrix induced by the actions of player $m$ given the choice probabilities $\sigma_m$ and where $E_i(\theta, \sigma)$ is a $K \times 1$ vector where each element $k$ is the ex-ante expected value of the choice-specific payoff in state $k$, $\sum_j \sigma_{ijk}[\psi_{ijk}+e_{ijk}(\theta, \sigma)]$ where $e_{ijk}(\theta, \sigma)$ is the expected value of $\varepsilon_{ijk}$ given that choice $j$ is optimal, 
\begin{equation*}
\frac{1}{\sigma_{ijk}}\int \varepsilon_{ijk} \cdot 1\{\varepsilon_{ij'k}-\varepsilon_{ijk} \leq \psi_{ijk} - \psi_{ij'k}+V_{l(i,j,k)}(\theta, \sigma)-V_{l(i,j',k)}(\theta, \sigma) \ \forall j' \} f(\varepsilon_{ik}) d \varepsilon_{ik}.
\end{equation*}

We now have two important equations that forms a players' decision problem. 

\begin{enumerate}
\item Bellman optimality: We stack \eqref{eq:value} across players $i = 1, \dots, N$ and define the policy valuation operator $\Upsilon(\theta, \sigma)$:
\begin{equation}
\Upsilon (\theta, \sigma) \equiv V(\theta, \sigma)=\Big[(\rho+N\lambda)I-\lambda \sum_{m=1}^N \Sigma_m(\theta, \sigma_m) -Q_0\Big]^{-1}[u+\lambda E(\theta, \sigma)].
\end{equation}

\item Conditional choice probability
\begin{equation*}
  \Gamma(v) \equiv \sigma
\end{equation*}
where $\sigma$ is a $NJK \times 1$ vector with $\sigma_{ijk}=\Pr[\delta_i(k, \varepsilon)=j|k]$. 
Assuming $\varepsilon$ follows i.i.d. $TIEV(0,1)$, 
\begin{equation}\label{eq:ccpexp}
\sigma_{ijk} = \frac{\exp(\psi_{ijk}+V_{l(i,j,k)})}{\sum_{j'}\exp(\psi_{ij'k}+V_{l(i,j',k)})}
\end{equation}

\end{enumerate}

We now introduce a policy iteration operator $\Psi$ in the space of conditional choice probabilities.  Substituting (1) into (2),  
\begin{equation}
\label{eq:fixed}
\sigma = \Psi(\theta, \sigma) \equiv \Gamma(\Upsilon(\theta, \sigma))
\end{equation}
where $\Upsilon$ is a policy valuation operator that maps an $NJK\times 1$ vector of conditional choice probabilities into an $NK\times 1$ vector in value function space. $\Gamma$ is a policy improvement operator that maps an $NK \times 1$ vector in value function space into $NJK\times 1$ vector of conditional choice probabilities. 

\subsection{Sequential vs. Simultaneous State Changes}
\label{sec:sequential}

The main fundamental difference between the continuous-time model we consider and traditional discrete-time models is not in the way time is measured or how payoffs are discounted and accrued, but rather that state changes occur in a stochastic, sequential manner.
This leads to both a reduction in computational complexity and a dampening in strategic interaction.
This does not, however, mean that in general there is unique equilibrium or that all fixed points are stable, but rather it can lead to improvements in convergence.

In terms of computation, in the continuous time model there is a reduction in the number of future possible states to consider because at any point in time the state can only be affected by a single agent's choice or by nature.
Therefore, the number of future states is linear in the number of players.
On the other hand, in discrete time all agents and nature move simultaneously and the number of possible future states increases exponentially in the number of players.
This leads to a relatively dense transition matrix over a discrete unit of time as opposed to a relatively sparse intensity matrix in continuous time.

To give a simple example, suppose there are $N$ players and each player can choose from $J$ actions.
In the continuous-time model, there are only $N \times J$ possible future states at a given instant (i.e., which player moves next and which choice do they make).
In discrete-time with simultaneous moves, there are $J^N$ possible future states to consider.
The difference will increase dramatically as either $N$ or $J$ get larger.
So with $N = 5$ players and $J=2$ as in our application, the difference is already three-fold: $5 \times 2 = 10$ continuation states in continuous time versus $2^5 = 32$ in discrete time.

The two main differences with our model and a typical discrete time model are the way payoffs are accrued as flows over time and the sequential nature of moves.
The latter, however, is the key distinction.
To show this, we briefly consider a simultaneous move model in continuous time and show that it is equivalent to a transformation of a discrete time model.
The Bellman equation in this model can be expressed as follows:
\begin{equation*}
V_{ik} = \mathrm{E}\left[ \int_0^\tau e^{-\rho t} u_{ik}\,dt + e^{-\rho \tau} \max_{j_i} \left\{ \psi_{i j_i k} + \varepsilon_{i j_i k} + \sum_{j_{-i}} \varsigma_{ik}(j_{-i}) V_{i,l(j_1,\cdots, j_N, k)} \right\} \middle\vert k \right]
\end{equation*}
where
$\tau$ is the random holding time until the next simultaneous move arrival,
$j_i$ is firm $i$'s choice when the next move occurs,
$\varsigma_{ik}(j_{-i})$ denotes player $i$'s belief that rival players will choose actions
$j_{-i} = (j_1, \dots, j_{i-1}, j_{i+1}, \dots, j_N)$ resulting in the continuation state
$k' = l(j_{1},\dots,j_{N}, k)$.
This expression is from the perspective of an arbitrary point in time which is not necessarily a move arrival time.
The first term inside the expectation is the flow utility accrued until the random next move increment $\tau$ and the second term is the discounted expected continuation value received when choosing the next action, given beliefs about rival actions. As in discrete time (and contrary to our continuous time model) there are no state changes between decision times.

To see that this model can be thought of as transformation of a standard discrete time model, define a transformed discount factor
$\beta = \mathrm{E}\left[ e^{-\rho \tau} \right]$,
where we assume the distribution of the holding time $\tau$ is independent of $k$,
and a transformed choice-specific utility function
$\tilde{u}_{ijk}(\varepsilon_{i}) = \psi_{ijk} + \mathrm{E}\left[ \int_0^\tau e^{-\rho t} u_{ik'}\,dt \mid k, j_i = j \right]$.
Then, we can restate the Bellman from the perspective of a decision time, as is usually done in discrete time.
Given choice-specific errors $\varepsilon_{ik}$ for player $i$, which are observed at a decision time, we have:
\begin{align*}
  W_{ik}(\varepsilon_{ik})
  &= \max_j \left\lbrace \psi_{ijk} + \varepsilon_{ijk} + \mathrm{E}\left[ V_{ik'} \mid k, j_i = j \right] \right\rbrace \\
  &= \max_j \left\lbrace \tilde{u}_{ijk} + \varepsilon_{ijk} + \beta \mathrm{E}\left[ W_{ik'}(\varepsilon_{ik'}) \mid k, j_i = j \right] \right\rbrace.
\end{align*}
The representation on the last line is in the form of a standard dynamic discrete choice game in discrete time.
Therefore, the key difference between the continuous-time model discussed in the paper and a standard discrete-time model is the stochastic, sequential nature of moves rather than the rather than the way time is measured, or that payoffs are accrued as flows.

One could also specify a sequential-move model in discrete time, but in that case one would need to also specify the order of moves, the number of moves by each player during a period, and a fixed time period between moves.
Some advantages of the current continuous time specification are the mathematical tractability in terms of a Markov jump process, the fact that the order of moves is stochastic and fully flexible (e.g., allowing multiple actions by a single player per ``period''), and that the times between actions are stochastic rather than fixed.

\section{The CTNPL Estimator}

\subsection{CTNPL Algorithm}

% FIXME: Use \cite commands for the references at the end...
Before describing the CTNPL algorithm, we first introduce two more assumptions for identification and estimation with continuous time data observed only at fixed time intervals (Assumptions 6, 7 in ABBE, 2016; Assumptions 8, 9 in Blevins, 2016).

\begin{assumption}
The mapping $Q \to \{Q_0, Q_1, \ldots, Q_N\}$ is known. 
\end{assumption}

\begin{assumption}
  (i) In every market $m=1, \ldots, M$, every agent expects a single equilibrium to be played which results an intensity matrix $Q$.
  (ii) The distribution of state transition at any time $t \in [0, \infty)$ is consistent with the intensity matrix $P(\Delta) = \exp(\Delta Q)$. 
\end{assumption}

For continuous time models, we consider two types of data. The first type is continuous time data which is observed every instant. Stock prices can be an example close to this type. Second, some data generated by a continuous time model can be only observed at discrete times. For example, population can change any time during the year but annual population data may only be observed once per year. We show how to construct log likelihood functions for estimation using both types of data.
We assume throughout that we have independent and identically distributed data across markets and that the Markov process is irreducible.

\textbf{Fully observed continuous time data.}
Consider a dataset of observations in market $m$ and $n$-th move, $\{k_{mn}, t_{mn}: m=1, \ldots, M, n= 1, \ldots, T_m\}$ sampled in interval $[0, \bar{T}]$. State $k_{mn}$ denotes the state immediately before state change at time $t_{mn}$ and time $t_{mn}$ is the time of $n$-th state change in market $m$.

We follow the derivation in Section 6.1. in \citetalias{Arcidiacono2016} to construct the likelihood function. Let $h$ denote a $K(K-1)+NJK$ vector of hazard rates for state change:
\begin{equation*}
h = (q_{12}, q_{13}, \ldots, q_{K-1, K}, \lambda_{11}\sigma_{111}, \ldots, \lambda_{11}\sigma_{1Jk}, \ldots, \lambda_{N1}\sigma_{N11}, \ldots, \lambda_{NK}\sigma_{NJK}).
\end{equation*}

The probability that there is a state change within $\tau$ units of time is:
\begin{equation*}
\left(\sum_{l \neq k} q_{kl} + \sum_i \lambda_{ik}\sum_{j \neq 0}\sigma_{ijk}\right)\exp\left[-\tau\left(\sum_{l \neq k} q_{kl} + \sum_i \lambda_{ik}\sum_{j \neq 0}\sigma_{ijk}\right)\right].
\end{equation*}
This is because state changes due to nature occur in state $k$ according to Poisson processes with rate parameters $q_{kl}$, for each state $l \neq k$, and each player $i$ takes some action according to a Poisson process with rate parameter $\lambda_{ik}$.

When there is a state change from state $k$, conditional probability that the change is due to agent $i$'s action $j$ is
\begin{equation*}
\frac{\lambda_{ik}\sigma_{ijk}}{\sum_{l \neq k}q_{kl} + \sum_i \lambda_{ik}\sum_{j \neq 0}\sigma_{ijk}}
\end{equation*}
and conditional probability that it is due to nature is 
\begin{equation*}
\frac{q_{kl}}{\sum_{l \neq k}q_{kl} + \sum_i \lambda_{ik}\sum_{j \neq 0}\sigma_{ijk}}.
\end{equation*}

Defining $g(\tau, k; h)=\exp\left[-\tau(\sum_{l \neq k} q_{kl} + \sum_i \lambda_{ik}\sum_{j \neq 0}\sigma_{ijk})\right]$, the likelihood of state change is
\begin{equation}
\label{eq:likelihoodchange}
\begin{cases} \lambda_{ik}\sigma_{ijk}g(\tau, k; h) & \text{if agent $i$ chooses action $j$}, \\
q_{kl}g(\tau, k; h) & \text{if nature moves}.
\end{cases}
\end{equation}

Let $\tau_{m,T+1}$ be the time interval between the last state change at time $T$ and $\bar{T}$, and $k_{m,T+1}\equiv k_{m,T}$. We introduce an indicator function $I_{mn}(i,j)$ which is 1 when $n$-th state change at market $m$ is induced by agent $i$'s action $j$ and $0$ otherwise. We also define $I_{mn}(0,l)$ as 1 when $n$-th state change at market $m$ is by nature's move to state $l$ and $0$ otherwise.

Combining the results above, the likelihood function becomes
\begin{multline}
    L(\theta, h) = \frac{1}{M}\sum_{m=1}^M \left[\sum_{n=1}^T \left\{\ln g(\tau, k; h) + \sum_{l \neq k_{mn}}I_{mn}(0,l)\ln q_{k_{mn},l} + \sum_i\sum_{j \neq 0}I_{mn}(i,j)\ln (\lambda_{imk}\sigma_{ijk}) \right\} \right. \\ \left.
    + \ln g(\tau_{m,T+1}, k_{m,T+1}; h)\right].
\end{multline}

The first term inside the summation is the common term in \eqref{eq:likelihoodchange} and the second term is from the likelihood of state changes due to nature. The third term is considering agents' moves and the last term is the likelihood of observing no state change after time $t_{m,T}$.

Since state changes due to nature are exogenous, the log-likelihood function of the model can be decomposed into terms involving the conditional choice probabilities of players and nature's transition probabilities. Consistent estimates for elements of $Q_0$, $q=(q_{12}, \ldots, q_{K-1,K})$, can be obtained from transition data without having to solve the dynamic game. Therefore, we assume in the remainder of the paper that $q$ is known and focus on the estimation of $\theta$ and $\sigma=(\sigma_{111}, \ldots, \sigma_{NJK})$. Then, the likelihood function becomes
\begin{equation*}
L_M(\theta, \sigma) = \frac{1}{M}\sum_{m=1}^M \left[\sum_{n=1}^T \left\{\ln g(\tau, k; \sigma) + \sum_i\sum_{j \neq 0}I_{mn}(i,j)\ln (\lambda_{imk}\sigma_{ijk}) \right\} + \ln g(\tau_{m,T+1}, k_{m,T+1}; \sigma)\right].
\end{equation*}

\textbf{Discretely observed data.} Consider a dataset $\{k_{mn}: m=1, \ldots, M, n=0, \ldots, T\}$ consisting of observations that are sampled at times on the lattice $\{n\Delta:n=0, \ldots, T\}$. As before, we assume that $q$ is known. Given $q$ and $\sigma$, we can construct a transition matrix. Define the pseudo likelihood function 
\begin{equation}\label{eq:likelihooddis} L_M(\theta, \sigma)=\frac{1}{M}\sum_{m=1}^M \sum_{n=1}^T \ln P_{k_{m,n-1},k_{mn}}(\Delta; \Psi(\theta, \sigma))\end{equation}
where $P_{k,l}(\Delta;\Psi(\theta, \sigma))$ denotes the $(k,l)$ element of the transition matrix induced by $\sigma=\Psi(\theta, \sigma)$. 

We calculate $P(\Delta)$ using matrix exponential as in \citetalias{Arcidiacono2016}. First, assume that we have a continuous time data observed only at discrete times with unit interval, $\Delta=1$. Let $Z(\theta, \sigma)^r$ denote the transition matrix from some state $k$ to another state $k'$ after an unit interval of time $\Delta=1$ in $r$ steps. Then, let $a_{m,n}$ denote an indicator vector for an observation $n$ in market $m$ which is one in position $k_{m,n}$ and zero elsewhere. We can write the likelihood function as below:
\begin{equation}
  L_M(\theta, \sigma)=\frac{1}{M}\sum_{m=1}^M \sum_{n=1}^T \ln\left[ \sum_{r=0}^\infty \frac{\lambda^r \exp(-\lambda)}{r!}a_{m,n}' Z(\theta, \sigma)^r a_{m,n+1}\right].
\end{equation}
Then, using the definition of matrix exponential, we can rewrite the previous equation as:
\begin{equation}
L_M(\theta, \sigma)=\frac{1}{M}\sum_{m=1}^M \sum_{n=1}^T \ln \left[a_{m,n}' \exp\left(Q(\theta, \sigma)\right) a_{m,n+1}\right],
\end{equation}
where $Q$ is the aggregate intensity matrix for the model.

\textbf{CTNPL algorithm.} We adapt the NPL algorithm to obtain the CTNPL estimator in the present continuous time model.
Let $\hat{\sigma}^0$ be an initial estimate or guess of the vector of players' choice probabilities. Given $\hat{\sigma}^0$, for $\npliter \geq 1$, perform the following steps:
\begin{enumerate}
\item Given $\hat{\sigma}^{\npliter-1}$, update $\hat{\theta}$ by 
\begin{equation}\label{eq:npl1} \hat{\theta}^{\npliter} = \argmax_{\theta \in \Theta} L_M(\theta, \hat{\sigma}^{\npliter-1})\end{equation}
\item Update $\hat{\sigma}^\npliter$ using the equilibrium condition, i.e.
\begin{equation}\label{eq:npl2}\hat{\sigma}^{\npliter} = \Psi(\hat{\theta}^{\npliter}, \hat{\sigma}^{\npliter-1}) \end{equation}
\end{enumerate}
Iterate in $\npliter$ until convergence in $\sigma$ and $\theta$ is reached.
Since we usually have observe continuous time data at discrete times in practice, we focus on properties of the CTNPL estimator using the likelihood function in \eqref{eq:likelihooddis} in next section.

\textbf{Remark on the definition of the NPL estimator.} \cite{AM02} define $\hat{\theta}^\npliter$ as policy iteration (PI) estimator and the estimation procedure as the NPL algorithm.  In \cite{AM07}, they define the same estimator as $K$-stage estimator\footnote{They call it $K$-stage estimator as they use $K$ as a notation for the number of stages. They also define ($\hat{\theta}^1$, $\hat{\sigma}^1)$ as the Pseudo Maximum Likelihood (PML) estimator.} and define the NPL estimator as the limit of the sequence $\{\hat{\theta}^\npliter, \hat{\sigma}^\npliter\}_{\npliter=1}^\Npliter$ with certain conditions. Specifically, if the sequence $\{\hat{\theta}^\npliter, \hat{\sigma}^\npliter\}_{\npliter=1}^\Npliter$ converges, its limit $(\tilde{\theta}, \tilde{\sigma})$, which is called NPL fixed point, satisfies two conditions:
\[\tilde{\theta} = \argmax_{\theta \in \Theta} L(\theta, \tilde{\sigma}) \quad \text{and} \quad \tilde{\sigma}=\Psi(\tilde{\theta}, \tilde{\sigma}).\]
It is possible that there exist multiple NPL fixed points, so they choose the one with highest likelihood and define it as the NPL estimator, $(\hat{\theta}_{\text{NPL}}, \hat{\sigma}_{\text{NPL}})$. Subsequent studies use the definition of NPL estimator in \cite{AM07}, so we follow the same definition \citep{PS10, KS12}. Henceforth, we define $(\hat{\theta}^\npliter, \hat{\sigma}^\npliter)$ as the PI estimator and $(\hat{\theta}, \hat{\sigma})=(\hat{\theta}_{\text{CTNPL}}, \hat{\sigma}_{\text{CTNPL}})$ as the CTNPL estimator, dropping the subscript for simplicity. 

\subsection{Large Sample Properties}

In this section, we consider the large sample properties of the CTNPL estimator. We start with the CTNPL estimator in the case of single agent continuous time models, then we consider both the PI estimator (for a finite number of iterations $\npliter$) and the CTNPL estimator continuous time dynamic games.

\subsubsection{Single Agent Dynamic Discrete Choice Models}

The asymptotic properties of the CTNPL estimator for single agent models depend on a zero Jacobian property for $\Psi(\theta, \sigma)$.  Specifically, we show that the Jacobian matrix
$\Psi_\theta \equiv \nabla_{\theta}\Psi(\theta^*, \sigma^*)$ is equal to the zero matrix in agent models.  This result is presented formally in Proposition~\ref{prop:zero} immediately below, and we then show that it facilitates adaptive two-step estimation in Proposition~\ref{prop:single}, where we show the estimator is efficient.  The same zero Jacobian property will be used later, in Section~\ref{sec:convergence}, where we show that it also ensures local stability of the estimator.
All proofs are provided in Appendix~\ref{sec:proofs}.

\begin{proposition}\label{prop:zero}
In a single agent model, the Jacobian matrix $\Psi_\sigma$ is zero at the fixed point $(\theta^*,\sigma^*)$.
\end{proposition}

Next, Proposition~\ref{prop:single} shows that, under certain regularity conditions, the NPL estimator in a continuous time single agent discrete choice model is consistent, asymptotically normal, and efficient.
\begin{proposition}
\label{prop:single}
Let $\Theta$ denote the parameter space, let $\sigma^*$ denote the true CCPs, and define $\beliefspace \equiv \underbrace{\Delta^{J} \times \dots \times \Delta^{J}}_{K}$ to be the set of possible values of $\sigma$, where $\Delta^J$ denotes the standard $J$-simplex. Consider the following regularity conditions:

\begin{enumerate}[(a)]
\item $\Theta$ is a compact set.

\item $\Psi(\theta, \sigma)$ is continuous and twice continuously differentiable in $\theta$ and $\sigma$. 

\item Let $\Psi_{jk}(\theta, \sigma)$ be the corresponding element of $\Psi(\theta, \sigma)$ for some choice $j$ and state $k$. Then, $\Psi_{jk}(\theta, \sigma)>0$ for any $j=0, 1, \ldots, J-1$, $k = 1, \ldots, K$, and any $(\theta, \sigma) \in \Theta \times \beliefspace$.

\item There is a unique $\theta^* \in \interior(\Theta)$ such that $\sigma^* = \Psi(\theta^*, \sigma^*)$.

\item  $\hat{\sigma}^0$ is a strongly consistent estimator of $\sigma^*$ such that 
\begin{equation*}
 [ \sqrt{M} \nabla_\theta L_M(\theta^*, \sigma^*); \sqrt{M}(\hat{\sigma}^0-\sigma^*)'\Big]'\xrightarrow{\text{d}} \Normal(0, \Omega)
\end{equation*}
\end{enumerate}
Then, 
\begin{equation}
\label{eq:singlenormal}
\sqrt{M}(\hat{\theta}^{\npliter}-\theta^*)\xrightarrow{\text{d}}\Normal(0,\Omega_{\theta \theta'}^{-1})
\end{equation}
where $\Omega_{\theta \theta'}=\E[\nabla_\theta s_m \nabla_{\theta'}s_m]$ and $s_m$ is defined as
\[s_m \equiv \sum_{n=1}^T \ln P_{k_{m,n-1}k_{mn}}\big(\Delta; \Psi(\theta^*, \sigma^*)\big). \]
\end{proposition}

Conditions (a), (b), and (c) are usual regularity conditions, while (d) is needed for identification. Condition (e) requires that the initial CCP estimates are consistent. As can be seen from the limiting distribution in \eqref{eq:singlenormal}, the asymptotic variance is not affected by the first-step estimation of the CCPs. In other words, in a single agent model there is no efficiency loss from two-step estimation, as was established in discrete time models by \cite{AM02}. This is due to the so-called zero Jacobian property, which we show also holds in continuous-time models below in Proposition~\ref{prop:zero}, where we also relate the property to convergence of the algorithm.

\subsubsection{Dynamic Discrete Games}

We now extend the model to a dynamic discrete game with $N$ players and present the large sample theory first for the PI estimator for $\npliter = 1, \dots, \Npliter$ in Proposition~\ref{prop:largegame} and then the CTNPL estimator in Proposition~\ref{prop:largegame2}.

\begin{proposition}
\label{prop:largegame}
Let $\sigma^*$ be the true set of conditional choice probabilities and let $\Theta$ and $\beliefspace \equiv \underbrace{\Delta^J \times \dots \times \Delta^J}_{NK}$ be the set of possible values of $\theta$ and $\sigma$, where $\Delta^J$ denotes the standard $J$-simplex. Consider the following regularity conditions:

(a) $\Theta$ is a compact set. 

(b) $\Psi(\theta, \sigma)$ is continuous and twice continuously differentiable in $\theta$ and $\sigma$. 

(c) $\Psi_{ijk}(\theta, \sigma)$ be the corresponding element of $\Psi(\theta, \sigma)$ for player $i$'s choice $j$ at state $k$. Then, $\Psi_{ijk}(\theta, \sigma)>0$ for any $i=1, \ldots, N$, $j=0, 1, \ldots, J-1$, $k=1, \ldots, K$, and any $\{\theta, \sigma\} \in \Theta \times \beliefspace$. 
    
(d) There is a unique $\theta^* \in \interior(\Theta)$ such that $\sigma^* = \Psi(\theta^*, \sigma^*)$.

(e) $\hat{\sigma}^0 = M^{-1} \sum_{m=1}^M r_m$ is a strongly consistent nonparametric estimator of $\sigma^*$ such that $\sqrt{M}(\hat{\sigma}^0 - \sigma^*) \xrightarrow{d} \Normal(0, \Sigma^0)$.

 Then, for all $\npliter \leq \Npliter$, 
\[\sqrt{M}(\hat{\theta}^{\npliter} - \theta^*) \xrightarrow{d} \Normal\left(0, \Omega_{\theta \theta'}^{-1} + \Omega_{\theta\theta'}^{-1}\Omega_{\theta \sigma'}\Sigma^{\npliter-1}\Omega_{\theta \sigma'}'\Omega_{\theta \theta'}^{-1} \right) \]
where $s_m = \sum_{n=1}^T P_{k_{m,n-1},k_{m,n}}(\Delta; \Psi(\theta, \sigma))$,
$\Omega_{\theta\theta'} = \E[\nabla_\theta s_m \nabla_{\theta'}s_m]$, $\Omega_{\theta \sigma'} = \E[\nabla_{\theta}s_m\nabla_{\sigma'}s_m]$, and $\Sigma^\npliter = \nabla_\theta\Psi\left(\Omega_{\theta\theta'}^{-1} + \Omega_{\theta\theta'}^{-1}\Omega_{\theta\sigma'}\Sigma^{\npliter-1}\Omega_{\theta\sigma'}'\Omega_{\theta\theta'}^{-1}\right)\nabla_\theta\Psi' + \nabla_\sigma\Psi \Sigma^{\npliter-1}\nabla_\sigma\Psi'$ is the asymptotic variance of $\sqrt{M}(\hat{\sigma}^\npliter - \sigma^*)$ with $\Psi_\theta \equiv \nabla_\theta \Psi(\theta^*, \sigma^*)$ and $\Psi_\sigma \equiv \nabla_\sigma \Psi(\theta^*, \sigma^*)$ denoting Jacobian matrices.
\end{proposition}

The results above are for the PI estimator $\hat{\theta}^\npliter$ for $\npliter =1, \ldots, \Npliter$. \cite{AM07} provide large sample properties for the NPL estimator (i.e., at convergence, if the algorithm converges). We now present similar properties for the CTNPL estimator $\hat{\theta}$ in our continuous time model. A local convergence condition for $\hat{\theta}^\npliter$ will be provided in the next section.

Proposition~\ref{prop:largegame2} below establishes the large sample properties for $(\hat{\theta}, \hat{\sigma})$, but we first introduce a few definitions and notation from \cite{AM07}.
Let $\phi_M(\cdot)$ denote the NPL operator used to represent an iteration of the algorithm:
\[\phi_M(\sigma) = \Psi(\tilde{\theta}_M(\sigma), \sigma) \quad \text{where} \quad \tilde{\theta}_M(\sigma) \equiv \argmax_{\theta \in \Theta} L_M(\theta, \sigma).\]

Let $\Upsilon_M$ be the set of fixed points of $\phi$, $\Upsilon_M \equiv \{(\theta, \sigma) \in \Theta \times [0,1]^{NJK}:\theta = \tilde{\theta}(\sigma) \ \text{and} \ \sigma = \phi_M(\sigma)\}$. Then, choosing the fixed point that has the maximum value of the likelihood,
\[(\hat{\theta}, \hat{\sigma}) \equiv \argmax_{(\theta, \sigma) \in \Upsilon_M} L_M(\theta, \sigma).\]
We denote the population counterparts of $\tilde{\theta}_M(\sigma)$ and $\tilde{\phi}_M(\sigma)$ as $\tilde{\theta}^*(\sigma)\equiv \argmax_{\theta \in \Theta} L(\theta, \sigma)$ and $\phi^*(\sigma) = \Psi(\tilde{\theta}^*(\sigma),\sigma)$.

\begin{proposition}\label{prop:largegame2} Let $\sigma^*$ be the true set of conditional choice probabilities and let $\Theta$ and $\beliefspace \equiv [0,1]^{NJK}$ be the set of possible values of $\theta$ and $\sigma$, respectively. Consider the following regularity conditions:

(a) $\Theta$ is a compact set. 

(b) $\Psi(\theta, \sigma)$ is continuous and twice continuously differentiable in $\theta$ and $\sigma$. 

(c) $\Psi_{ijk}(\theta, \sigma)$ be the corresponding element of $\Psi(\theta, \sigma)$ for player $i$'s choice $j$ at state $k$. Then, $\Psi_{ijk}(\theta, \sigma)>0$ for any $i=1, \ldots, N$, $j=1, \ldots, J-1$, $k=1, \ldots, K$, and any $\{\theta, \sigma\} \in \Theta \times \beliefspace$. 
    
(d) There is a unique $\theta^* \in \interior(\Theta)$ such that $\sigma^* = \Psi(\theta^*, \sigma^*)$.

(e) $(\theta^*, \sigma^*)$ is an isolated population NPL fixed point. 

(f) There exists a closed neighborhood of $\sigma^*$, $\mathcal{N}(\sigma)$, such that, for all $\sigma \in \mathcal{N}(\sigma^*)$, $L(\theta, \sigma)$ is globally concave in $\theta$ and $\partial^2 L(\theta, \sigma^*)/\partial\theta \partial\theta'$ is a nonsingular matrix.

(g) The operator $\phi^*(\sigma)-\sigma$ has a nonsingular Jacobian matrix at $\sigma^*$.

Then,
\[\sqrt{M}(\hat{\theta} - \theta^*) \xrightarrow{d} \Normal\big(0, [\Omega_{\theta \theta'} + \Omega_{\theta\sigma'}(I-\Psi_\sigma)^{-1}\Psi_\theta]^{-1}\Omega_{\theta \theta'}[\Omega_{\theta \theta'} + \Psi_\theta'(I-\Psi'_\sigma)^{-1}\Omega_{\theta\sigma'}']^{-1} \big) \]
where $\Psi_\theta \equiv \nabla_\theta \Psi(\theta^*, \sigma^*)$ and $\Psi_\sigma \equiv \nabla_\sigma \Psi(\theta^*, \sigma^*)$ denote Jacobian matrices.
\end{proposition}

Assumptions from (a)--(d) are the same as Assumptions (a)--(d) in Proposition \ref{prop:largegame}.
Assumptions (e)--(g) correspond to Assumptions (v)--(vii) of Proposition 2 in \cite{AM07}. Note that consistency of $\hat{\sigma}^0$ is not required.

\subsection{Convergence of the CTNPL Algorithm}
\label{sec:convergence}

Originally introduced for discrete time models, the NPL algorithm is a sequential method to find the NPL estimator. This is why it is possible for the NPL algorithm to fail to converge to a consistent estimator even if the NPL estimator is consistent. Ideally, the NPL algorithm could converge to the consistent estimator if it can be evaluated at every fixed point, but this may be infeasible in practice.

Starting from \cite{PS10}, several papers have investigated the convergence of the NPL algorithm.
\cite{KS12} show that a key determinant of the convergence of the NPL algorithm is the contraction property of the mapping $\Psi$ and propose a local convergence condition for NPL estimator in discrete time.
\cite{Aguirregabiria2019} investigate this further and show that even when the population mapping is stable the algorithm can fail to converge when the sample NPL mapping is unstable.
We extend these arguments to the CTNPL estimator and derive a local convergence condition for continuous time models that is analogous to the condition of \cite{KS12}.
Although we see improved convergence in practice due to the sequential nature of moves, as discussed in Section~\ref{sec:sequential}, it is still possible to have multiple equilibria and unstable equilibria in the continous time setting.\footnote{We explore a simple $2 \times 2$ entry model with multiple equilibria in Appendix~\ref{sec:psd08}.}

For simplicity, we define $\Psi_\theta \equiv \nabla_{\theta}\Psi(\theta^*, \sigma^*)$ and $\Psi_\sigma \equiv \nabla_{\sigma}\Psi(\theta^*, \sigma^*)$. Let $P$ be an aggregate transition probability matrix between states where the elements $P_{kl}$ of $P$ are
\[P_{kl} = \begin{cases} \lambda_{ik}\sigma_{ijk} \quad \text{where} \ l(i,j,k) = l & \text{if state change from $k$ to $l$ is endogenous.} \\
[\exp(Q_0)]_{kl} & \text{if state change from $k$ to $l$ is exogenous.} \end{cases}.\]

Let $P^* = \vec P(\Delta; \Psi(\theta^*, \sigma^*))$ denote the $K^2 \times 1$ vectorized version of $P(\Delta; \Psi(\theta^*, \sigma^*))$.
Then, define a scalar $a_{ij,kl}$ as:
\[a_{ij,kl} = \begin{cases} 1 & \text{ if } l(i,j,k) = l \\
0 & \text{otherwise}.\end{cases}\]
We then construct a $NJK \times K^2$ matrix, denoted $A$, by stacking $a_{ij,kl}$ accordingly, and define $\Delta_\sigma=A \diag(P^{*})^{-1}A'$.
Then, we can rewrite $\Omega_{\theta \theta'} = \Psi_{\theta}^{\prime}\Delta_\sigma\Psi_{\theta}$ and $\Omega_{\theta \sigma'} = \Psi^{\prime}_{\theta}\Delta_\sigma\Psi_{\sigma}$. We also define $M_{\Psi_\theta}\equiv I-\Psi_{\theta}(\Psi_{\theta'}\Delta_\sigma \Psi_\theta)^{-1}\Psi_{\theta'} \Delta_\sigma$.

% \begin{comment}
% Let $\sigma$ be a $NJK \times 1$ vector with $\sigma_{ijk}$ for $i=1, \ldots, N$, $j=0, 1, \ldots, J-1$, and $k=1, \ldots, K$, and $\hat{\sigma}$ a sample counterpart of which each element is a conditional choice probability calculated from observations. We denote $\nabla_\sigma = \diag(\lambda\sigma)\\diag(\hat{\lambda\sigma})^{-2}$ where $\lambda$ is a $NJK \times 1$ vector with corresponding $\lambda_{ijk}$. Then, we can rewrite $\Omega_{\theta \theta'} = \Psi_{\theta}^*'\Delta_\sigma\Psi_{\theta}^*$ and $\Omega_{\theta \sigma'} = \Psi^*_{\theta}'\Delta_\sigma\Psi_{\theta}^*$. We also define $M_{\Psi_\theta}\equiv I-\Psi_{\theta}^*(\Psi_{\theta'}^*\Delta_\sigma \Psi_\theta^*)^{-1}\Psi_{\theta'}^* \Delta_\sigma$.
% \end{comment}

\begin{assumption}
\label{convergence assumption}
(a) The conditions of Proposition \ref{prop:largegame2} hold.
(b) $\Psi(\theta, \sigma)$ is three times continuously differentiable in a closed neighborhood $\mathcal{N}$ of $(\theta^*, \sigma^*)$.
(c) $\Omega_{\theta \theta'}$ is nonsingular.
\end{assumption}

We now present a local convergence condition for the CTNPL estimator $\hat{\sigma}^{\npliter}$. If $\hat{\sigma}^\npliter$ converges to $\hat{\sigma}$, $\hat{\theta}^\npliter$ converges to $\hat{\theta}$ by the continuity of $L(\theta, \sigma)$ and applying the Theorem of the Maximum in \eqref{eq:npl1}. So, it is sufficient to show the convergence for $\hat{\sigma}^\npliter$. The convergence condition below is analogous to Proposition 2 in \cite{KS12}.

\begin{proposition}\label{convergence}
Suppose that Assumption \ref{convergence assumption} holds and that $r(M_{\Psi_\theta}\Psi_\sigma)<1$. Then, there exists a neighborhood $\mathcal{N}_1$ of $\sigma^*$ such that, for any initial value $\hat{\sigma}^0 \in \mathcal{N}_1$, we have $\lim_{\Npliter \to \infty}\hat{\sigma}^{\npliter} = \hat{\sigma}$ almost surely. 
\end{proposition}

In single agent models, the Jacobian matrix $\Psi_\sigma$ is zero at the fixed point $\sigma^* = \Psi(\theta^*, \sigma^*)$ by Proposition~\ref{prop:zero}, it follows that $r(M_\Psi \Psi_\sigma)$ is zero. Then, by Proposition~\ref{convergence}, the CTNPL estimator is stable and always converges locally to $(\hat{\theta}, \hat{\sigma})$ in single agent models.
Next, we will investigate many of the properties we have derived in the context of a simulation study.

\section{Monte Carlo Experiments}
\label{sec:mc}

In this section we present the results of a series of Monte Carlo experiments carried out using the model of Example 1 from Section 2, which is a continuous time version of the dynamic discrete game described in Section 4 of \cite{AM07}.
We also provide further implementation details of the CTNPL algorithm for this model in Appendix~\ref{app:detail}.

\subsection{Data Generating Process (DGP)}

We assume the profit functions of firms are:
\begin{equation*}
  u_{ik} =\theta_{\text{RS}}\ln(s_k) - \theta_{\text{RN}}\ln\left(1+\sum_{m\neq i}a_{mk}\right)- \theta_{FC,i}.
\end{equation*}
The parameters we estimate are the fixed costs $\theta_{FC,i}$, entry costs $\theta_{\text{EC}}$, competition effect $\theta_{\text{RN}}$, and market effect $\theta_{\text{RS}}$. Fixed costs $\theta_{FC,i}$ are set as $(\theta_{\text{FC,1}}, \theta_{\text{FC,2}}, \theta_{\text{FC,3}}, \theta_{\text{FC,4}}, \theta_{\text{FC,5}}) = (-1.9, -1.8, -1.7, -1.6, -1.5)$. The parameter $\theta_{\text{RS}}$ equals to 1. We run six experiments varying by entry costs $\theta_{\text{EC}} = (1.0, 2.0, 4.0)$ and $\theta_{\text{RN}}=(0.0, 1.0, 2.0)$. The stochastic component in instantaneous payoff $\epsilon_{ijk}$ are independently and identically distributed and follow Type 1 extreme value distribution.

The support of the logarithm of market size is $\{1,2,3,4,5\}$ and each player chooses an action from $a_{ik} = \{0,1\}$. So, the size of state space is $K = 5 \times 2^5 = 160$.
We assume that market size increases or decreases by at most 1 step at a single instant and that the rate is constant across states $k$. Let $q_1$ and $q_2$, respectively, denote the rate of market size increases and decreases, respectively. These two parameters fully characterize nature's intensity matrix $Q_0$. As mentioned above, we can obtain consistent estimates for $q_1$ and $q_2$, so for our experiments we assume that $Q_0$ is known and focus on estimating firms' payoff parameters. In particular, we choose the intensity matrix $Q_0$ so that the discrete-time transition probability matrix for market size is as in \cite{AM07}\footnote{The transition probability matrix used for market size $s_k$ in \cite{AM07} is
\[\begin{bmatrix} 0.8 & 0.2 & 0.0 & 0.0 & 0.0 \\ 0.2 & 0.6 & 0.2 & 0.0 & 0.0 \\ 0.0 & 0.2 & 0.6 & 0.2 & 0.0 \\ 0.0 & 0.0 & 0.2 & 0.6 & 0.2 \\ 0.0 & 0.0 & 0.0 & 0.2 & 0.8 \end{bmatrix}.\]}.
Throughout the experiments, we fix the discount rate at $\rho_{ik}=0.05$ and the move arrival rate at $\lambda_{ik}=1$ for all $i=1, \ldots, N$ and $k=1, \ldots, K$.

We assume in this section (unless specified otherwise), that the DGP is a continuous time model, however, we consider two distinct sampling regimes for estimation.
First, we consider continuous time data observed every instant.
Since the model is a finite-state jump process, it is equivalent to observe each jump time and the resulting state.
We refer to this as continuous time data.
To generate continuous time data, we calculate the MPE of the game, obtain the steady-state distribution, and draw the initial states $k_{m,0}$ for each market $m$.
From these initial states, we draw the subsequent actions and state $(a_{k_{m,n}}, x_{k_{m,n-1}})$ using the equilibrium conditional choice probabilities.
Second, we also consider discrete time data observed at fixed intervals of length $\Delta = 1$.
We refer to this as discrete time data, noting that the underlying DGP is a continuous time model.
In this case, initial states are also drawn using the steady-state distribution and then we draw the future states according to the transition probability matrix over the time interval, $P(\Delta)$.

For each experiment, we carry out 100 replications using a sample of $M = 400$ markets in each.
For simplicity, we focus on the case where there is one observed continuous-time event or discrete-time observation per market.
Table~\ref{table:dgpdis} summarizes the DGP in the case of discretely-observed data. The parameters for the first three experiments differ only by the strategic interaction parameter, $\theta_{\text{RN}}$, which is increasing from Experiment 1 to 3. Increasing $\theta_{\text{RN}}$ decreases the average number of active firms. % The direction of its effect on the average number of entrants and exits is uncertain in general, because the probability of entry (exit) is increasing (decreasing) but average number of inactive (active) firms is decreasing. 
In this setting, both the average number of entries and exits are increasing when $\theta_{\text{RN}}$ increases.
As excess turnover increases, there will be a larger change in the number of active firms between periods.
As a result, the AR(1) coefficient for the number of current active firms decreases.
As $\theta_{\text{RN}}$ decreases the flow payoff for all firms, firms are less likely to be active regardless of fixed costs.

In Experiments 4--6, $\theta_{\text{EC}}$ varies from 0 to 4. The number of active firms increases as $\theta_{\text{EC}}$ increases, which might be surprising, but this is because the probability of being active differs across firms. Specifically, more efficient firms tend to stay active as the entry cost increases.
This can be seen with Firms 4 and 5, which have lower fixed costs.

\begin{table}[htbp]
\centering
\hspace*{-0.7cm}
\begin{threeparttable}
\caption{DGP Steady State (Discrete Time Data)}
\label{table:dgpdis}
\begin{tabular}{lcccccc}
\toprule
                                      & Exp. 1            & Exp. 2            & Exp. 3            & Exp. 4            & Exp. 5            & Exp. 6            \\
Descriptive                                      & $\theta_{\text{EC}}=1.0$ & $\theta_{\text{EC}}=1.0$ & $\theta_{\text{EC}}=1.0$ & $\theta_{\text{EC}}=0.0$ & $\theta_{\text{EC}}=2.0$ & $\theta_{\text{EC}}=4.0$ \\    
Statistics             & $\theta_{\text{RN}}=0.0$ & $\theta_{\text{RN}}=1.0$ & $\theta_{\text{RN}}=2.0$ & $\theta_{\text{RN}}=1.0$ & $\theta_{\text{RN}}=1.0$ & $\theta_{\text{RN}}=1.0$ \\
\midrule
Average \#active firms\tnote{1}                            & 3.7107            & 2.7744            & 2.0468            & 2.7351            & 2.8027            & 2.8214            \\
& (1.4427)          & (1.5338)          & (1.2510)          & (1.3921)          & (1.6612)          & (1.8139)          \\
AR(1) for \#active firms\tnote{2}    & 0.8012            & 0.7879            & 0.6909            & 0.6720            & 0.8648            & 0.9381            \\
Average \#entrants     & 0.3783            & 0.5024            & 0.5388            & 0.6514            & 0.3653            & 0.1861            \\
 Average \#exits     &                  
 0.3779            & 0.5008            & 0.5385            & 0.6464            & 0.3667            & 0.1870            \\
Excess turnover\tnote{3}          &0.2025            & 0.3096            & 0.3798            & 0.4770            & 0.1768            & 0.0413            \\

Correlation between         &  -0.0030           & -0.0859           & -0.0669           & -0.1240           & -0.0607           & -0.0545           \\
\ entries and exits &                   &                   &                   &                   &                   &              \\
Prob. of being active            &                   &                   &                   &                   &                   &                   \\
\ Firm 1            & 0.7030            & 0.4980            & 0.3352            & 0.5032            & 0.4878            & 0.4567            \\
\ Firm 2            & 0.7237            & 0.5286            & 0.3694            & 0.5263            & 0.5244            & 0.5045            \\
\ Firm 3            & 0.7449            & 0.5530            & 0.4115            & 0.5468            & 0.5611            & 0.5549            \\
\ Firm 4            & 0.7602            & 0.5806            & 0.4443            & 0.5687            & 0.5953            & 0.6179            \\
\ Firm 5            & 0.7790            & 0.6141            & 0.4863            & 0.5902            & 0.6341            & 0.6875     \\ 
\bottomrule
\end{tabular}
    \begin{tablenotes}
\footnotesize
      \item[1] Values in parentheses are standard deviations. 
      \item[2] AR(1) for \#active is the autoregressive coefficient regressing the number of current active firms on the number of active firms in previous period. 
      \item[3] Excess turnover is defined as (\#entrants + \#exits) - abs(\#entrants - \#exits). 
    \end{tablenotes}
\end{threeparttable}
\end{table}  

\subsection{Estimation Results}
\label{sec:result}

We consider several choices for the initial values of the CCPs:
(i) true CCPs,
(ii) frequency estimates,
(iii) semi-parametric logit estimates (using market size and number of rival firms as explanatory variables), and
(iv) random $\Uniform(0,1)$ draws.
Although (i) is infeasible, it yields a $\sqrt{M}$-consistent, asymptotically normal, and efficient estimator that serves as a benchmark \citep[][p. 16]{AM07}.
For (ii), (iii), and (iv), we consider the two-step estimator ($\npliter=1$) and the converged CTNPL ($\npliter = 20$).

Table~\ref{table:mcconti} shows the means and standard deviations, in parentheses, of the Monte Carlo replications in the case of continuous time data.
Based on the same replications, Table~\ref{table:mseconti} displays the square-root of the mean square error (RMSE) of each estimator relative to the infeasible two-step PML estimates.
Using these two tables we can compare the performance of the various estimators.

First, we note that the CTNPL estimator converged to the same estimates within 20 iterations regardless of which CCP estimates were used for initialization (frequency estimates, semi-parametric logit estimates, or random $\Uniform(0,1)$ draws).
Therefore, we report these in a single row denoted ``CTNPL'' in the table.

Of the two-step estimators, initializing using the semi-parametric logit estimates appears to the best choice in terms of low biases, standard deviations, and therefore low relative RMSE values.
Not surprisingly, using inconsistent, random $\Uniform(0,1)$ starting values is a worse choice.
Perhaps surprisingly, the frequency estimates are the worst choice for the initial CCPs.
Even though frequency estimates are consistent, they result in a large bias with low variance across all parameters in Table~\ref{table:mcconti} which in turn leads to high relative RMSE values in Table~\ref{table:mseconti}.

Overall, the performance of the CTNPL estimator---regardless of how it is initialized---is on par with the infeasible infeasible two-step PML estimator (which is initialized using the true CCPs).
The RMSE values for the CTNPL estimator are close to one, meaning the estimators have similar performance, and in most cases the RMSEs are lower than those the other feasible two-step estimators.
For some parameters in some experiments, the two-step semi-parametric logit estimator performs slightly better than CTNPL, but the CTNPL estimator is much more robust than the two-step PML estimator in terms of how it is initialized.
Even if the frequency estimates are used to initialize the CTNPL estimator, the finite sample bias can be reduced by repeatedly imposing the MPE conditions through iteration.

We also carried out these experiments with a much smaller sample size $n=100$.
We report the results in Tables~\ref{table:mcconti100} and \ref{table:mseconti100}.
Here, the benefits of iterating are more pronounced:
Even the 2S-Logit estimates have large amounts of finite sample bias and the improvements from iterating using CTNPL are more substantial.

% The CTNPL estimator has lower standard deviations. For all experiments, biases in $\theta_{\text{RN}}$ are lower in the NPL estimates than in the two-step estimates. This is because the NPL estimator repeatedly impose the MPE condition leading to relative efficiency and more precise estimates for strategic interaction.

% Therefore, the CTNPL estimator is feasible and typically offers better performance than the feasible two-step estimators.

Finally, Table \ref{table:mcdis} and \ref{table:msedis} report the similar results for estimation using discretely observed data.
In this case, the results are even more distinct: the best two-step PML estimator (semi-parametric logit) exhibits more bias while the CTNPL estimator remains closer to the infeasible two-step estimator.
The bias is particularly acute for the strategic interaction parameter $\theta_{\text{RN}}$.

\begin{table}[htbp]
\centering
\begin{threeparttable}
\caption{Monte Carlo Results (Continuous Time Data)}
\label{table:mcconti}
\begin{tabular}{llcccc}
\toprule
        &                      & \multicolumn{4}{c}{Parameters}                                          \\
        \cline{3-6}
Exp. & Estimator & $\theta_{\text{FC,1}}$ & $\theta_{\text{RS}}$ & $\theta_{\text{EC}}$ & $\theta_{\text{RN}}$ \\
\midrule
\textbf{1} & \textbf{True values} & \textbf{-1.9000} & \textbf{1.0000} & \textbf{1.0000} & \textbf{0.0000} \\
 & 2S-True & -1.9112 (0.2933) & 1.0052 (0.0954) & 1.0274 (0.1342) & \ 0.0169 (0.2802) \\
 & 2S-Freq & -0.3623 (0.2410) & 0.3700 (0.1126) & 1.7908 (0.1435) & \ 0.3567 (0.3015) \\
 & 2S-Logit & -1.9136 (0.2948) & 1.0071 (0.0952) & 1.0272 (0.1346) & \ 0.0194 (0.2805) \\
 & 2S-Random & -2.1992 (0.3977) & 1.1214 (0.1031) & 1.0314 (0.1352) & -0.0427 (0.3415) \\
 & CTNPL & -1.9144 (0.2950) & 1.0072 (0.0953) & 1.0273 (0.1345) & \ 0.0190 (0.2819) \\
           \addlinespace[0.2cm]
\textbf{2} & \textbf{True values} & \textbf{-1.9000} & \textbf{1.0000} & \textbf{1.0000} & \textbf{1.0000} \\
 & 2S-True & -1.9099 (0.2009) & 1.0016 (0.0953) & 1.0183 (0.1306) & \ 1.0137 (0.2643) \\
 & 2S-Freq & -0.8255 (0.1896) & 0.4671 (0.1117) & 1.4694 (0.1211) & \ 0.5877 (0.2860) \\
 & 2S-Logit & -1.8932 (0.2113) & 1.0005 (0.1023) & 1.0185 (0.1304) & \ 1.0251 (0.2807) \\
 & 2S-Random & -1.7843 (0.2889) & 0.9587 (0.0724) & 1.0224 (0.1309) & \ 1.0140 (0.3030) \\
 & CTNPL & -1.9163 (0.2078) & 1.0032 (0.1030) & 1.0184 (0.1304) & \ 1.0145 (0.2761) \\
           \addlinespace[0.2cm]
\textbf{3} & \textbf{True values} & \textbf{-1.9000} & \textbf{1.0000} & \textbf{1.0000} & \textbf{2.0000} \\
 & 2S-True & -1.9401 (0.1770) & 1.0117 (0.0808) & 1.0070 (0.0960) & \ 2.0182 (0.2665) \\
 & 2S-Freq & -1.1756 (0.2523) & 0.5865 (0.1211) & 1.3317 (0.1055) & \ 1.1303 (0.3315) \\
 & 2S-Logit & -1.8700 (0.2034) & 1.0077 (0.0916) & 1.0070 (0.0962) & \ 2.0714 (0.2999) \\
 & 2S-Random & -1.2428 (0.2885) & 0.8001 (0.0624) & 1.0215 (0.0928) & \ 1.8789 (0.3228) \\
 & CTNPL & -1.9488 (0.2013) & 1.0147 (0.0930) & 1.0070 (0.0961) & \ 2.0196 (0.2867) \\
           \addlinespace[0.2cm]
\textbf{4} & \textbf{True values} & \textbf{-1.9000} & \textbf{1.0000} & \textbf{0.0000} & \textbf{1.0000} \\
 & 2S-True & -1.9208 (0.2382) & 1.0219 (0.0950) & -0.0071 (0.1149) & \ 1.0665 (0.3163) \\
 & 2S-Freq & -0.9979 (0.2103) & 0.5025 (0.1096) & \ 0.4399 (0.1097) & \ 0.6593 (0.3496) \\
 & 2S-Logit & -1.9081 (0.2458) & 1.0233 (0.1015) & -0.0073 (0.1146) & \ 1.0835 (0.3347) \\
 & 2S-Random & -1.5546 (0.2965) & 0.8156 (0.0697) & \ 0.0249 (0.1138) & \ 0.8344 (0.2922) \\
 & CTNPL & -1.9260 (0.2410) & 1.0247 (0.1021) & -0.0071 (0.1149) & \ 1.0727 (0.3297) \\
           \addlinespace[0.2cm]
\textbf{5} & \textbf{True values} & \textbf{-1.9000} & \textbf{1.0000} & \textbf{2.0000} & \textbf{1.0000} \\
 & 2S-True & -1.9320 (0.1902) & 1.0175 (0.1033) & 2.0155 (0.1334) & 1.0336 (0.2189) \\
 & 2S-Freq & -0.7240 (0.2086) & 0.4921 (0.1171) & 2.4653 (0.1393) & 0.6563 (0.2428) \\
 & 2S-Logit & -1.9079 (0.2029) & 1.0165 (0.1118) & 2.0152 (0.1337) & 1.0493 (0.2381) \\
 & 2S-Random & -2.1011 (0.3297) & 1.1627 (0.0987) & 1.9982 (0.1292) & 1.3001 (0.3330) \\
 & CTNPL & -1.9416 (0.2045) & 1.0219 (0.1132) & 2.0155 (0.1336) & 1.0390 (0.2346) \\
           \addlinespace[0.2cm]
\textbf{6} & \textbf{True values} & \textbf{-1.9000} & \textbf{1.0000} & \textbf{4.0000} & \textbf{1.0000} \\
 & 2S-True & -1.9277 (0.1966) & 1.0167 (0.1094) & 4.0347 (0.2232) & 1.0180 (0.1934) \\
 & 2S-Freq & -0.5088 (0.1907) & 0.4460 (0.1182) & 4.4598 (0.2255) & 0.5811 (0.2107) \\
 & 2S-Logit & -1.8957 (0.2064) & 1.0113 (0.1205) & 4.0319 (0.2237) & 1.0291 (0.2098) \\
 & 2S-Random & -2.8593 (0.4994) & 1.6734 (0.1611) & 3.9831 (0.2210) & 2.0784 (0.5526) \\
 & CTNPL & -1.9448 (0.2087) & 1.0214 (0.1206) & 4.0332 (0.2232) & 1.0180 (0.2040)                        \\ 
\bottomrule
\end{tabular}
    \begin{tablenotes}
\footnotesize
      \item Displayed values are means with standard deviations in parentheses.
    \end{tablenotes}
\end{threeparttable}
\end{table}

\begin{table}[htbp]
\centering
\begin{threeparttable}
\caption{RMSE Relative to Two-Step PML with True CCPs (Continuous Time Data)}
\label{table:mseconti}
\begin{tabular}{llccccc}
\toprule
     &           & \multicolumn{4}{c}{Parameters}                                  \\
     \cline{3-6}
Exp. & Estimator & $\theta_{\text{FC,1}}$ & $\theta_{\text{RS}}$ & $\theta_{\text{EC}}$ & $\theta_{\text{RN}}$ \\
\midrule
1 & 2S-Freq & 5.3031 & 6.6989 & 5.8694 & 1.6638 \\
 & 2S-Logit & 1.0056 & 0.9992 & 1.0027 & 1.0015 \\
 & 2S-Random & 1.6955 & 1.6671 & 1.0137 & 1.2260 \\
 & CTNPL & 1.0064 & 1.0004 & 1.0022 & 1.0063 \\
            \addlinespace[0.2cm]
2 & 2S-Freq & 5.4232 & 5.7112 & 3.6768 & 1.8962 \\
 & 2S-Logit & 1.0506 & 1.0727 & 0.9987 & 1.0650 \\
 & 2S-Random & 1.5469 & 0.8746 & 1.0070 & 1.1461 \\
 & CTNPL & 1.0362 & 1.0806 & 0.9991 & 1.0450 \\
            \addlinespace[0.2cm]
3 & 2S-Freq & 4.2254 & 5.2762 & 3.6157 & 3.4849 \\
 & 2S-Logit & 1.1325 & 1.1259 & 1.0016 & 1.1543 \\
 & 2S-Random & 3.9534 & 2.5644 & 0.9896 & 1.2909 \\
 & CTNPL & 1.1408 & 1.1534 & 1.0012 & 1.0759 \\
            \addlinespace[0.2cm]
4 & 2S-Freq & 3.8739 & 5.2251 & 3.9373 & 1.5103 \\
 & 2S-Logit & 1.0287 & 1.0685 & 0.9973 & 1.0674 \\
 & 2S-Random & 1.9037 & 2.0219 & 1.0114 & 1.0392 \\
 & CTNPL & 1.0140 & 1.0778 & 0.9993 & 1.0445 \\
            \addlinespace[0.2cm]
5 & 2S-Freq & 6.1921 & 4.9729 & 3.6155 & 1.9001 \\
 & 2S-Logit & 1.0525 & 1.0779 & 1.0018 & 1.0978 \\
 & 2S-Random & 2.0023 & 1.8152 & 0.9619 & 2.0241 \\
 & CTNPL & 1.0818 & 1.0998 & 1.0015 & 1.0738 \\
            \addlinespace[0.2cm]
6 & 2S-Freq & 7.0717 & 5.1198 & 2.2676 & 2.4137 \\
 & 2S-Logit & 1.0397 & 1.0937 & 1.0005 & 1.0901 \\
 & 2S-Random & 5.4465 & 6.2583 & 0.9815 & 6.2374 \\
 & CTNPL & 1.0752 & 1.1074 & 0.9993 & 1.0540     \\ 
\bottomrule
\end{tabular}
\end{threeparttable}
\end{table}

\begin{table}[tbph]
\centering
\begin{threeparttable}
\caption{Monte Carlo Results (Continuous Time Data, $n=100$)}
\label{table:mcconti100}
\begin{tabular}{llcccc}
\toprule
        &                      & \multicolumn{4}{c}{Parameters}                                          \\
        \cline{3-6}
Exp. & Estimator & $\theta_{\text{FC,1}}$ & $\theta_{\text{RS}}$ & $\theta_{\text{EC}}$ & $\theta_{\text{RN}}$ \\
\midrule
\textbf{1} & \textbf{True values} & \textbf{-1.9000} & \textbf{1.0000} & \textbf{1.0000}  & \textbf{0.0000}  \\
  & 2S-True              & -2.1018 (0.9937) & 1.1040 (0.3446) & 0.9456 (0.6370)  & \ 0.0216 (0.8111)  \\
  & 2S-Freq              & \ 0.2295 (0.4986)  & 0.3171 (0.1187) & 2.0410 (0.5382)  & \ 0.5863 (0.3954)  \\
  & 2S-Logit             & -1.9883 (0.9231) & 1.0436 (0.3135) & 0.9396 (0.6676)  & \ 0.0304 (0.6545)  \\
  & 2S-Random            & -2.4118 (1.2242) & 1.2346 (0.3154) & 0.9368 (0.6475)  & -0.0328 (1.0578) \\
  & CTNPL                & -2.1606 (1.0891) & 1.1208 (0.3496) & 0.9551 (0.6276)  & \ 0.0217 (0.8738)  \\
\textbf{2} & \textbf{True values} & \textbf{-1.9000} & \textbf{1.0000} & \textbf{1.0000}  & \textbf{1.0000}  \\
  & 2S-True              & -2.0290 (0.7257) & 1.0847 (0.3819) & 0.9574 (0.5941)  & 1.1370 (0.9101)  \\
  & 2S-Freq              & -0.5441 (0.4981) & 0.5121 (0.1918) & 1.8630 (0.5064)  & 0.8795 (0.4138)  \\
  & 2S-Logit             & -1.8820 (0.6793) & 0.9341 (0.3265) & 0.9550 (0.6072)  & 0.8434 (0.7188)  \\
  & 2S-Random            & -1.9295 (0.7970) & 1.0537 (0.2469) & 0.9073 (0.6378)  & 1.1632 (0.9019)  \\
  & CTNPL                & -2.1008 (0.7896) & 1.1160 (0.4660) & 0.9578 (0.5952)  & 1.1713 (1.0605)  \\
\textbf{3} & \textbf{True values} & \textbf{-1.9000} & \textbf{1.0000} & \textbf{1.0000}  & \textbf{2.0000}  \\
  & 2S-True              & -2.0736 (0.5717) & 1.0699 (0.3333) & 0.9536 (0.4541)  & 2.1462 (0.9528)  \\
  & 2S-Freq              & -0.8284 (0.4540) & 0.4904 (0.1804) & 1.9270 (0.3658)  & 0.9465 (0.4555)  \\
  & 2S-Logit             & -1.8239 (0.5270) & 0.8629 (0.2682) & 0.9546 (0.4581)  & 1.5919 (0.7613)  \\
  & 2S-Random            & -1.3648 (0.6798) & 0.8583 (0.2111) & 0.9348 (0.4690)  & 2.0179 (0.9448)  \\
  & CTNPL                & -2.1490 (0.6867) & 1.1023 (0.4116) & 0.9526 (0.4554)  & 2.2183 (1.1231)  \\
\textbf{4} & \textbf{True values} & \textbf{-1.9000} & \textbf{1.0000} & \textbf{0.0000}  & \textbf{1.0000}  \\
  & 2S-True              & -2.1030 (1.1315) & 1.1795 (0.499)  & -0.1843 (0.7842) & 1.3009 (1.2033)  \\
  & 2S-Freq              & -0.4899 (0.4885) & 0.4487 (0.1908) & 1.1170 (0.4428)  & 0.8117 (0.4851)  \\
  & 2S-Logit             & -1.9415 (1.0210)  & 0.9800 (0.4106) & -0.2136 (0.8581) & 0.8986 (0.8859)  \\
  & 2S-Random            & -1.5944 (1.0009) & 0.9107 (0.2640)  & -0.1169 (0.7311) & 1.0397 (0.9615)  \\
  & CTNPL                & -2.1926 (1.2283) & 1.2265 (0.5703) & -0.1817 (0.7789) & 1.3779 (1.3453)  \\
\textbf{5} & \textbf{True values} & \textbf{-1.9000} & \textbf{1.0000} & \textbf{2.0000}  & \textbf{1.0000}  \\
  & 2S-True              & -1.9874 (0.5975) & 1.0542 (0.3354) & 2.0038 (0.4998)  & 1.0762 (0.6414)  \\
  & 2S-Freq              & -0.6184 (0.5366) & 0.5642 (0.2447) & 2.6802 (0.5558)  & 0.8779 (0.4226)  \\
  & 2S-Logit             & -1.8277 (0.5542) & 0.9380 (0.2761) & 2.0046 (0.4991)  & 0.8852 (0.5318)  \\
  & 2S-Random            & -2.1763 (0.8047) & 1.2299 (0.2693) & 1.9683 (0.4961)  & 1.4027 (0.8789)  \\
  & CTNPL                & -2.0425 (0.6532) & 1.0891 (0.3809) & 2.0042 (0.4982)  & 1.1229 (0.7150)  \\
\textbf{6} & \textbf{True values} & \textbf{-1.9000} & \textbf{1.0000} & \textbf{4.0000}  & \textbf{1.0000}  \\
  & 2S-True              & -1.9388 (0.6295) & 1.0158 (0.3350)  & 4.1877 (0.5996)  & 1.0171 (0.5425)  \\
  & 2S-Freq              & -0.7903 (0.8146) & 0.6775 (0.3716) & 4.4727 (0.6089)  & 0.8876 (0.5212)  \\
  & 2S-Logit             & -1.8265 (0.6019) & 0.9373 (0.3132) & 4.1830 (0.6019)  & 0.8971 (0.5241)  \\
  & 2S-Random            & -3.1589 (1.1246) & 1.8696 (0.4750) & 4.1274 (0.6109)  & 2.4380 (1.3796)  \\
  & CTNPL                & -2.0335 (0.6859) & 1.0465 (0.3722) & 4.1792 (0.6004)  & 1.0277 (0.5927)  \\
\bottomrule
\end{tabular}
    \begin{tablenotes}
\footnotesize
      \item Displayed values are means with standard deviations in parentheses.
    \end{tablenotes}
\end{threeparttable}
\end{table}

\begin{table}[htbp]
\centering
\begin{threeparttable}
\caption{RMSE Relative to Two-Step PML with True CCPs (Continuous Time Data, $n=100$)}
\label{table:mseconti100}
\begin{tabular}{llccccc}
\toprule
     &           & \multicolumn{4}{c}{Parameters}                                  \\
     \cline{3-6}
Exp. & Estimator & $\theta_{\text{FC,1}}$ & $\theta_{\text{RS}}$ & $\theta_{\text{EC}}$ & $\theta_{\text{RN}}$ \\
\midrule
1         & 2S-Freq   & 2.1569 & 1.9255 & 1.8330 & 0.8715 \\
 & 2S-Logit  & 0.9144 & 0.8792 & 1.0484 & 0.8075 \\
 & 2S-Random & 1.3085 & 1.0919 & 1.0176 & 1.3043 \\
 & CTNPL     & 1.1044 & 1.0277 & 0.9841 & 1.0772 \\
2         & 2S-Freq   & 1.9598 & 1.3402 & 1.6800 & 0.4683 \\
 & 2S-Logit  & 0.9220 & 0.8514 & 1.0222 & 0.7993 \\
          & 2S-Random & 1.0821 & 0.6460 & 1.0821 & 0.9958 \\
 & CTNPL     & 1.1054 & 1.2278 & 1.0018 & 1.1671 \\
3         & 2S-Freq   & 1.9479 & 1.5876 & 2.1831 & 1.1907 \\
 & 2S-Logit  & 0.8912 & 0.8846 & 1.0085 & 0.8962 \\
 & 2S-Random & 1.4480 & 0.7464 & 1.0373 & 0.9804 \\
 & CTNPL     & 1.2226 & 1.2454 & 1.0029 & 1.1869 \\
4         & 2S-Freq   & 1.2981 & 1.1002 & 1.4916 & 0.4195 \\
 & 2S-Logit  & 0.8888 & 0.7752 & 1.0978 & 0.7189 \\
 & 2S-Random & 0.9103 & 0.5256 & 0.9192 & 0.7759 \\
 & CTNPL     & 1.0984 & 1.1572 & 0.9930 & 1.1266 \\
5         & 2S-Freq   & 2.3009 & 1.4712 & 1.7573 & 0.6810 \\
 & 2S-Logit  & 0.9255 & 0.8330 & 0.9985 & 0.8423 \\
          & 2S-Random & 1.4090 & 1.0422 & 0.9945 & 1.4966 \\
 & CTNPL     & 1.1071 & 1.1515 & 0.9967 & 1.1231 \\
6         & 2S-Freq   & 2.1825 & 1.4671 & 1.2268 & 0.9824 \\
 & 2S-Logit  & 0.9613 & 0.9524 & 1.0013 & 0.9842 \\
 & 2S-Random & 2.6764 & 2.9543 & 0.9932 & 3.6718 \\
 & CTNPL     & 1.1078 & 1.1184 & 0.9973 & 1.0932      \\
\bottomrule
\end{tabular}
\end{threeparttable}
\end{table}

\begin{table}[htbp]
\centering
\begin{threeparttable}
\caption{Monte Carlo Results (Discrete Time Data)}
\label{table:mcdis}
\begin{tabular}{llcccc}
\toprule
        &                      & \multicolumn{4}{c}{Parameters}                                          \\
        \cline{3-6}
Exp. & Estimator & $\theta_{\text{FC,1}}$ & $\theta_{\text{RS}}$ & $\theta_{\text{EC}}$ & $\theta_{\text{RN}}$ \\
\midrule
\textbf{1} & \textbf{True values} & \textbf{-1.9000} & \textbf{1.0000} & \textbf{1.0000} & \textbf{0.0000} \\
 & 2S-True & -1.8696 (0.4613) & 1.0321 (0.1864) & 1.0126 (0.3056) & \ 0.0710 (0.4025) \\
 & 2S-Freq & -0.2995 (0.2997) & 0.3971 (0.1259) & 1.9794 (0.3101) & \ 0.3994 (0.2739) \\
 & 2S-Logit & -1.7440 (0.4069) & 0.9755 (0.1655) & 1.0055 (0.3129) & \ 0.0860 (0.3060) \\
 & 2S-Random & -2.1636 (0.5416) & 1.1508 (0.1539) & 1.0134 (0.2922) & -0.0045 (0.4948) \\
 & CTNPL & -1.8830 (0.4711) & 1.0383 (0.1952) & 1.0146 (0.3040) & \ 0.0747 (0.4094) \\
     \addlinespace[0.2cm]
\textbf{2} & \textbf{True values} & \textbf{-1.9000} & \textbf{1.0000} & \textbf{1.0000} & \textbf{1.0000} \\
 & 2S-True & -1.9433 (0.3396) & 1.0385 (0.1613) & 0.9871 (0.2539) & 1.0938 (0.3765) \\
 & 2S-Freq & -0.6917 (0.2469) & 0.4395 (0.1124) & 1.6766 (0.2417) & 0.5812 (0.2729) \\
 & 2S-Logit & -1.7812 (0.3104) & 0.8945 (0.1339) & 0.9858 (0.2555) & 0.8286 (0.2973) \\
 & 2S-Random & -1.7902 (0.4500) & 0.9946 (0.1114) & 0.9748 (0.2538) & 1.1124 (0.4437) \\
 & CTNPL & -1.9558 (0.3573) & 1.0461 (0.1775) & 0.9880 (0.2527) & 1.1051 (0.4033) \\
     \addlinespace[0.2cm]
\textbf{3} & \textbf{True values} & \textbf{-1.9000} & \textbf{1.0000} & \textbf{1.0000} & \textbf{2.0000} \\
 & 2S-True & -1.9298 (0.2940) & 1.0302 (0.1646) & 0.9909 (0.2158) & 2.0662 (0.4787) \\
 & 2S-Freq & -0.8447 (0.2504) & 0.4236 (0.1133) & 1.5584 (0.2076) & 0.7669 (0.3106) \\
 & 2S-Logit & -1.6782 (0.2705) & 0.8325 (0.1329) & 0.9908 (0.2160) & 1.5422 (0.3803) \\
 & 2S-Random & -1.2187 (0.3895) & 0.8255 (0.1008) & 0.9921 (0.2087) & 1.9675 (0.5099) \\
 & CTNPL & -1.9521 (0.3397) & 1.0416 (0.1953) & 0.9912 (0.2159) & 2.0869 (0.5427) \\
     \addlinespace[0.2cm]
\textbf{4} & \textbf{True values} & \textbf{-1.9000} & \textbf{1.0000} & \textbf{0.0000} & \textbf{1.0000} \\
 & 2S-True & -1.9361 (0.4290) & 0.9941 (0.1813) & 0.0073 (0.3143) & 0.9462 (0.4605) \\
 & 2S-Freq & -0.6702 (0.2656) & 0.3328 (0.0856) & 0.8494 (0.2739) & 0.3373 (0.2498) \\
 & 2S-Logit & -1.7521 (0.3658) & 0.8327 (0.1387) & 0.0044 (0.3153) & 0.6578 (0.3278) \\
 & 2S-Random & -1.6005 (0.4331) & 0.8264 (0.1074) & 0.0254 (0.2973) & 0.7897 (0.4397) \\
 & CTNPL & -1.9443 (0.4427) & 0.9993 (0.1939) & 0.0068 (0.3146) & 0.9524 (0.4843) \\
     \addlinespace[0.2cm]
\textbf{5} & \textbf{True values} & \textbf{-1.9000} & \textbf{1.0000} & \textbf{2.0000} & \textbf{1.0000} \\
 & 2S-True & -1.9124 (0.2865) & 1.0140 (0.1627) & 1.9971 (0.2054) & 1.0125 (0.3300) \\
 & 2S-Freq & -0.6527 (0.2805) & 0.4642 (0.1309) & 2.5805 (0.2464) & 0.5930 (0.2421) \\
 & 2S-Logit & -1.7567 (0.2654) & 0.8959 (0.1376) & 1.9996 (0.2050) & 0.8150 (0.2730) \\
 & 2S-Random & -2.0970 (0.3999) & 1.1894 (0.1401) & 1.9635 (0.2184) & 1.3346 (0.4834) \\
 & CTNPL & -1.9285 (0.3105) & 1.0209 (0.1811) & 1.9961 (0.2065) & 1.0190 (0.3562) \\
     \addlinespace[0.2cm]
\textbf{6} & \textbf{True values} & \textbf{-1.9000} & \textbf{1.0000} & \textbf{4.0000} & \textbf{1.0000} \\
 & 2S-True & -1.9634 (0.2460) & 1.0438 (0.1470) & 4.0477 (0.2482) & 1.0649 (0.2560) \\
 & 2S-Freq & -0.7220 (0.3098) & 0.5674 (0.1777) & 4.4336 (0.2896) & 0.7483 (0.2930) \\
 & 2S-Logit & -1.8252 (0.2365) & 0.9631 (0.1388) & 4.0460 (0.2452) & 0.9616 (0.2545) \\
 & 2S-Random & -2.9300 (0.5703) & 1.8220 (0.2097) & 3.9921 (0.2574) & 2.4322 (0.7285) \\
 & CTNPL & -1.9921 (0.2694) & 1.0570 (0.1675) & 4.0473 (0.2497) & 1.0799 (0.2934) \\ 
\bottomrule
\end{tabular}
    \begin{tablenotes}
\footnotesize
      \item Displayed values are means with standard deviations in parentheses.
    \end{tablenotes}
\end{threeparttable}
\end{table}

\begin{table}[htbp]
\centering
\begin{threeparttable}
\caption{RMSE Relative to Two-Step PML with True CCPs (Discrete Time Data)}
\label{table:msedis}
\begin{tabular}{llccccc}
\toprule
     &           & \multicolumn{4}{c}{Parameters}                                  \\
     \cline{3-6}
Exp. & Estimator & $\theta_{\text{FC,1}}$ & $\theta_{\text{RS}}$ & $\theta_{\text{EC}}$ & $\theta_{\text{RN}}$ \\
\midrule
1 & 2S-Freq & 3.5224 & 3.2572 & 3.3591 & 1.1850 \\
 & 2S-Logit & 0.9426 & 0.8844 & 1.0232 & 0.7777 \\
 & 2S-Random & 1.3030 & 1.1393 & 0.9565 & 1.2108 \\
 & CTNPL & 1.0197 & 1.0519 & 0.9951 & 1.0182 \\
      \addlinespace[0.2cm]
2 & 2S-Freq & 3.6028 & 3.4469 & 2.8259 & 1.2885 \\
 & 2S-Logit & 0.9709 & 1.0280 & 1.0065 & 0.8844 \\
 & 2S-Random & 1.3533 & 0.6723 & 1.0030 & 1.1797 \\
 & CTNPL & 1.0566 & 1.1059 & 0.9949 & 1.0742 \\
      \addlinespace[0.2cm]
3 & 2S-Freq & 3.6707 & 3.5096 & 2.7585 & 2.6315 \\
 & 2S-Logit & 1.1837 & 1.2774 & 1.0009 & 1.2316 \\
 & 2S-Random & 2.6560 & 1.2038 & 0.9673 & 1.0574 \\
 & CTNPL & 1.1632 & 1.1926 & 1.0006 & 1.1373 \\
      \addlinespace[0.2cm]
4 & 2S-Freq & 2.9227 & 3.7093 & 2.8391 & 1.5275 \\
 & 2S-Logit & 0.9166 & 1.1981 & 1.0030 & 1.0222 \\
 & 2S-Random & 1.2232 & 1.1255 & 0.9492 & 1.0514 \\
 & CTNPL & 1.0335 & 1.0692 & 1.0012 & 1.0497 \\
      \addlinespace[0.2cm]
5 & 2S-Freq & 4.4578 & 3.3784 & 3.0697 & 1.4344 \\
 & 2S-Logit & 1.0517 & 1.0566 & 0.9978 & 0.9989 \\
 & 2S-Random & 1.5544 & 1.4431 & 1.0777 & 1.7804 \\
 & CTNPL & 1.0871 & 1.1165 & 1.0055 & 1.0802 \\
      \addlinespace[0.2cm]
6 & 2S-Freq & 4.7959 & 3.0490 & 2.0634 & 1.4625 \\
 & 2S-Logit & 0.9768 & 0.9366 & 0.9873 & 0.9747 \\
 & 2S-Random & 4.6357 & 5.5309 & 1.0191 & 6.0838 \\
 & CTNPL & 1.1210 & 1.1537 & 1.0055 & 1.1515    \\ 
\bottomrule
\end{tabular}
\end{threeparttable}
\end{table}  

\subsection{Strategic Interaction and Convergence}

In the specification of Experiment 1, the strategic interaction between agents, denoted by $\theta_{\text{RN}}$, is set to zero which means that this setting corresponds to a collection of single agent models.
In light of the zero Jacobian property of Proposition~\ref{prop:zero}, the $\Psi$ mapping is stable in this case and the algorithm will converge. However, as $\theta_{\text{RN}}$ increases it becomes more likely that the $\Psi$ mapping is unstable (both in a finite sample and in the population). \cite{KS12} showed, using a two-firm example, that a MPE for which the local convergence holds exists if the contemporaneous and dynamic interaction between firms is small.

In Section 2.3, \cite{KS12} showed that $r(M_{\Psi_P} \Psi_P)$ is close to $r(\Psi_P)$ in a typical setting where the number of states is larger than the number of parameters to estimate. This is also true in continuous-time games, so similarly we will focus here on $r(\Psi_\sigma)$ given that $\Psi_\sigma$ is often related to the characteristics of the economic model \citep{KS12}.

\cite{Aguirregabiria2019} study the relationship between $\theta_{\text{RN}}$ and the spectral radius $r(\Psi_P)$ for the discrete time NPL estimator using the five player example game of \cite{AM07}. We conjecture that $r(\Psi_\sigma)$ is smaller in continuous time games than the spectral radius in a comparable discrete time game since continuous time games do not allow simultaneous moves between agents, making it more likely that the best response mapping is more stable.  We investigate this in the context of the model used for our Monte Carlo experiments.

\begin{figure}[htbp]
\centering
\begin{subfigure}[b]{0.75\textwidth}
  \centering
  \includegraphics[width=\linewidth]{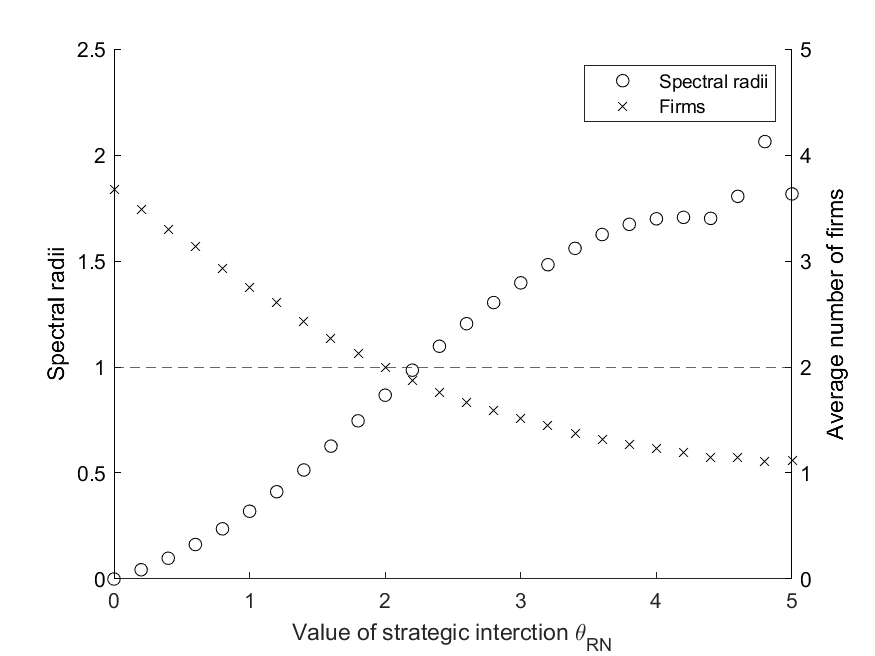}
  \caption{Discrete Time}
  \label{fig:spec_dtdt}
\end{subfigure}

\begin{subfigure}[b]{0.75\textwidth}
  \centering
  \includegraphics[width=\linewidth]{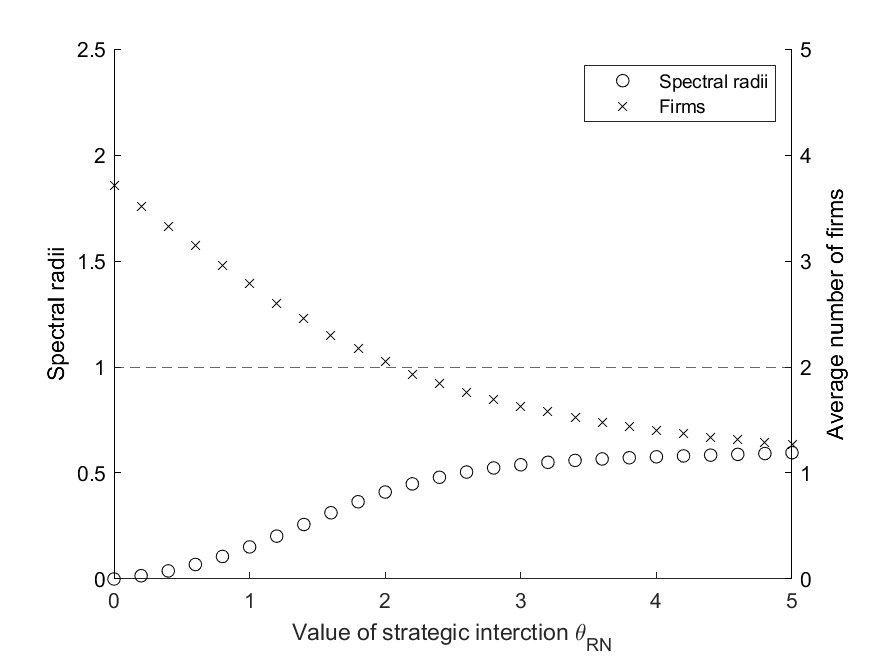}
  \caption{Continuous Time}
  \label{fig:spec_ctdt}
\end{subfigure}
\caption{Spectral Radius vs. $\theta_{\text{RN}}$}
\label{fig:test}
\end{figure}

Figure~\ref{fig:test} shows that the conjecture holds in the present model.
Displayed in the top panel is our replication of a similar figure in \cite{Aguirregabiria2019}.
This figure shows the relationship between the spectral radii for different values of $\theta_{\text{RN}}$ using a discrete time version of the five player game in \cite{AM07}.
In the figure, the spectral radius becomes larger than 1 when $\theta_{\text{RN}}$ is larger than 2.2.
The lower panel of Figure~\ref{fig:test} shows the relationship between the spectral radius and strategic interaction parameter in the continuous time model.
To construct the figure, we use the same setting of Experiment 1 above and examine how $r(\Psi_\sigma)$ changes with respect to $\theta_{\text{RN}}$.
To produce the figure, we generate different data sets for each value of $\theta_{\text{RN}}$ and calculate the spectral radius of $\Psi_{\sigma}$.
Even when $\theta_{\text{RN}}$ increases to 5 the spectral radius is 0.6, still far below the threshold for instability.
On the right y-axis, we also report the average number of active firms in each game.
As the strategic interaction parameter increases, the number of firms varies from almost 4 to around 1.

Although we can see that the number of active firms in the discrete time and continuous time models is similar, the meaning of a particular value for the strategic interaction parameter $\theta_{\text{RN}}$ is different in the two models.
In order to provide another means to compare the spectral radii, we also compare the discrete time model with the best-fit continuous time model in Figure~\ref{fig:spectral}.
To do so, we use the same discrete time data for each replication and use it to estimate the continuous time model.
Then we calculate the Jacobian $\Psi_\sigma$ as before and calculate the spectral radius.

\begin{figure}[htbp]
\centering
\caption{Spectral Radius vs $\theta_{\text{RN}}$: Best Fit Comparison}
\label{fig:spectral}
\includegraphics[width=12cm]{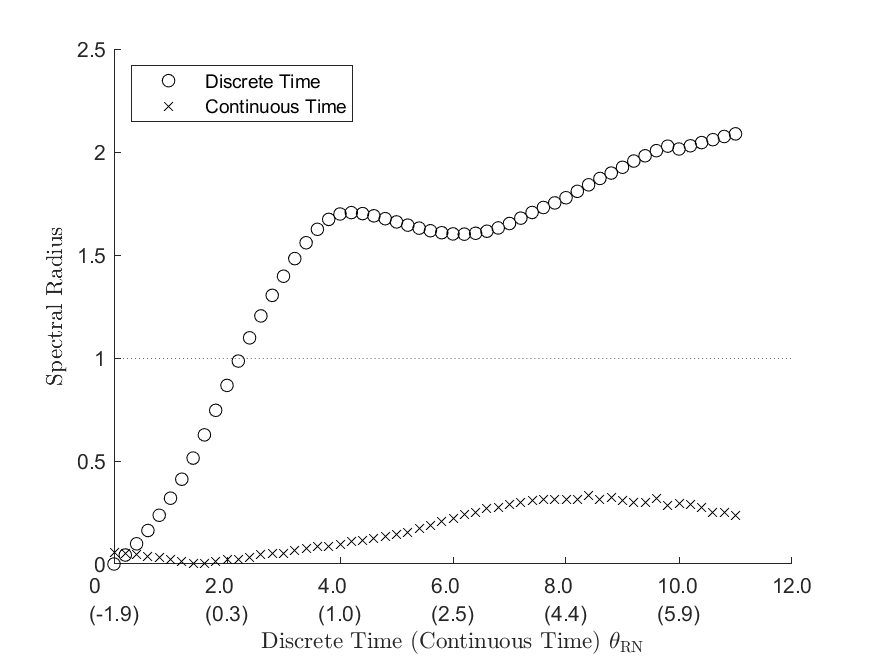}
\floatfoot{Note: The numbers in parentheses are corresponding estimates for $\theta_{\text{RN}}$ from a continuous time game.}
\end{figure}

As such, the horizontal axis of Figure~\ref{fig:spectral} has two labels. The upper labels denote the original values of $\theta_{\text{RN}}$ used to generate the discrete time data. The numbers in parentheses below are the corresponding estimates of $\theta_{\text{RN}}$ from the continuous time model. Note that when $\theta_{\text{RN}}$ is zero in continuous time, the spectral radius is also zero as established in Proposition~\ref{prop:zero}. Note that the continuous time estimates for the strategic interaction parameter are lower, and sometimes even negative, but the spectral radii are calculated using the absolute values of the eigenvalues so they are always positive. As $\theta_{\text{RN}}$ increases to 11, spectral radius is at most around 0.23, even farther from one than before, so it it highly unlikely that there will be any convergence issues. As a result, in this example we can allow for a wider range of strategic interaction values than in discrete time without worrying about convergence to an inconsistent estimator.
Even though the estimates always converged for the continuous time model, while the discrete time model failed to converge in many cases, this is not a universal result.
In general, there can be multiple equilibria in continuous time models and unstable CTNPL fixed points, as discussed in Appendix~\ref{sec:psd08}.
However, the sequential nature of moves in continuous time appears to improve convergence in practice in this heterogeneous agent entry-exit model.

\subsection{Estimation of Misspecified Discrete Time Models}
\label{sec:mis}

Finally, in this section we consider the effects of misspecification
when one estimates a simultaneous move, discrete time model but in
reality the DGP is a continuous time model with asynchronous moves.
We focus on this direction of misspecification because currently
the most common practice is for applied researchers to use
simultaneous move, discrete time models estimated using snapshot
(time-aggregated) data.
Continuous time models make up a small, but emerging share of the
literature and we consider the other direction of
misspecification in Appendix~\ref{app:dtct}.

It is not a foregone conclusion that estimating a misspecified model would lead all parameter estimates and counterfactuals to be incorrect.
Any finite-state, continuous-time Markov jump process, including those generated by the structural model we consider, has an embedded discrete time Markov chain, or jump chain, that characterizes transitions at jump times without regard for the time elapsed between jumps \citep[][p. 153]{karlin81second}.
Therefore, one may wonder whether the discrete time game could approximate the continuous time game well enough to avoid bias in some parameters.
We show that this is not the case in our example entry model: the transition matrix is only one component of the complete structural model and the specification of sequential vs simultaneous moves is also important and ultimately still leads to misspecification bias.

We report estimates for the misspecified discrete time NPL estimator in Table~\ref{table:mis:ctdt:dtdt} along with the correctly specified CTNPL estimates for comparison (the latter being reproduced from Table~\ref{table:mcdis}).
There is bias in all parameters, however it is important to note that there are severe biases in the important strategic interaction and entry cost parameters.
It is also of interest to learn how these biases behave in different settings.
In particular, the biases tend to be larger when entry costs are smaller.

Although the estimates for $\theta_{\text{RS}}$ have small standard deviations, the biases are larger than those of estimates from correctly specified models in all experiments. This result suggests that when a researcher wrongly chooses a discrete time model to estimate parameters from continuous data, it is possible to arrive at inconsistent estimates with large finite sample bias although seeming precise with small standard errors.

The underlying issue is that the discrete time model conflates higher rates of entry due to low entry costs with higher rates of entry due to low levels of competition.
To see this, consider a two-firm model and suppose the market is empty at time $t = 0$.
Suppose firm 1 first enters at time $t = 0.1$ and firm 2 later decides not to enter at time $t = 0.2$.
Suppose that this market structure remains until time $t=1$.
Although firm 1 has spent nearly the entire unit of time in the market, the discrete time model interprets this as no firms in the market at time $t = 0$ and a firm-1 monopoly beginning at time 1.
Through the discrete time lens, firm 2 chose to not enter an empty market and therefore entry costs must be high.
On the other hand, the continuous time model---even with discrete time data---allows for the possibility that firm 1 entered earlier and firm 2 may have chosen to remain out because of the competitive effect and the presence of firm 1 rather than high entry costs.

Figure~\ref{fig:mis} illustrates the possibility of misleading estimates.
We plot the misspecified discrete time model estimates from Experiment 1 and the true values.
Even though the confidence intervals for the discrete time estimates are small, the estimates are far from the true parameters. This suggests that a researcher should carefully specify the model by considering whether sequential moves are important and understanding how and when the state variables change.

\begin{figure}[htbp]
\centering
\begin{subfigure}[b]{0.45\textwidth}
  \centering
  \includegraphics[width=\linewidth]{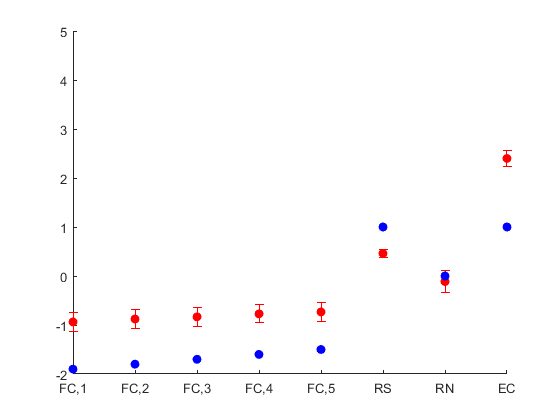}
  \caption{}
\end{subfigure}
\begin{subfigure}[b]{0.45\textwidth}
  \centering
  \includegraphics[width=\linewidth]{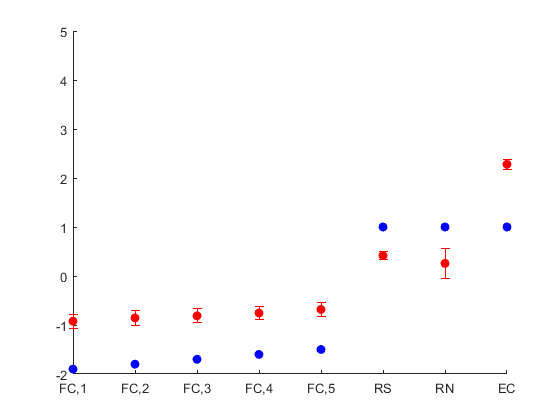}
  \caption{}
\end{subfigure}
\begin{subfigure}[b]{0.45\textwidth}
  \centering
  \includegraphics[width=\linewidth]{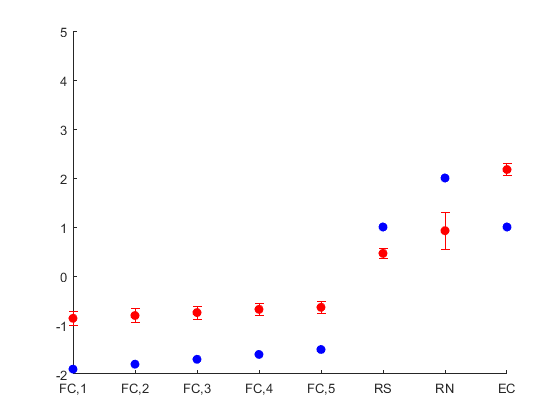}
  \caption{}
\end{subfigure}
\begin{subfigure}[b]{0.45\textwidth}
  \centering
  \includegraphics[width=\linewidth]{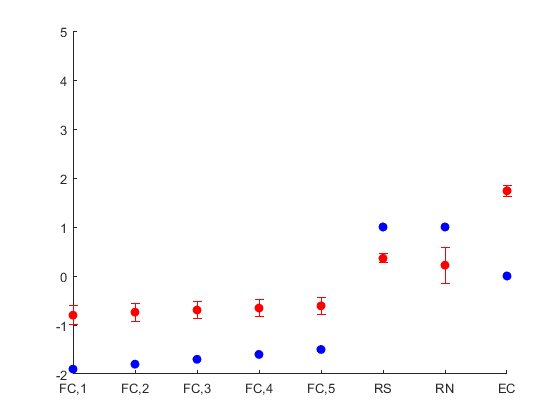}
  \caption{}
\end{subfigure}
\begin{subfigure}[b]{0.45\textwidth}
  \centering
  \includegraphics[width=\linewidth]{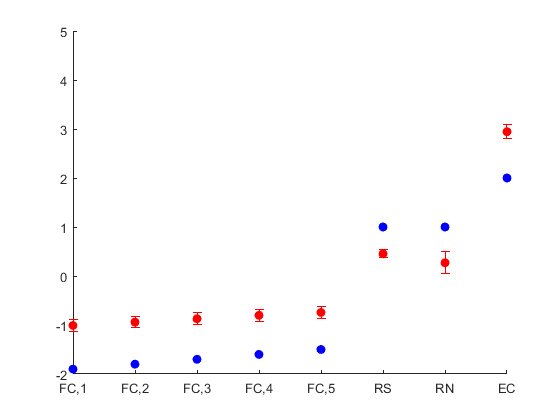}
  \caption{}
\end{subfigure}
\begin{subfigure}[b]{0.45\textwidth}
  \centering
  \includegraphics[width=\linewidth]{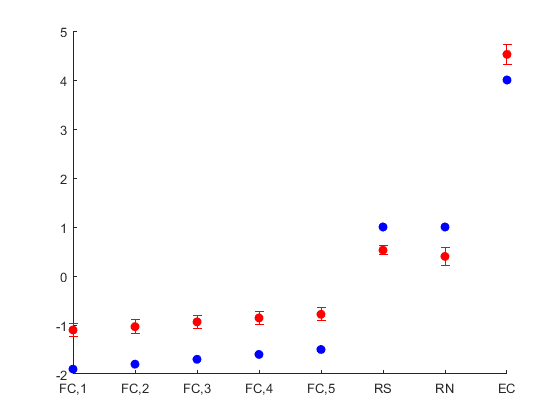}
  \caption{}
\end{subfigure}
\caption{Misspecified Discrete Time Model Estimates}
\label{fig:mis}
\floatfoot{Note: Blue dots are the true parameters from the continuous time model and red dots with bars are the estimates and 95\% confidence intervals from the estimated discrete time model.}
\end{figure}

% Continuous time models allow firms to move at random times in a sequential manner and there will be more entry actions when entry costs decrease.
% Since discrete time models restrict the moves to be simultaneous (and at most once) at prespecified times, the biases will be larger when there are more entries between periods due to lower entry costs.

% In the example, the NPL estimator yields larger bias in entry costs when the model is misspecified as a discrete time model. Continuous time models distinguish between flow payoffs and instantaneous payoffs where as discrete time models assume that agents receive all payoffs at once. Since entry costs are instantaneous payoffs received only when an agent makes a choice, they can cause larger bias when we do not distinguish two different types of payoffs.

\begin{table}[htbp]
\centering
\begin{threeparttable}
\caption{Misspecified Monte Carlo Results (Discrete Time Estimation of Continuous Time DGP)}
\label{table:mis:ctdt:dtdt}
\begin{tabular}{llccccc}
\toprule
        &                      & \multicolumn{4}{c}{Parameters}                                          \\
        \cline{3-6}
Exp. & Values  & $\theta_{\text{FC,1}}$ & $\theta_{\text{RS}}$ & $\theta_{\text{EC}}$ & $\theta_{\text{RN}}$ \\ \midrule
\textbf{1} & \textbf{True values} & \textbf{-1.9000} & \textbf{1.0000} & \textbf{1.0000} & \textbf{0.0000} \\
 & Correct & -1.8830 (0.4711) & 1.0383 (0.1952) & 1.0146 (0.3040) & 0.0747 (0.4094) \\
 & Misspecified & -0.9369 (0.1928) & 0.4604 (0.0762) & 2.3993 (0.1555) & -0.1135 (0.2300) \\
 & Bias & 0.9631 & -0.5396 & 1.3993 & -0.1135 \\
             \addlinespace[0.2cm]
\textbf{2} & \textbf{True values} & \textbf{-1.9000} & \textbf{1.0000} & \textbf{1.0000} & \textbf{1.0000} \\
 & Correct & -1.9558 (0.3573) & 1.0461 (0.1775) & 0.9880 (0.2527) & 1.1051 (0.4033) \\
 & Misspecified & -0.9221 (0.1444) & 0.4176 (0.0889) & 2.2785 (0.1083) & 0.2574 (0.3024) \\
 & Bias & 0.9779 & -0.5824 & 1.2785 & -0.7426 \\
             \addlinespace[0.2cm]
\textbf{3} & \textbf{True values} & \textbf{-1.9000} & \textbf{1.0000} & \textbf{1.0000} & \textbf{2.0000} \\
 & Correct & -1.9521 (0.3397) & 1.0416 (0.1953) & 0.9912 (0.2159) & 2.0869 (0.5427) \\
 & Misspecified & -0.8631 (0.1433) & 0.4615 (0.0999) & 2.1725 (0.1188) & 0.9247 (0.3744) \\
 & Bias & 1.0369 & -0.5385 & 1.1725 & -1.0753 \\
             \addlinespace[0.2cm]
\textbf{4} & \textbf{True values} & \textbf{-1.9000} & \textbf{1.0000} & \textbf{0.0000} & \textbf{1.0000} \\
 & Correct & -1.9443 (0.4427) & 0.9993 (0.1939) & 0.0068 (0.3146) & 0.9524 (0.4843) \\
 & Misspecified & -0.7994 (0.1909) & 0.3607 (0.0887) & 1.7383 (0.1084) & 0.2212 (0.3696) \\
 & Bias & 1.1006 & -0.6393 & 1.7383 & -0.7788 \\
             \addlinespace[0.2cm]
\textbf{5} & \textbf{True values} & \textbf{-1.9000} & \textbf{1.0000} & \textbf{2.0000} & \textbf{1.0000} \\
 & Correct & -1.9285 (0.3105) & 1.0209 (0.1811) & 1.9961 (0.2065) & 1.0190 (0.3562) \\
 & Misspecified & -1.0088 (0.1240) & 0.4522 (0.0797) & 2.9451 (0.1373) & 0.2705 (0.2210) \\
 & Bias & 0.8912 & -0.5478 & 0.9451 & -0.7295 \\
             \addlinespace[0.2cm]
\textbf{6} & \textbf{True values} & \textbf{-1.9000} & \textbf{1.0000} & \textbf{4.0000} & \textbf{1.0000} \\
 & Correct & -1.9921 (0.2694) & 1.0570 (0.1675) & 4.0473 (0.2497) & 1.0799 (0.2934) \\
 & Misspecified & -1.1024 (0.1418) & 0.5281 (0.0937) & 4.5233 (0.2036) & 0.3995 (0.1889) \\
 & Bias & 0.7976 & -0.4719 & 0.5233 & -0.6005    \\ 
\bottomrule
\end{tabular}
    \begin{tablenotes}
      \footnotesize
      \item Displayed values are means with standard deviations in parentheses.
    \end{tablenotes}
\end{threeparttable}
\end{table}  

We can also examine how this problem affects counterfactuals. Dynamic game estimation results are often used for estimating the impact of a policy affecting firms' entry-exit decisions. Suppose that the entry costs are one million dollars and the government is considering subsidizing firms by two hundred thousand dollars. In particular, the policy would reduce the fixed costs by \$200K. Since the correct model and the misspecified model have different estimates for entry costs and fixed costs, the policy will result in different steady state outcomes.

To compare the different counterfactual predictions, we change the parameter values according to the counterfactual policy, recalculate the equilibrium, and compare the results over 50,000 simulations.
Table~\ref{table:cf} presents the average number of active firms before and after subsidization. The first column is the average number of active firms before firms receiving the subsidy. Then, in Experiment 1, if firms receive the subsidy the continuous time model predicts that the number of active firms increase by 5.6\% but the misspecified discrete time model suggests 20.4\% increase in the number of active firms.
The prediction from the discrete time model is four times larger.
This overestimation of the effect when using the misspecified model also occurs for the other experiments.

\begin{table}[htbp]
    \caption{Average Number of Active Firms Before and After Subsidization}
    \label{table:cf}
\resizebox{\linewidth}{!}{\begin{tabular}{@{\extracolsep{4pt}}ccccccccc@{}} 
        \toprule
\multirow{3}{*}{Exp.} & \multicolumn{2}{c}{Before subsidization}                          & \multicolumn{4}{c}{After subsidization}                                     & \multicolumn{2}{c}{$\Delta$\#Active firms (\%)}      \\ \cline{2-3} \cline{4-7} \cline{8-9}
                      & \multirow{2}{*}{Mean} &
    \centering \multirow{2}{*}{(s.d.)} & \multicolumn{2}{c}{Correct} & \multicolumn{2}{c}{Misspecified} & \multirow{2}{*}{Correct} & \multirow{2}{*}{Misspecified} \\
                      \cline{4-5} \cline{6-7}
            
                      &                            &                         & Mean      & (s.d.)     & Mean         & (s.d.)        &                          &              
                      \\
                      \midrule
1          & 3.70          & (1.47)     & 3.91           & (1.33)     & 4.46             & (0.83)        & 5.6                    & 20.4                       \\
2          & 2.77          & (1.53)     & 2.96           & (1.51)     & 3.63             & (1.32)        & 6.7                    & 30.8                       \\
3          & 2.05          & (1.25)     & 2.20           & (1.27)     & 2.68             & (1.37)        & 7.6                    & 31.1                       \\
4          & 2.74          & (1.39)     & -              & -          & -                & -             & -                      & -                          \\
5          & 2.80          & (1.66)     & 3.20           & (1.58)     & 3.94             & (1.23)        & 14.2                   & 40.5                       \\
6          & 2.82          & (1.81)     & 3.67           & (1.48)     & 4.47             & (0.85)        & 30.1                   & 58.6             \\
\bottomrule 
\end{tabular}}
\end{table}

\section{Conclusion}

This paper introduced an NPL estimator for dynamic discrete choice models in continuous time, which we refer to as the CTNPL estimator.
We derived its properties, and demonstrated the performance in a series of Monte Carlo experiments involving a model with five heterogeneous firms.
Specifically, we first showed that the CTNPL estimator is consistent and asymptotically normal both with and without initial consistent and asymptotically normal CCP estimates.
Second, we presented a local convergence condition in the iterative CTNPL algorithm.
Researchers have documented problems regarding the convergence of the NPL algorithm in discrete time.
We showed that the algorithm always converges in continuous-time single agent models and that provided simulation evidence that convergence failures are much less likely to affect comparable continuous-time games.
Third, our Monte Carlo experiments based on those of \cite{AM07} showed that the CTNPL estimator is more robust than two-step estimators initialized from estimates or random draws for the CCPs.
Finally, our results highlight the potential for economically misleading estimates and counterfactuals when estimating discrete time models in cases where events in the data generating process are unfolding in continuous time.

%\newpage

\bibliographystyle{chicago}
\bibliography{bibliography}

%\newpage
\begin{appendix}

\section{Proofs}
\label{sec:proofs}

\subsection{Proof for Proposition \ref{prop:zero}}

We introduce notation from \citetalias{Arcidiacono2016}, prove two preliminary Lemmas, and then prove the zero Jacobian property in a single agent dynamic choice model in continuous time.

For choice $j$ in state $k$, let $v_{jk}=\psi_{jk}+V_{l(j,k)}$ denote an arbitrary choice-specific valuation and $v_k=(v_{0k}, \ldots, v_{J-1,k}) \in \mathbb{R}^J$ denote a $J$-vector of valuations in state $k$. Define $\tilde{v}_k$ as a $(J-1)$-vector of normalized valuations based on choice 0 on the $(J-1)$-dimensional space $\mathcal{V}=\{\tilde{v}_k \in \mathbb{R}^{J-1}: \tilde{v}_{0k}=0\}$. 

Now consider the mapping $\tilde{H}_k: \mathcal{V} \to \mathcal{P}_k$ where $\mathcal{P}_k \in \Delta^{J-1}$ is the space of all CCP vectors $\sigma_k$ in state $k$. By Proposition 1 of \cite{Hotz93}, $\tilde{H}_k$ is one-to-one with $\tilde{H}_k(\sigma_k) = \tilde{v}_k$ and invertible. 

\begin{lemma} Let $\sigma^0_k$ denote an arbitrary vector of $J-1$ choice probabilities in state $k$. 

\begin{equation} \frac{\partial}{\partial \sigma_k} \sum_j\sigma^0_{jk} e_{jk}(\theta, \sigma^0_k) = -\tilde{H}_k^{-1}(\sigma^0_k) \end{equation}

where $e_{jk}(\theta, \sigma)$ is the expected value of $\varepsilon_{jk}$ given that choice $j$ is optimal at state $k$. 
\end{lemma}
\begin{proof}
Let $W_{jk}(\tilde{v})$ represent the expectation $e_{jk}$ as a function of $J-1$ normalized valuations. Since $\tilde{H}(\cdot)$ is invertible, we can write $e_{jk}(\sigma^0) = W(\tilde{H}_k^{-1}(\sigma^0))$. 

\begin{align*}
\sum_{j=1}^J \sigma_j^0 e_{jk} (\theta, \sigma^0)& = \sum_{j=1}^J \sigma_{jk}^0 W_{jk}(\tilde{H}^{-1}(\sigma^0)) \\
& = \sigma^0 \cdot W(\tilde{H}^{-1}(\sigma^0)) \quad \text{in matrix form}
\end{align*}

Differentiating with respect to $\sigma = (\sigma_1, \sigma_2, \ldots, \sigma_J)$, 
\[\frac{\partial}{\partial \sigma} \sigma^0\cdot W(\tilde{H}^{-1}(\sigma^0)) = \frac{\partial \sigma^0}{\partial \sigma}W(\tilde{v}) + \frac{\partial W(\tilde{v})}{\partial \tilde{H}^{-1}(\sigma)}\frac{\partial \tilde{H}^{-1}(\sigma)}{\partial \sigma}\frac{\partial \sigma}{\partial \sigma^0}\sigma^0 
\]
Each term becomes
\begin{enumerate}
\item $\displaystyle\frac{\partial \sigma^0}{\partial \sigma}W(\tilde{v}) = \frac{\partial (\sigma_1, \cdots, \sigma_{J-1})'}{\partial (\sigma_0, \cdots, \sigma_{J-1})'}W(\tilde{v})
= \begin{bmatrix} -1 & -1 & -1 & \cdots & -1 \\ 1& 0 & 0 & \cdots & 0 \\ 0 & 1 & 0 & \cdots & 0\\ \vdots & \vdots & \vdots& \ddots & \vdots \\ 0 & 0 & 0 & \cdots & 1 \end{bmatrix}W(\tilde{v}) = [-i; I]W(\tilde{v})$
\item $\displaystyle\frac{\partial W(\tilde{v})}{\partial \tilde{H}^{-1}(\sigma)}= \frac{\partial W(\tilde{v})}{\partial \tilde{v}}$
\item $\displaystyle \frac{\partial \tilde{H}^{-1}(\sigma)}{\partial \sigma} = \frac{\partial \tilde{v}}{\partial \tilde{H}(\tilde{v})} = \Big[\frac{\partial \tilde{H}(\tilde{v})}{\partial \tilde{v}}\Big]^{-1}$
\item $\displaystyle\frac{\partial \sigma}{\partial \sigma^0}\sigma^0 = [-i; I]\sigma^0 = \begin{bmatrix}-\sigma_1-\cdots - \sigma_{J-1} \\ \sigma_1 \\ \vdots \\ \sigma_{J-1} \end{bmatrix} = \begin{pmatrix}1-i'\sigma^0 \\ \sigma^0\end{pmatrix}$
\end{enumerate}
where $i$ is a $1\times (J-1)$ vector of ones and $I$ is $(J-1) \times (J-1)$ identity matrix. From 1--4, we get an analogous equation to (Ap.1) in \cite{AM02}.
\begin{equation}\label{eq:lemma}\frac{\partial}{\partial \sigma}\sum_{j=1}^J \sigma_{j}^0 e_{j} = [-i; I] W(\tilde{v}) + \Big[\frac{\partial \tilde{H}(\tilde{v})}{\partial \tilde{v}}\Big]^{-1} \frac{\partial W(\tilde{v})'}{\partial \tilde{v}}\begin{pmatrix}1-i'\sigma^0 \\ \sigma^0\end{pmatrix} \end{equation}

From the surplus function, 
\[S(\tilde{v}) = [1-i'\tilde{H}(\tilde{v}); \tilde{H}(\tilde{v})](v+W(\tilde{v})).\]
Differentiating with respect to $\tilde{v}$, 
\begin{align*}
\sigma=\frac{\partial S(\tilde{v})}{\partial \tilde{v}} & = \begin{bmatrix}-i'\frac{\partial \tilde{H}(\tilde{v})}{\partial \tilde{v}} \\ \frac{\partial \tilde{H}(\tilde{v})}{\partial \tilde{v}}\end{bmatrix} [v+W(\tilde{v})] +  \Big[\frac{\partial v}{\partial \tilde{v}}+\frac{\partial W(\tilde{v})}{\partial \tilde{v}} \Big]\begin{bmatrix}1-i'\tilde{H}(\tilde{v}) \\ \tilde{H}(\tilde{v}) \end{bmatrix} \\ 
& = \frac{\partial \tilde{H}(\tilde{v})}{\partial \tilde{v}}\begin{bmatrix} -i \\ I \end{bmatrix}[v+W(\tilde{v})] + \Big[\frac{\partial v}{\partial \tilde{v}}+\frac{\partial W(\tilde{v})}{\partial \tilde{v}} \Big]\begin{bmatrix}1-i'\sigma^0 \\ \sigma^0\end{bmatrix} \\ 
\Big[\frac{\partial \tilde{H}(\tilde{v})}{\partial \tilde{v}}\Big]^{-1}\sigma& = \begin{bmatrix} -i \\ I \end{bmatrix}[v+W(\tilde{v})]+\Big[\frac{\partial \tilde{H}(\tilde{v})}{\partial \tilde{v}}\Big]^{-1} \Big[\frac{\partial v}{\partial \tilde{v}}+\frac{\partial W(\tilde{v})}{\partial \tilde{v}} \Big]\begin{bmatrix}1-i'\sigma^0\\ \sigma^0\end{bmatrix} \\ 
& = \begin{bmatrix} -i \\ I \end{bmatrix}v+\Big[\frac{\partial \tilde{H}(\tilde{v})}{\partial \tilde{v}}\Big]^{-1} \frac{\partial v}{\partial \tilde{v}}\begin{bmatrix}1-i'\sigma^0\\ \sigma^0\end{bmatrix} + \begin{bmatrix}-i\\I\end{bmatrix} W(\tilde{v}) + \Big[\frac{\partial \tilde{H}(\tilde{v})}{\partial \tilde{v}}\Big]^{-1} \frac{\partial W(\tilde{v})'}{\partial \tilde{v}}\begin{bmatrix}1-i'\sigma^0 \\ \sigma^0\end{bmatrix} \\
& = \begin{bmatrix} -i \\ I \end{bmatrix}v+\Big[\frac{\partial \tilde{H}(\tilde{v})}{\partial \tilde{v}}\Big]^{-1} \begin{bmatrix}-i\\I\end{bmatrix}\begin{bmatrix}1-i'\sigma^0\\ \sigma^0\end{bmatrix} +\frac{\partial}{\partial \sigma}\sum_{j=1}^J \sigma_{j}^0 e_{j} \ \text{from \eqref{eq:lemma}}  \\
& = \begin{bmatrix} -i \\ I \end{bmatrix}v+\Big[\frac{\partial \tilde{H}(\tilde{v})}{\partial \tilde{v}}\Big]^{-1} \sigma + \frac{\partial}{\partial \sigma}\sum_{j=1}^J \sigma_j^0 e_j 
\end{align*}
Rearranging, we have the stated result:
\[\frac{\partial}{\partial \sigma}\sum_{j=1}^J \sigma_{j}^0 e_{j}(\theta, \sigma^0) = -\tilde{v} = -\tilde{H}^{-1}(\sigma^0).\]
\end{proof}

\begin{lemma}
Let $\sigma^0$ be an arbitrary set of conditional choice probabilities, and define the mapping 
\[G(\sigma^0, V) = [(\rho+\lambda)I-\lambda\Sigma(\sigma)-Q_0]^{-1}[\lambda \Sigma(\sigma)V+u + \lambda E(\theta, \sigma^0)]. \]
Let $\sigma_k^0$ be $(J-1)$ column vector of CCPs in $\sigma^0$ associated with state $k$. Then (i) $\partial G_l/ \partial \sigma_k^0=0$ for $l \neq k$ and (ii) $\partial G_k/\partial \sigma_k^0=[(\rho+\lambda)I-\lambda\Sigma(\sigma)-Q_0]^{-1}\big[\lambda[\tilde{v}-\tilde{H}^{-1}_k(\sigma^0)]]$.
\end{lemma}
\begin{proof}
(i) Given that $G_k$ does not depend on the probabilities for states different to $k$, (i) follows trivially. 
\\ (ii) Using the results from Lemma 1 and denoting $\tilde{v}_k$ as $(K-1)$ column vector of differenced values corresponding to state $k$, 
\begin{align*}
\frac{\partial G_k}{\partial \sigma_k^0} &=[(\rho+\lambda)I-\lambda\Sigma(\sigma)-Q_0]^{-1}\Big[ \lambda\frac{\partial \Sigma_k(\theta, \sigma)}{\partial \sigma_k^0} V_k + \frac{\partial}{\partial \sigma} \sum_{j=1}^J \sigma^0_{jk}e_{jk}  \Big]\\
& = [(\rho+\lambda)I-\lambda\Sigma(\sigma)-Q_0]^{-1}\Big[\lambda [-i; I] V_{l(j,k)} - \lambda \tilde{H}_k^{-1}(\sigma^0)\Big] \\
& = [(\rho+\lambda)I-\lambda\Sigma(\sigma)-Q_0]^{-1}\big[\lambda[\tilde{v}_k - \tilde{H}_k^{-1}(\sigma^0)\big]]
\end{align*}
\end{proof}

We present the proof for Proposition \ref{prop:zero} below.

\begin{proof}
Rearranging \eqref{eq:value} and using the value function operator,
\[\Upsilon(\theta, \sigma) = [(\rho+\lambda)I-\lambda\Sigma(\sigma)-Q_0]^{-1}[\lambda \Sigma(\sigma)\Upsilon(\theta, \sigma)+u + \lambda E(\sigma)]\]
Therefore, $\Upsilon(\theta, \sigma)=G(\sigma, \Upsilon(\theta, \sigma))$. Differentiating on both sides with respect to $\sigma$, 
\begin{align*}
\frac{\partial \Upsilon(\sigma)}{\partial \sigma} & = \frac{\partial G(\cdot)}{\partial \sigma}+\frac{\partial G(\cdot)}{\partial \Upsilon(\cdot)}\frac{\partial \Upsilon(\sigma)}{\partial \sigma} \\
\Big[I-\frac{\partial G(\cdot)}{\partial \Upsilon(\cdot)}\Big]\frac{\Upsilon(\sigma)}{\partial \sigma} & = \frac{\partial G(\cdot)}{\partial \sigma} \\
\frac{\partial \Upsilon(\sigma)}{\partial \sigma} & =\Big[I-\frac{\partial G(\cdot)}{\partial \Upsilon(\cdot)}\Big]^{-1} \frac{\partial G(\cdot)}{\partial \sigma} 
\end{align*}

By Lemma 2, $\partial G_l/\partial \sigma_k=0$ for $l \neq k$ and $\partial G_k/ \partial \sigma_k=[(\rho+\lambda)I-\lambda\Sigma(\sigma)-Q_0]^{-1}\big[\lambda[\tilde{v}_k-\tilde{H}^{-1}_k(\sigma)]\big]$. In the latter case, let $\tilde{\sigma}$ be the fixed point of $\Upsilon(\theta, \sigma)$. Then $\tilde{v}_k = \tilde{H}^{-1}(\Upsilon(\theta, \tilde{\sigma}))=\tilde{H}^{-1}(\tilde{\sigma})$. So at the fixed point $\tilde{v}_k = \tilde{H}_k^{-1}(\sigma)$ which means $\partial G_k/\partial \sigma_k=0$, and hence, $\partial G(\cdot)/\partial \sigma^0=0$. Then, at the fixed point $\sigma^0$, $\partial \Upsilon(\sigma^0)/\partial \sigma^0=0$. 

Since $\Gamma(v)$ is continuous, and $\Psi(\theta, \sigma) = \Gamma(\Upsilon(\theta, \sigma))$, $\Psi_\sigma=0$ at the fixed point $\sigma$. 
\end{proof}

\subsection{Proof of Proposition~\ref{prop:single}}

We provide the proof for Proposition~\ref{prop:single} for single agent models.
The proof proceeds along similar lines as the proof of Proposition 4 in \cite{AM02}.
Let $g(k_m, \theta, \sigma) = \sum_{n=1}^T \ln P_{k_{m,n-1},k_{m,n}}(\Delta;\Psi(\theta,\sigma))$ for $k_m = (k_{m,0},k_{m,1}, \dots, k_{m,T})$ denote the likelihood function for an individual market $m$ so that the sample log likelihood function can be written as $L_M(\theta,\sigma) = M^{-1} \sum_{m=1}^M g(x_m, \theta, \sigma)$.

\textbf{Consistency.} First, we prove the consistency of $\hat{\theta}^{\npliter}$.

\textit{Step 1. If $\hat{\sigma}^{\npliter-1}$ is consistent, then $L_M(\theta, \hat{\sigma}^{\npliter-1})$ converges a.s. and uniformly in $\theta$ to a deterministic function $L(\theta, \sigma^*)$.}

By Lemma 24.1 of \cite{Gourieroux1995}, we have that if (i) $L_M(\theta, \sigma)$ converges a.s. and uniformly in $(\theta, \sigma)$ to $L(\theta, \sigma)$; (ii) $L(\theta, \sigma)$ is uniformly continuous in $(\theta, \sigma)$; and (iii) $\{\hat{\sigma}^{\npliter-1}\}$ converges a.s. to $\sigma^*$; then $L_M(\theta, \hat{\sigma}^{\npliter-1})$ converges a.s. and uniformly in $\theta$ to $L(\theta, \sigma^*)$.

To establish condition (i), we invoke the uniform strong law of large numbers of \cite[][Theorem 2]{Jennrich1969}.  This requires that (1) The data are i.i.d.; (2) $\Theta$ is compact.; (3) $g(k, \theta, \sigma)$ is continuous at each $(\theta, \sigma)$ for each $k$ and a measurable function of $k$ for each $(\theta,\sigma)$; (4) $\lvert g(k,\theta,\sigma) \rvert \leq h(k)$ for all $k$ and $(\theta,\sigma)$ where $\E[h(k)]<\infty$.
Condition (1) follows from our maintained assumption; condition (2) follows from assumption (a); condition (3) follows from assumption (b) and the fact that the matrix exponential preserves continuity; and condition (4) follows because the transition probability matrix for an irreducible continuous-time Markov jump process has elements strictly bounded between 0 and 1.
% To justify (4), see Proposition 6.1 of the following notes:
% https://galton.uchicago.edu/~lalley/Courses/313/ContinuousTime.pdf
% x-devonthink-item://D705CBAB-D084-445A-92D0-6989ECCC6AA2

To establish condition (ii), note as for (3) above that $\Psi(\theta, \sigma)$ is continuous in $\theta$ by condition (b) and the matrix exponential operation preserves continuity.
Finally, condition (iii) holds by assumption (e).

\textit{Step 2. If $\hat{\theta}^{\npliter-1}$ is consistent, then $\hat{\theta}^{\npliter} \equiv \argmax_{\theta \in \Theta}L_M(\theta, \hat{\sigma}^{\npliter-1})$ converges a.s. to $\theta^*$.}

Let $\hat{\sigma}^{\npliter} = \Psi(\argmax_{\theta \in \Theta}L_M(\theta, \hat{\sigma}^{\npliter-1}), \hat{\sigma}^{\npliter-1})$. By Property 24.2 in \cite[][p.392]{Gourieroux1995}, if (i) $L_M(\theta, \sigma)$ converges a.s. and uniformly in $\theta$ to $L(\theta, \sigma)$ and (ii) $L(\theta,\sigma)$ has a unique maximum in $\Theta$ at $\theta^*$, then $\hat{\theta}^{\npliter}\equiv \argmax_{\theta \in \Theta} L_M(\theta, \hat{\sigma}^{\npliter-1})$ converges a.s. to $\theta^*$. Condition (i) was established above in Step 1.  Condition (ii) follows from the identification assumption (d), which implies that $\theta^*$ is the only element of $\Theta$ that maximizes $L_M(\theta, \hat{\sigma}^{\npliter-1})$, and the information inequality.

\textit{Step 3. For $\npliter \geq 1$, if $\hat{\sigma}^{\npliter-1}\xrightarrow{\text{a.s.}}\sigma^*$, then $\hat{\sigma}^{\npliter} \xrightarrow{\text{a.s.}} \sigma^*$.}

By definition, $\hat{\sigma}^{\npliter} = \Psi(\hat{\theta}^{\npliter}, \hat{\sigma}^{\npliter-1})$. By Step 2, $\hat{\theta}^{\npliter} \xrightarrow{\text{a.s.}} \theta^*$. Since $\Psi$ is continuous in $(\theta, \sigma)$, by Slutsky's theorem, $\hat{\sigma}^{\npliter} \xrightarrow{\text{a.s.}} \sigma^*$.
By assumption (e), $\hat{\sigma}^0 \xrightarrow{\text{a.s.}} \sigma^*$, then by induction on step 2 and 3, we have the result.

\textbf{Asymptotic normality.}
Now, we prove the asymptotic normality of $\hat{\theta}$. To apply an induction argument, we show that if $(\hat{\theta}^{\npliter}, \hat{\sigma}^{\npliter-1}$) is $\sqrt{M}$-consistent and asymptotic normal, then $\sqrt{M}(\hat{\theta}^{\npliter}-\theta^*) \xrightarrow{\text{d}} \Normal(0, \Omega_{\theta \theta'}^{-1})$, and $\hat{\sigma}^{\npliter} \equiv \Psi(\hat{\theta}^{\npliter}, \hat{\sigma}^{\npliter-1})$ is also $\sqrt{M}$-consistent and asymptotic normal.

\textit{Step 1. $\sqrt{M}(\hat{\theta}^{\npliter}-\theta^*) \xrightarrow{\text{d}} \Normal(0, V^*)$ and $V^*$ only depends on the upper left $r \times r$ submatrix of $\Omega$ where $r$ is the dimension of the parameter vector $\theta$.}

First, assume that $\hat{\sigma}^{\npliter-1}$ is a consistent estimator of $\sigma^*$ such that $[\sqrt{M}\nabla_\theta L_M(\theta^*, \sigma^*); \sqrt{M}(\hat{\sigma}^{\npliter-1}-\sigma^*)']' \xrightarrow{d} \Normal(0, \Omega)$. Given assumptions (b) and (d) and the definition of $\hat{\theta}^{\npliter}$, the first order conditions of optimality imply that $\nabla_\theta L_M(\hat{\theta}^{\npliter}, \hat{\sigma}^{\npliter-1})=0$. Since $L_M(\theta, \sigma)$ is twice continuously differentiable, we can apply the mean value theorem:
\[0 = \nabla_\theta L_M(\theta^*, \sigma^*) + \nabla_{\theta \theta'}L_M(\theta^*, \sigma^*)(\hat{\theta}^{\npliter}-\theta^*) +\nabla_{\theta \sigma'}L_M(\theta^*, \sigma^*)(\hat{\sigma}^{\npliter-1}-\sigma^*) + o_p(1)\]
Note that  $\nabla_{\theta \theta'} L_M(\theta^*, \sigma^*) \xrightarrow{\text{p}} - \Omega_{\theta \theta'}$ and $\nabla_{\theta \sigma'} L_M(\theta^*, \sigma^*) \xrightarrow{\text{p}} -\Omega_{\theta \sigma'}$, where $\Omega_{\theta \sigma'}\equiv \E[\nabla_\theta s_m \nabla_{\sigma'} s_m]$ by Theorems 4.2.1 and 4.1.5 of \cite{Amemiya1985}. Rearranging the equation above,  
\[\sqrt{M}(\hat{\theta}^l - \theta^*) = \Omega_{\theta\theta'}^{-1} \big[\sqrt{M} \nabla_\theta L_M(\theta^*, \sigma^*) - \Omega_{\theta \sigma'}\sqrt{M}(\hat{\sigma}^{l-1}-\sigma^*)\big] + o_p(1) \]
which leads to 
\[\sqrt{M}(\hat{\theta}^{\npliter}-\theta^*) \xrightarrow{\text{d}} \Normal(0, V^*)\]
where
\[V^* = \Omega_{\theta \theta'}^{-1}[I; \Omega_{\theta \sigma'}] \Omega  [I;\Omega_{\theta \sigma'}]' \Omega_{\theta \theta'}^{-1}\]
Note that by the zero Jacobian property of $\Psi_\sigma$ in Proposition~\ref{prop:zero},
\begin{align*}
    \nabla_\sigma L_M(\theta^*, \Psi(\theta^*, \sigma^*)) & = \frac{\partial L_M(\cdot)}{\partial \Psi(\cdot)}\frac{\partial \Psi(\theta^*, \sigma^*)}{\partial \sigma} \\ 
    & = 0 
\end{align*}
By the information matrix equality,
\[\nabla_{\theta \sigma'} L(\theta^*, \sigma^*) = \E\Big[\frac{\partial L_M(\theta^*, \sigma^*)}{\partial \theta}\frac{\partial L_M(\theta^*, \sigma^*)}{\partial \sigma'}\Big] = 0\]
and this leads to $\Omega_{\theta \sigma}=0$.
Then, 
\begin{align*}
V^* &= \Omega_{\theta \theta'}^{-1}\Omega_{r\times r} \Omega_{\theta \theta'}^{-1} \\
 & = \Omega_{\theta \theta'}^{-1}
 \end{align*}
where $\Omega_{r \times r}$ is the upper left $r \times r$ submatrix of $\Omega$ and the second equality comes from the fact that 
\begin{align*}
    \Omega_{r \times r} & = \var(\sqrt{M} \nabla_\theta L_M(\theta^*, \sigma^*)) \\
    & = M \var(\nabla_\theta L_M(\theta^*, \sigma^*)) \\ 
    & = M \var(\frac{1}{M}\sum_{m=1}^M\nabla_\theta s_m) \\
    & = \frac{1}{M}\sum_{m=1}^M \var(\nabla_\theta s_m) \\
    & = \frac{1}{M}\sum_{m=1}^M \E[\nabla_\theta s_m \nabla_{\theta'} s_m] \\
    & = \Omega_{\theta \theta'}.
\end{align*}

\textit{Step 2. $[\sqrt{M}\nabla_\theta L_M(\theta^*, \sigma^*); \sqrt{M}(\hat{\sigma}^{\npliter}-\sigma^*)']'\xrightarrow{\text{d}} \Normal(0, \Omega^*)$, and the upper left $r \times r$ submatrices of $\Omega$ and $\Omega^*$ are identical.}

Define 
\[\omega_M^{\npliter}\equiv [\sqrt{M}\nabla_\theta L_M(\theta^*, \sigma^*); \sqrt{M}(\hat{\sigma}^{\npliter}-\sigma^*)'\Big]'. \]
We know that $\hat{\sigma}^{\npliter+1} = \Psi(\hat{\theta}^{\npliter+1}, \hat{\sigma}^{\npliter})$ and $\sigma^* = \Psi(\theta^*, \sigma^*)$. We apply mean value expansion to ($\hat{\sigma}^{\npliter+1} - \sigma^*$):
\begin{align*}
    \hat{\sigma}^{\npliter+1} - \sigma^* &= \Psi(\hat{\theta}^\npliter, \hat{\sigma}^{\npliter} )- \Psi(\theta^*, \sigma^*) \\ & = \Psi_\theta(\theta^*, \sigma^*)(\hat{\theta}^{\npliter+1} - \theta^*) + \Psi_\sigma(\theta^*, \sigma^*)(\hat{\sigma}^{\npliter}-\sigma^*)+o_p(1)
\end{align*}
Then, from the mean value expansion in step 1, we have $\sqrt{M}(\hat{\theta}^\npliter- \theta^*) = \Omega_{\theta \theta'}^{-1}\sqrt{M}\nabla_\theta L_M(\theta^*, \sigma^*) +o_p(1)$, so we substitute into the equation above,
\[\sqrt{M}(\hat{\sigma}^{\npliter+1} - \sigma^*) = \Psi_\theta(\theta^*, \sigma^*)\Omega_{\theta\theta'}^{-1}\sqrt{M}\nabla_\theta L_M(\theta^*, \sigma^*) + \Psi_\sigma(\theta^*, \sigma^*)\sqrt{M}(\hat{\sigma}^{\npliter}-\sigma^*)+o_p(1).\]
We rewrite in matrix form:
\begin{equation}
\label{eq:matrix}
\begin{bmatrix} \sqrt{M}\nabla_\theta L_M(\theta^*, \sigma^*) \\ \sqrt{M} (\hat{\sigma}^{\npliter+1} - \sigma^*)\end{bmatrix}
= \begin{bmatrix} I & 0 \\ \Psi_\theta(\theta^*, \sigma^*)\Omega_{\theta \theta'}^{-1} & \Psi_\sigma(\theta^*, \sigma^*) \end{bmatrix} 
\begin{bmatrix} \sqrt{M}\nabla_\theta L_M(\theta^*, \sigma^*) \\ \sqrt{M} (\hat{\sigma}^{\npliter}- \sigma^*)\end{bmatrix} + o_p(1).
\end{equation}
So, we have $\omega_M^{\npliter+1} = A \omega_M^{\npliter+1}$ where $A$ is the first matrix in the RHS of \eqref{eq:matrix}. It follows that if $\omega_M^{\npliter}$ is asymptotically normal, $\omega_M^{\npliter+1}$ is also asymptotically normal. From Step 1, we also know that $V^*$ depends on upper $r \times r$ submatrix of $\Omega$ and the upper-left $r\times r$ matrix of $A$ is the identity matrix from \eqref{eq:matrix}. Therefore, $r \times r$ submatrices of $\Omega$ and $\Omega^*$ are equal. Since $\sqrt{M}(\hat{\theta}^{\npliter}-\theta^*)=\Omega_{\theta \theta'}^{-1}\sqrt{M}\nabla_\theta L_M(\theta^*, \sigma^*) + o_p(1)$, and denoting $\sqrt{M}(\hat{\theta}^{\npliter}-\theta^*) \xrightarrow{d} \Normal(0, V^{**})$, we have
\begin{align*}
     V^{**} &= \Omega_{\theta \theta'}^{-1}\Omega^*_{r \times r} \Omega_{\theta \theta'}^{-1} \\ 
     & = \Omega_{\theta \theta'}^{-1}\Omega_{r \times r} \Omega_{\theta \theta'}^{-1} \\
     & = \Omega_{\theta\theta'}^{-1}.
\end{align*}

From step 1, we showed that when $\hat{\sigma}^{\npliter-1}$ is a consistent estimator of $\sigma^*$ s.t. $\omega^{\npliter} \xrightarrow{d} \Normal(0, \Omega)$, it follows that $\sqrt{M}(\hat{\theta}^\npliter- \theta^*) \xrightarrow{d} \Normal(0, \Omega_{\theta \theta'}^{-1})$. Step 2 proves that $\omega^{\npliter+1}\xrightarrow{d} \Normal(0, \Omega^*)$ which implies $\sqrt{M}(\hat{\theta}^{\npliter+1} - \theta^*) \xrightarrow{d} \Normal(0, \Omega_{\theta \theta'}^{-1})$.   Since we start from asymptotic normal estimator $\hat{\sigma}^0$ such that $\omega^0 \xrightarrow{d} \Normal(0, \Omega)$, this is true for all $\npliter \leq L-1$ by induction. 

\subsection{Proof of Proposition~\ref{prop:largegame}}

We now extend our results for single agent models to the case of games and prove the large sample properties of the CTNPL estimator stated in Proposition~\ref{prop:largegame}.

\begin{proof}
  \textbf{Consistency.} 
Before proving the main result, we first establish that for an arbitrary consistent estimator $\tilde{\sigma}$ for $\sigma^*$, the estimator $\tilde{\theta} = \arg\max_{\theta\in\Theta} L_M(\theta,\tilde\sigma)$ is strongly consistent for $\theta^*$.  The argument is the same as in Step 1 of the proof of Proposition~\ref{prop:single}. First, by the strong uniform law of large numbers of \cite{Jennrich1969} we have a.s. convergence of $L_M$ to $L$ uniformly in $(\theta,\sigma)$. Second, we use Lemma 24.1 of \cite{Gourieroux1995} to establish a.s. convergence of $L_M(\theta, \tilde{\sigma})$ to $\theta$ to $L(\theta, \sigma^*)$ uniformly in $\theta$. Finally, strong consistency of the extremum estimator $\tilde\theta$ follows from Property 24.2 of \cite[][p.392]{Gourieroux1995}.

Now, we prove the main result $(\hat{\theta}^\npliter, \hat{\sigma}^{\npliter-1}) \asto (\theta^*,\sigma^*)$ by induction. For $\npliter=1$, $\hat{\sigma}^0 \asto \sigma^*$ follows directly from assumption (e), that we begin with a strongly consistent estimator $\hat{\sigma}^0$ for $\sigma^*$. To show strong consistency of $\hat{\theta}^1$ starting from $\hat{\sigma}^0$, we take $\hat{\sigma}^0$ to be the consistent estimator $\tilde{\sigma}$ above which yields strong consistency of $\hat{\theta}^1$ for $\theta^*$ at iteration $l = 1$.

Now assume that $(\hat{\theta}^\npliter, \hat{\sigma}^{\npliter-1}) \asto (\theta^*,\sigma^*)$ for all $\npliter \leq L-1$.
Then, $\hat{\sigma}^l = \Psi(\hat{\theta}^\npliter, \hat{\sigma}^{\npliter-1}) \asto \Psi(\theta^*, \sigma^*) = \sigma^*$.
Using $\hat{\sigma}^l$ as our consistent estimator $\tilde\sigma$, as above, yields $\hat\theta^{l+1} \asto \theta^*$.
Since we start from a strongly consistent estimator $\hat{\sigma}^0$, we have strong consistency of $\hat{\theta}^{\npliter}$ for all $\npliter \leq \Npliter$ by induction.

\textbf{Asymptotic normality.}
As with consistency, we prove that
$\hat{\sigma}^{\npliter-1} \dto \Normal(0, \Sigma^{\npliter-1})$
and
$\hat{\theta}^\npliter \dto \Normal(0, \Omega_{\theta \theta'} + \Omega_{\theta \theta'}\Omega_{\theta \sigma}\Sigma^{\npliter-1}\Omega_{\theta \sigma'} \Omega_{\theta \theta'}^{-1}$ by induction.
We first show that the result holds for $\npliter = 1$.
As in assumption (e), let $\hat\sigma^0 = \frac{1}{M}\sum_{m=1}^M r_m$ denote the initial nonparametric estimator for $\sigma^*$.
Let $\nabla_\theta s_m^\npliter = \nabla_\theta \ln P_{k_{m,n-1}, k_m,n}(\Delta; \Psi(\theta, \hat{\sigma}^{\npliter-1}))$ denote the pseudo score at iteration $\npliter$ based on estimated CCPs from iteration $\npliter-1$.
We expand the first order condition $\nabla_\theta L_M(\hat{\theta}^1, \hat{\sigma}^0) = 0$ around $(\theta^*,\sigma^*)$:
\begin{equation*}
%\label{eq:normality1}
0 = \nabla_\theta L_M(\theta^*, \sigma^*) + \nabla_{\theta \theta'}L_M(\theta^*, \sigma^*)(\hat{\theta}^1-\theta^*) + \nabla_{\theta \sigma'} L_M(\theta^*, \sigma^*)(\hat{\sigma}^0 - \sigma^*) + o_p(1)
\end{equation*}
By the generalized information matrix equality, $\nabla_{\theta \theta'} L_M(\theta^*, \sigma^*) \xrightarrow{p} -\Omega_{\theta \theta'}$ and $\nabla_{\theta \sigma'} L_M(\theta^*, \sigma^*) \xrightarrow{p} - \Omega_{\theta \sigma'}$. Then, 
\begin{equation}\label{eq:asym0}\sqrt{M}(\hat{\theta}^1 - \theta^*) = \Omega_{\theta \theta'}^{-1}\Big\{-\Omega_{\theta \sigma'} \Big(\frac{1}{\sqrt{M}} \sum_{m=1}^M \nabla_\sigma r_m \Big) + \Big(\frac{1}{\sqrt{M}}\sum_{m=1}^M \nabla_\theta s_m^0\Big)\Big\}+o_p(1).\end{equation}

From condition (b), $\Psi(\theta, \sigma)$ is continuously differentiable and all elements of $\exp(Q(\theta, \sigma))$ are strictly bounded between 0 and 1 since the Markov chain is irreducible. Then, by the generalized information matrix equality, $\E[\nabla_\sigma r_m \nabla_{\theta'} s_m^0] = 0$.%
%, and $\E[\nabla_\sigma r_m \nabla_\sigma s_m^0] = I$.
Then,
\[\frac{1}{\sqrt{M}}\Big(\sum_{m=1}^M \nabla_\theta s_m^0 \Big) - \Omega_{\theta \sigma'} \Big(\frac{1}{\sqrt{M}} \sum_{m=1}^M \nabla_\sigma r_m\Big) \xrightarrow{d} \Normal(0, \Omega_{\theta \theta'}+ \Omega_{\theta \sigma'}\Sigma^0 \Omega_{\theta \sigma'}').\]
Substituting into equation (\ref{eq:asym0}), we have asymptotic normality for $\hat{\theta}^1$. 
\[\sqrt{M}(\hat{\theta}^1 - \theta^*) \xrightarrow{d}\Normal(0, \Omega_{\theta \theta'}^{-1}+\Omega_{\theta\theta'}^{-1}\Omega_{\theta \sigma'}\Sigma^0 \Omega_{\theta \sigma'}'\Omega_{\theta \theta'}^{-1}).\]

Similarly, we can expand the first order condition $\nabla_\sigma L_M(\hat{\theta}^1, \hat{\sigma}^0) = 0$ around $(\theta^*,\sigma^*)$:
\begin{equation*}
%\label{eq:normality1}
0 = \nabla_\theta L_M(\theta^*, \sigma^*) + \nabla_{\theta \theta'}L_M(\theta^*, \sigma^*)(\hat{\theta}^1-\theta^*) + \nabla_{\theta \sigma'} L_M(\theta^*, \sigma^*)(\hat{\sigma}^0 - \sigma^*) + o_p(1)
\end{equation*}

Now we assume that $\sqrt{M}(\hat{\sigma}^{\npliter-1} - \sigma^*) \xrightarrow{d} \Normal(0, \Sigma^{\npliter-1})$ and $\sqrt{M}(\hat{\theta}^\npliter - \theta^*) \xrightarrow{d}\Normal(0, \Omega_{\theta \theta'}^{-1}+\Omega_{\theta\theta'}^{-1}\Omega_{\theta \sigma'}\Sigma^{\npliter-1} \Omega_{\theta \sigma'}'\Omega_{\theta \theta'}^{-1})$ for some $\npliter \leq \Npliter-1$.
% FIXME: Establish distribution of $\hat\sigma^{l}$ first...
We can consider $\hat{\theta}^{\npliter+1}$ as a two step estimator based on the moment equation $\E[\nabla_\theta L_M(\theta, \hat{\sigma}^{\npliter})]=0$. We use the preliminary estimate $\hat{\sigma}^{\npliter}$ of $\sigma^*$ from the previous iteration based on the moment equation, $\E[\Psi(\hat{\theta}^{\npliter},\hat{\sigma}^{\npliter-1})-\sigma]=0$. We can write the moment conditions as $\tilde{g}(\theta, \sigma)$. 
% \[
%   \tilde{g}(\theta, \sigma) = \begin{bmatrix}
%     g(\theta, \sigma) \\
%     h(\sigma)
%   \end{bmatrix} = \begin{bmatrix}
%     \nabla_\theta L_M(\theta, \hat{\sigma}^{\npliter}) \\
%     \Psi(\theta, \hat{\sigma}^{\npliter}) - \sigma
%   \end{bmatrix}.
% \]
\[
  \tilde{g}(\theta, \sigma) = \begin{bmatrix}
    g(\theta, \sigma) \\
    h(\sigma)
  \end{bmatrix} = \begin{bmatrix}
    \nabla_\theta L_M(\theta, \sigma) \\
    \Psi(\hat\theta^{\npliter}, \hat{\sigma}^{\npliter-1}) - \sigma
  \end{bmatrix}.
\]
Applying the generalized information matrix inequality,
\[\E[\nabla_\theta \tilde{g}(\theta,\sigma)] = -\E[\tilde{g}(\theta, \sigma) \nabla_\theta s_m^{\npliter+1 \prime}] \]

This results in $\E[\nabla_\theta h(\sigma)]=-\E[\nabla_\theta (\Psi_m(\hat{\theta}^{\npliter}, \sigma)- \sigma^*)\nabla_\theta s_m^{\npliter+1 \prime}]=0$. As before, we expand the first moment condition $\nabla_\theta L_M(\hat{\theta}^{\npliter+1}, \hat{\sigma}^{\npliter})=0$ using the intermediate value theorem.
\begin{equation}
\label{eq:normality2} 0= \nabla_\theta L_M(\theta^*, \sigma^*) + \nabla_{\theta \theta'}L_M(\theta^*, \sigma^*)(\hat{\theta}^{\npliter+1}-\theta^*) + \nabla_{\theta \sigma'} L_M(\theta^*, \sigma^*)(\hat{\sigma}^\npliter - \sigma^*) + o_p(1).
\end{equation}
Then, $\sqrt{M}(\hat{\theta}^{\npliter+1}-\theta^*)$ can be written as 
\begin{equation}
\label{eq:asym1}
\sqrt{M}(\hat{\theta}^{\npliter+1} - \theta^*) = \Omega_{\theta \theta'}^{-1}\left\{-\Omega_{\theta \sigma'} \sqrt{M}\left(\hat\sigma^{\npliter} - \sigma^*\right) + \left(\frac{1}{\sqrt{M}}\sum_{m=1}^M \nabla_\theta s_m^{\npliter+1}\right)\right\}+o_p(1)
\end{equation}
where $\sqrt{M}(\hat{\sigma}^{\npliter}-\sigma^*) \xrightarrow{d} \Normal(0, \Sigma^{\npliter})$.
Then, from equation (\ref{eq:asym1}),
\begin{equation*}
\sqrt{M}(\hat{\theta}^{\npliter+1} - \theta^*) \xrightarrow{d}\Normal(0, \Omega_{\theta \theta'}^{-1}+\Omega_{\theta\theta'}^{-1}\Omega_{\theta \sigma'}\Sigma^{\npliter} \Omega_{\theta \sigma'}'\Omega_{\theta \theta'}^{-1}).
\end{equation*}
The recursive expression for $\Sigma^l$ stated in the proposition follows by expanding $\hat\sigma^\npliter - \Psi(\hat\theta^\npliter, \hat\sigma^{\npliter-1})$ around $(\theta^*, \sigma^*)$ to obtain
\begin{equation*}
\hat\sigma^\npliter - \sigma^* = \nabla_\theta \Psi^* (\hat\theta^\npliter - \theta^*) + \nabla_\sigma \Psi^* (\hat\sigma^{\npliter-1} - \sigma^*) + o_p(1),
\end{equation*}
and using the asymptotic normality of both $\hat\theta^\npliter$ and $\hat\theta^\npliter$ from the inductive hypothesis.
\end{proof}

\subsection{Proof of Proposition~\ref{prop:largegame2}}

We omit the detailed proof of Proposition~\ref{prop:largegame2} as it is essentially the same as the proof in the Appendix of \cite{AM07}, since the CTNPL operator has the same required properties as the NPL operator. Changing $Q_0(\theta,P)$ to $L(\theta, \sigma)$ and $P$ to $\sigma$ will be sufficient to prove an analogous proposition for the CTNPL estimator.

% \begin{comment}
% We define a population counterpart of $Y_M$ as $Y^* = \{(\theta, \sigma) \in \Theta \times [0,1]^{NJK}: \theta = \tilde{\theta}^*(\sigma) \ \text{and} \ \sigma=\phi^*(\sigma)\}$ where $\phi^*(\sigma) \equiv \Psi(\tilde{\theta}^*(\sigma), \sigma)$ and $\tilde{\theta}^*(\sigma) \equiv \argmax_{\theta \in \Theta}L(\theta, \sigma)$.  

% \begin{proof} 

% \textit{Step 1.} $(\theta^*, \sigma^*)$ uniquely maximizes $L(\theta, \sigma)$ in $Y^*$. \\
% First, $(\theta^*, \sigma^*) \in \matcal{Y}^*$ from Equation (\ref{eq:fixed}) and the Kullback-Leibler information inequality. The uniqueness follows from condition (d). 

% \textit{Step 2.} Every element of $Y_M$ belongs to the union of a set of small open balls around the elements of $Y^*$ with probability approaching 1. \\
% Define a function
% \[\tilde{L}(\theta, \sigma) \equiv \max_{c \in \Theta}\{L(\c, \sigma)\} - L(\theta, \sigma).\]

% \textit{Step 3.} The function $\psi_M$ converges to $\psi^*$ in probability uniformly in $\sigma \in \mathcal{N}(\sigma^*)$.  

% \textit{Step 4.} There exists an element $(\tilde{\theta}_M, \tilde{\sigma}_M) \in Y_M$ in an open ball around $(\theta^*, \sigma^*)$ with probability approaching 1. 

% \textit{Step 5.}

% \end{proof}
% \end{comment}

% As \cite{KS12} mention in Proposition 6, it is also possible to change weak consistency to strong consistency by changing ``with probability approaching 1" to ``a.s." in Step 2-5.

\subsection{Proof of Proposition~\ref{convergence}}

We first prove an auxiliary proposition needed for convergence condition. This is analogous to Proposition 7 in \cite{KS12}.

\begin{proposition}\label{prop:update}
Suppose that Assumption \ref{convergence assumption} holds. Then, there exists a neighborhood $\mathcal{N}_1$ of $\sigma^*$ such that $\hat{\theta}^\npliter - \hat{\theta} = O(||\hat{\theta}^{\npliter-1}-\hat{\theta}||)$ a.s. and $\hat{\sigma}^\npliter - \hat{\sigma} = M_{\Psi_\theta} \Psi_{\sigma}^*(\hat{\sigma}^{\npliter-1} - \hat{\sigma}) + O(M^{-1/2}||\hat{\sigma}^{\npliter-1}-\hat{\sigma}|| +  ||\hat{\sigma}^{\npliter-1} - \hat{\sigma}||^2)$ a.s. uniformly in $\hat{\sigma}^{\npliter-1} \in \mathcal{N}_1$. 
\end{proposition}
\begin{proof}
For $\varepsilon>0$, define a neighborhood $\mathcal{N}(\varepsilon) = \{(\theta, \sigma): ||\theta - \theta^*|| + ||\sigma - \sigma^*|| < \varepsilon\}$. Then, there exists $\varepsilon_1>0$ such that $\mathcal{N}(\varepsilon_1) \subset \mathcal{N}$
 and $\sup_{(\theta, \sigma) \in \mathcal{N}(\varepsilon_1)}||\nabla_{\theta \theta'}L(\theta, \sigma)^{-1}||<\infty$ since $\nabla_{\theta \theta'}L(\theta, \sigma)$ is continuous and $\nabla_{\theta \theta'} L(\theta^*, \sigma^*)$ is nonsingular. 
 
 We assume that $(\hat{\theta}^{\npliter}, \hat{\sigma}^{\npliter-1}) \in \mathcal{N}(\varepsilon_1)$ which we will show in the end of this proof. 
 
 \textit{Step 1.} We first prove the first statement: $\hat{\theta}^\npliter - \hat{\theta} = O(||\hat{\theta}^{\npliter-1}-\hat{\theta}||)$ a.s.
 
 From the first step in NPL iteration, we know that $\nabla_\theta L_M(\hat{\theta}^\npliter, \hat{\sigma}^{\npliter-1})=0$. Applying the mean value theorem around $(\hat{\theta}, \hat{\sigma})$, we have 
\begin{equation}\label{step1} 0 = \nabla_{\theta \theta'}L_M(\bar{\theta}, \bar{\sigma})(\hat{\theta}^\npliter-\hat{\theta}) + \nabla_{\theta \sigma'} L_M(\bar{\theta}, \bar{\sigma})(\hat{\sigma}^{\npliter-1} - \hat{\sigma})+o_p(1).\end{equation}

where $(\bar{\theta}, \bar{\sigma})$ lie between $(\hat{\theta}^\npliter, \hat{\sigma}^{\npliter-1})$ and $(\hat{\theta}, \hat{\sigma})$. 

Then, we can rewrite 
\begin{equation}
\label{eq:theta}
\hat{\theta}^\npliter - \hat{\theta}  = - \nabla_{\theta \theta'} L_M(\bar{\theta}, \bar{\sigma})^{-1}\nabla_{\theta \sigma'}L_M(\bar{\theta}, \bar{\sigma})(\hat{\sigma}^{\npliter-1} - \hat{\sigma})+o_p(1).
\end{equation}
Note that (i) $(\bar{\theta}, \bar{\sigma}) \in \mathcal{N}(\varepsilon_1)$ from assumption and $(\bar{\theta}, \bar{\sigma}) \in \mathcal{N}$ since $\mathcal{N}(\varepsilon_1) \subset \mathcal{N}$, and (ii) $\sup_{(\theta, \sigma) \in \mathcal{N}(\varepsilon_1)}||\nabla_{\theta\theta'}L_M(\theta, \sigma)^{-1}\nabla_{\theta \sigma'} L_M(\theta, \sigma)|| = O(1)$ a.s. since $\sup_{(\theta, \sigma) \in \mathcal{N}(\varepsilon_1)}||\nabla_{\theta \theta'}L(\theta, \sigma)^{-1}|| < \infty$ and $\sup_{(\theta, \sigma) \in \mathcal{N}}||\nabla^2L_M(\theta, \sigma) - \nabla^2L(\theta, \sigma)|| = o(1)$ by general uniform convergence.\footnote{(i) $\Theta$ is compact implies total boundedness. (ii) $\nabla^2 L_M(\theta, \sigma) \to \nabla^2 L(\theta, \sigma) \ a.s.\ \forall \theta \in \Theta$. (iii) $\nabla^2L(\theta, \sigma)$ is a nonrandom function that is uniformly continuous in $\theta \in \Theta$, and $|\nabla^2L_M(\theta, \sigma) - \nabla^2 L_M(\theta', \sigma')|\leq |\nabla^2L_M(\theta, \sigma) - \nabla^2 L_M(\theta^*, \sigma^*)|+|\nabla^2L_M(\theta', \sigma') - \nabla^2 L_M(\theta^*, \sigma^*)| \leq 2\delta \leq \delta'd((\theta, \sigma), (\theta', \sigma')) \ \text{for some small} \ \delta \ \text{and} \ \delta', \ \forall \theta',\theta \in \mathcal{N}$ a.s. since $d(\cdot)<\varepsilon$ and $L_M(\theta, \sigma)$ is continuous.  Then, by Theorem 2 and Lemma 1 in \cite{Andrews}, the result follows.}

Then, we have the first result: $\hat{\theta}^{\npliter} - \hat{\theta} = O(||\hat{\sigma}^{\npliter-1} - \hat{\sigma}||)$ a.s.

\textit{Step 2.} We now prove $\hat{\sigma}^\npliter - \hat{\sigma} = M_{\Psi_\theta} \Psi_{\sigma}^*(\hat{\sigma}^{\npliter-1} - \hat{\sigma}) + O(M^{-1/2}||\hat{\sigma}^{\npliter-1}-\hat{\sigma}|| - ||\hat{\sigma}^{\npliter-1} + \hat{\sigma}||^2)$ a.s. uniformly in $\hat{\sigma}^{\npliter-1} \in \mathcal{N}_1$. 

We use the Taylor expansion, root-M consistency of $(\hat{\theta}, \hat{\sigma})$, and the information matrix equality for the results below:

\begin{align}\label{eq:step2}
    \nabla_{\theta\theta'}L_M(\hat{\theta}, \hat{\sigma}) & = - \Omega_{\theta\theta'} + O(M^{-1/2})\\
    \nabla_{\theta\sigma'}L_M(\hat{\theta}, \hat{\sigma}) & = - \Omega_{\theta\sigma'} + O(M^{-1/2})\nonumber \\
    \nabla_{\theta} \Psi(\hat{\theta}, \hat{\sigma}) & = \Psi_{\theta}^* + O(M^{-1/2}) \nonumber \\
    \nabla_{\sigma}\Psi(\hat{\theta}, \hat{\sigma}) & = \Psi_{\sigma}^* + O(M^{-1/2}) \nonumber
\end{align}

We use the second step of NPL estimation $\hat{\sigma}^\npliter = \Psi(\hat{\theta}^\npliter, \hat{\sigma}^{\npliter-1})$ and expand twice around $(\hat{\theta}, \hat{\sigma})$. 
\begin{align}\label{eq:sigma} 
\hat{\sigma}^\npliter - \hat{\sigma} & = \nabla_\theta \Psi(\hat{\theta}, \hat{\sigma})(\hat{\theta}^{\npliter}-\hat{\theta}) + \nabla_\sigma \Psi(\hat{\theta}, \hat{\sigma})(\hat{\sigma}^{\npliter-1}-\hat{\sigma}) + O(||\hat{\sigma}^{\npliter-1} - \hat{\sigma}||^2)  \\ 
& = \Psi_\theta^*(\hat{\theta}^{\npliter} - \hat{\theta}) + \Psi_\sigma^*(\hat{\sigma}^{\npliter-1} - \hat{\sigma}) + O(||\hat{\sigma}^{\npliter-1} - \hat{\sigma}||^2) +O(M^{-1/2}||\hat{\sigma}^{\npliter-1} - \hat{\sigma}||). \nonumber
\end{align}
where the first equality follows from Step 1 and $\sup_{(\theta, \sigma)\in \mathcal{N}(\varepsilon_1)}\nabla^3 \Psi(\theta, \sigma)<\infty$ in assumption (b) and the second equality comes from \eqref{eq:step2}.

We expand $\nabla_{\theta \theta'}L_M(\bar{\theta},\bar{\sigma})$ around ($\hat{\theta}, \hat{\sigma})$. Since $||\bar{\theta} - \hat{\theta}|| \leq ||\hat{\theta}^{\npliter} -\hat{\theta}||$, $||\bar{\sigma} - \hat{\sigma}|| \leq ||\hat{\sigma}^{\npliter-1}-\hat{\sigma}||$ and $\hat{\theta}^{\npliter-1} - \hat{\theta} = O(||\sigma^{\npliter-1}-\hat{\sigma}||)$ from Step 1, we have $\nabla_{\theta\theta'}L_M(\bar{\theta}, \bar{\sigma}) = \nabla_{\theta \theta'}L_M(\hat{\theta}, \hat{\sigma}) + O(||\hat{\sigma}^{\npliter-1} - \hat{\sigma}||)$ a.s. Using \eqref{eq:step2} and repeating the similar process for $\nabla_{\theta \sigma'}L_M(\bar{\theta}, \bar{\sigma})$, 
\begin{align*}
\nabla_{\theta \theta'} L_M(\bar{\theta}, \bar{\sigma}) & = - \Omega_{\theta \theta'} + O(M^{-1/2}) + O(||\hat{\sigma}^{\npliter-1}- \hat{\sigma}||) \\
\nabla_{\theta \sigma'} L_M(\bar{\theta}, \bar{\sigma}) & = - \Omega_{\theta \sigma'} + O(M^{-1/2}) + O(||\hat{\sigma}^{\npliter-1}- \hat{\sigma}||)
\end{align*}
Then, applying \eqref{eq:theta},
\[\hat{\theta}^{\npliter}-\hat{\theta} = - \Omega_{\theta \theta'}^{-1} \Omega_{\theta\sigma'}(\hat{\sigma}^{\npliter-1} - \hat{\sigma}) + O(M^{-1/2}||\hat{\sigma}^{\npliter-1} - \hat{\sigma}|| + ||\hat{\sigma}^{\npliter-1}-\hat{\sigma}||^2).\]
Substituting the equation above into \eqref{eq:sigma}, we have the result.

\textit{Step 3.} Now we need to show that $(\hat{\theta}^{\npliter}, \hat{\sigma}^{\npliter-1}) \in \mathcal{N}(\varepsilon_1)$. We first prove that $||\hat{\theta}^{\npliter} - \hat{\theta}^0||< \varepsilon/2$ and then, show that $||\hat{\sigma}^{\npliter-1} - \hat{\sigma}|| < \varepsilon/2$ if we set $\mathcal{N}_1$ sufficiently small. 

Let $\mathcal{N}_\theta \equiv \{\theta: ||\theta - \theta^*||<\varepsilon_1/2\}$ and define $\Delta = L(\theta^*, \sigma^*) - \sup_{\theta \in \mathcal{N}_\theta^c \cap \Theta}L(\theta, \sigma^*) >0$. The inequality follows from information inequality, compactness of $\mathcal{N}_\theta^c \cap \Theta$, and continuity of $L(\theta, \sigma)$. Then, if $\hat{\theta}^\npliter \notin \mathcal{N}_\theta$, then $L(\theta^*, \sigma^*) - L(\hat{\theta}^\npliter, \sigma^*) \geq \Delta$. We can also derive that 
\begin{align*}
   & L(\theta^*, \sigma^*)  - L(\hat{\theta}^\npliter, \sigma^*) \\
    & \leq L_M(\theta^*, \sigma^*) - L_M(\hat{\theta}^\npliter, \sigma^*) + 2\sup\nolimits_{(\theta, \sigma) \in \Theta \times \beliefspace} |L_M(\theta, \sigma) - L(\theta, \sigma)| \\
    & \leq L_M(\theta^*, \hat{\sigma}^{\npliter-1}) - L_M(\hat{\theta}^\npliter, \hat{\sigma}^{\npliter-1}) + 2\sup\nolimits_{(\theta, \sigma) \in \Theta \times \beliefspace} |L_M(\theta, \sigma) - L(\theta, \sigma)| + 2\sup\nolimits_{\theta \in \Theta}|L(\theta, \sigma^*) - L(\theta, \hat{\sigma}^{\npliter-1})| \\
    & \leq 2\sup\nolimits_{(\theta, \sigma) \in \Theta \times \beliefspace} |L_M(\theta, \sigma) - L(\theta, \sigma)| + 2\sup\nolimits_{\theta \in \Theta}|L(\theta, \sigma^*) - L(\theta, \hat{\sigma}^{\npliter-1})|.
\end{align*}
The last equality follows from $\hat{\theta}^{\npliter} = \argmax_{\theta \in \Theta}L_M(\theta, \hat{\sigma}^{\npliter-1})$. We know that $2\sup\nolimits_{(\theta, \sigma) \in \Theta \times \beliefspace} |L_M(\theta, \sigma) - L(\theta, \sigma)| = o(1)$ a.s. from Step 1 of Proof for consistency in Proposition \ref{prop:largegame2}. Since $L(\theta, \sigma)$ is continuous, there exists $\varepsilon_\Delta$ such that  $2\sup\nolimits_{\theta \in \Theta}|L(\theta, \sigma^*) - L(\theta, \hat{\sigma}^{\npliter-1})|<\Delta/2$ if $||\sigma^* - \hat{\sigma}^{\npliter-1}||\leq \varepsilon_\Delta$. It follows that if $||\sigma^* - \hat{\sigma}^{\npliter-1}||\geq \varepsilon_\Delta$, $L(\theta^*, \sigma^*) - L(\hat{\theta}^\npliter, \sigma^*)\leq\Delta$ which means $\hat{\theta}^\npliter \in \mathcal{N}_\theta$. 

Then, if we set $\mathcal{N}_1 = \{\sigma:||\sigma - \sigma^* || \leq \min\{\varepsilon_1/2, \varepsilon_{\Delta}\}\}$, we have $||\sigma - \sigma^*||<\varepsilon/2$ which gives $(\hat{\theta}^{\npliter}, \hat{\theta}^{\npliter-1}) \in \mathcal{N}(\varepsilon_1)$ a.s. 
\end{proof}

We now present a proof for Proposition \ref{convergence}. This is analogous to Proposition 2 in \cite{KS12}.

\begin{proof}
Let $b>0$ be a constant such that 
\[\rho (M_{\Psi_\theta} \Psi_\sigma)+2b <1. \]
From Lemma 5.6.10 of \cite{Horn1985}, there is a matrix norm $||\cdot||_\alpha$ such that 
\[||M_{\Psi_\theta}\Psi_\sigma||_\alpha \leq r(M_{\Psi_\theta}\Psi_\sigma)+b.\]
Define a vector norm $||\cdot||_\beta$ for $x \in \mathbb{R}^{NJK}$ as $||x||_\beta \equiv ||[x \ 0 \cdots 0]||_\alpha$, then a direct calculation gives 
\[||Ax||_\beta = ||A [x\ 0 \cdots 0 ]||_\alpha \leq ||A||_\alpha ||x||_\beta\]
for any matrix $A$. 

From the equivalence of vector norms in $\mathbb{R}^{NJK}$, we can restate Proposition \ref{prop:update} in terms of $\beta$: there exists $c>0$ such that 
\[\hat{\sigma}^{\npliter}-\hat{\sigma} = M_{\Psi_\theta}\Psi_\sigma(\hat{\sigma}^{\npliter-1}-\hat{\sigma})+O(M^{-1/2}||\hat{\sigma}^{\npliter-1}-\hat{\sigma}||_\beta + ||\hat{\sigma}^{\npliter-1}-\hat{\sigma}||_\beta^2)\]
a.s. holds uniformly in $\hat{\sigma}^{\npliter-1} \in \{\sigma: ||\sigma-\sigma^*||_\beta<c\}$.

We rewrite this statement further so that it is amenable to recursive substitution.\\
(i) $||M_{\Psi_\theta}\Psi_\sigma(\hat{\sigma}^{\npliter-1}-\hat{\sigma})||_\beta \leq ||M_{\Psi_\theta}\Psi_\sigma||_\alpha ||\hat{\sigma}^{\npliter-1}-\hat{\sigma}||_\beta \leq (r(M_{\Psi_\theta}\Psi_\sigma)+b)||\hat{\sigma}^{\npliter-1}-\hat{\sigma}||_\beta$.\\
(ii) The remainder term can be written as $O(M^{-1/2}+||\hat{\sigma}^{\npliter-1}-\hat{\sigma}||_\beta)||\hat{\sigma}^{\npliter-1}-\hat{\sigma}||_\beta$. Then setting $c<b$ and using consistency of $\hat{\sigma}$, this term is smaller than $b||\hat{\sigma}^{\npliter-1}-\hat{\sigma}||_\beta$ a.s.\\
(iii) Since $\hat{\sigma}$ is consistent, $\{\sigma:||\hat{\sigma}-\hat{\sigma}||_\beta<c/2\} \subset \{\sigma:||\sigma-\sigma^*||_\beta<c\}$ a.s.\\
From (i) to (iii), 
\[||\hat{\sigma}^{\npliter-1}-\hat{\sigma}||_\beta \leq (r(M_{\Psi_\theta}\Psi_\sigma)+2b)||\hat{\sigma}^{\npliter-1}-\hat{\sigma}||_\beta\]
holds a.s. for all $\hat{\sigma}^{\npliter-1} \in \{\sigma:||\hat{\sigma}^{\npliter-1}-\hat{\sigma}||_\beta<c/2\}$. Because each NPL updating of $(\theta, \sigma)$ uses the same pseudo-likelihood function, we may recursively substitute for the $\hat{\sigma}_j$'s.
\begin{align*}
    ||\hat{\sigma}^{\npliter-1}-\hat{\sigma}||_\beta & \leq (r(M_{\Psi_\theta}\Psi_\sigma)+2b)||\hat{\sigma}^{\npliter-1}-\hat{\sigma}||_\beta \\
    & \leq (r(M_{\Psi_\theta}\Psi_\sigma)+2b)^2||\hat{\sigma}^{\npliter-2}-\hat{\sigma}||_\beta \\
    & \vdots \\
    & \leq (r(M_{\Psi_\theta}\Psi_\sigma)+2b)^\npliter||\hat{\sigma}^0-\hat{\sigma}||_\beta
\end{align*}
Then, $\lim_{\npliter\to \infty}\hat{\sigma}^\npliter=\hat{\sigma}$ a.s. if $||\hat{\sigma}^{\npliter-1}-\hat{\sigma}||_\beta<c/2$. Applying the equivalence of vector norms in $\mathcal{R}^L$ to $||\hat{\sigma}^0-\hat{\sigma}||_\beta$ and $||\hat{\sigma}^0-\hat{\sigma}||$ and consistency of $\hat{\sigma}$, the result follows. 
\end{proof}

\section{Implementation Details}
\label{app:detail}

In this section, we provide more details on estimation process. Our goal is to estimate $\theta = (\theta_{\text{FC,1}}, \ldots, \theta_{\text{FC,5}}, \theta_{\text{RS}}, \theta_{\text{RN}}, \theta_{\text{EC}})$ using the NPL algorithm. We first review two steps needed for estimating:
\begin{enumerate}
\item Given $\hat{\sigma}^{\npliter-1}$, update $\hat{\theta}$ by 
\[ \hat{\theta}^{\npliter} = \argmax_{\theta \in \Theta} L_M(\theta, \hat{\sigma}^{\npliter-1})= \argmax_{\theta \in \Theta}\frac{1}{M}\sum_{m=1}^M \sum_{n=1}^T \ln P_{k_{m,n-1},k_{mn}}(\Delta; \Psi(\theta, \hat{\sigma}^{\npliter-1})) \tag{\ref{eq:npl1}}.\]
\item Update $\hat{\sigma}$ using the equilibrium condition, i.e.
\[\hat{\sigma}^{\npliter} = \Psi(\hat{\theta}^{\npliter}, \hat{\sigma}^{\npliter-1}) \tag{\ref{eq:npl2}} \]
where 
\[ \hat{\sigma}_{jk} = \frac{\exp(\psi_{jk}+V_{l(j,k)})}{\sum_j' \exp(\psi_{j'k}+V_{l(j',k)})}.\]
\end{enumerate}

We need to express the likelihood function in terms of $\theta$ to maximize it with regard to $\theta$. First, we start from calculating the value function in \eqref{eq:ccpexp} using value function below: 
\[V_i(\theta, \sigma)=\Big[(\rho+N\lambda)I-\lambda \sum_{m=1}^N \Sigma_m(\sigma_m) -Q_0\Big]^{-1}[u_i(\theta)+\lambda_i E_i(\theta, \sigma)]  \tag{\ref{eq:value}} \]
where $\Sigma_m(\sigma_m)$ is the $K\times K$ state transition matrix induced by the actions of player $m$ given the choice probabilities $\sigma_m$ and where $E_i(\theta, \sigma)$ is a $K \times 1$ vector where each element $k$ is the ex-ante expected value of the choice-specific payoff in state $k$, $\sum_j \sigma_{ijk}[\psi_{ijk}+e_{ijk}(\theta, \sigma)]$.

We can define the first parenthesis in \eqref{eq:value} as $\Xi$. We assume $\rho_i=1$, $\lambda_{ik}=1$ for all $i=1, \ldots, N$ and $k=1, \ldots, K$, and that $Q_0$ is known. We can also calculate $\Sigma_m(\sigma_m)$ using $\hat{\sigma}_m^{\npliter-1}$. We start from the true probabilities $\sigma^*$ for the 2S-True estimator and from random draws for $\hat{\sigma}^0$ for NPL-random estimator.  

Then, we rewrite the second parenthesis in \eqref{eq:value}, in terms of $\theta$. The flow payoff $u_{ik}$ can be rewritten as:
\begin{align*}
u_{ik} & =\theta_{\text{RS}}\ln(s_k) - \theta_{\text{RN}}\ln\big(1+\sum_{m\neq i}a_{mk}\big)- \theta_{FC,i} \\
& = z_{ik}^u\theta^u
\end{align*}
where $z^u_{ik} = [I_{ik} \ \ln(s_k) \ \ln\big(1+\sum_{m\neq i}a_{mk}\big)]$ where $I_{ik}$ is a $1\times 5$ vector with 1 in $i$-th position and zero elsewhere, and $\theta^u = (\theta_{\text{FC,1}}, \ldots, \theta_{\text{FC,5}}, \theta_{\text{RS}}, \theta_{\text{RN}})'$. 

Now we express $E_{ik}(\theta, \sigma)$, corresponding element of $E(\theta, \sigma)$ in term of $\theta$. The determinant part of instantaneous payoff has the structure:
\[\psi_{ijk} = z^\psi_{ijk}\theta_{\text{EC}}\]
where
\[z_{ijk}^\psi=\begin{cases} -1 & \text{if} \ j=1 \\
0 & \text{otherwise} \end{cases}\]

The expected value of deterministic instantaneous payoff  $z^\psi_{ijk}$ differs by the choice, so
\begin{align*}
\E z_{ik}^\psi & = \sum_{j\in \mathcal{A}}\sigma_{ijk} z_{ijk}^\psi \\
&= \lambda[\sigma_{0k} \times 0 + \sigma_{1k} \times (-1)]
\end{align*}

The stochastic part $\varepsilon_{ijk}$ has the expected value $e_{ijk}=\E[\varepsilon_{ijk}]$:
\begin{equation*}
  e_{ijk} = \lambda[\sigma_{0k}(\gamma-\log(\sigma_{0k})) + \sigma_{1k}(\gamma-\log(\sigma_{1k})) + \sigma_{-1,k}(\gamma-\log(\sigma_{-1,k}))]
\end{equation*}
using T1EV distribution of $\varepsilon_{ijk}$. 

Collecting all terms for the value function, we can rewrite the value function as
\begin{align*} V_{ij}(\theta, \sigma) & = \Xi^{-1}[u(\theta) + \lambda E(\theta, \sigma)] \\
& = \Xi^{-1}[z_{ij}^u\theta^u + \lambda \E(z_{ij}^\psi\theta_{\text{EC}} + \varepsilon_{ij})] \\ 
& = \Xi^{-1}\E[z_{ij}\theta+\varepsilon_{ij}] \\
& = W_{ij} \theta' + \tilde{e}_{ij}. \end{align*}
where $z_{ij}$ is a $K \times |\theta|$ vector with each $k$th element $(z_{ik}^u, z_{ijk}^\psi)$,  $W_{ijk}$ is a $K \times 1$ vector which equals $\Xi^{-1}\E z_{ij}$, and $\tilde{e}_{ij}$ is $\Xi^{-1}e_{ij}$.  Now we can write conditional choice probabilities in terms of $\theta$, 
\begin{equation}\label{eq:estimation} \hat{\sigma}_{ijk} = \frac{\exp(W_{ijk}\theta'+\tilde{e}_{ijk})}{\sum_{j'} \exp(W_{ij'k}\theta'+\tilde{e}_{ij'k})}\end{equation}

Since we can calculate $W_{ijk}$ and $\tilde{e}_{ijk}$ from known parameters or conditional choice probabilities, the only unknown term in $\hat{\sigma}$, hence the likelihood function, is $\theta$. Then, we can form $Q$ in terms of $\theta$ using $Q_0$ and $\hat{\sigma}$, and maximize the likelihood with regard to $\theta$. We use starting values of $\theta^0= (1,1,1,1,1,1,1,1)$. We use the solver \texttt{fminunc} in MATLAB. 

Once $\hat{\theta}^\npliter$ is estimated, the conditional choice probability $\hat{\sigma}^\npliter$ is updated. The entire process is repeated 20 times (\Npliter = 20). For Monte Carlo experiments, we generate 100 different data set using same parameters and estimate separately for each data set. Then, we report the mean of estimates as estimates for each estimator and the standard deviation from 100 estimates as standard error. 

\section{Estimation of Misspecified Continuous Time Models}
\label{app:dtct}

Discrete time models are widely used in the literature and so
previously in Section~\ref{sec:mis} we considered estimating a
misspecified discrete time, simultaneous move model when the data
generating process was a continuous time model with asynchronous
moves.
Here, we report the results of the reverse form of misspecification:
the data generating process is a simultaneous move game in discrete
time and the estimated model is a continuous time model with
asynchronous moves.
As before, we suppose the data has been aggregated to discrete time
intervals and note that this is an additional disadvantage for the
continuous time model, in addition to misspecification.

The results are reported in Table~\ref{table:mis:dtdt:ctdt}.
In this case there are also large biases in the parameter estimates.
Therefore, if simultaneous moves are an important institutional
feature in an application (e.g., a repeated sealed-bid auction game),
then estimating a continuous time model with asynchronous moves would
likely result in misleading inference.
% However, the current state of the literature is such that simultaneous
% move discrete time models are the de facto standard, so the overall
% risk of this form of misspecification seems much lower.
When considering a continuous time model, we suggest that researchers
ask whether or not an ideal (i.e., not time aggregated) dataset would
contain distinct dates or times of player actions.
If so, then an asynchronous move, continuous time model may be well
suited to the application.
If not, then a simultaneous move model in discrete time is likely to
be a better choice.

\begin{table}[tbph]
\centering
\begin{small}
\begin{threeparttable}
\caption{Misspecified Monte Carlo Results (Continuous Time Estimation of Discrete Time DGP)}
\label{table:mis:dtdt:ctdt}
\begin{tabular}{llccccc}
\toprule
        &                      & \multicolumn{4}{c}{Parameters}                                          \\
        \cline{3-6}
Exp. & Values  & $\theta_{\text{FC,1}}$ & $\theta_{\text{RS}}$ & $\theta_{\text{EC}}$ & $\theta_{\text{RN}}$ \\ \midrule
\textbf{1} & \textbf{True values} & \textbf{-1.9000} & \textbf{1.0000} & \textbf{1.0000} & \textbf{0.0000} \\
 & Correct & -1.8830 (0.4711) & 1.0383 (0.1952) & \ 1.0146 (0.3040) & \  0.0747 (0.4094)  \\
 & Misspecified & -5.0609 (0.7332) & 1.4650 (0.2091) & -1.2079 (0.3799) & -1.9307 (0.6838)  \\
 & Bias & -2.3015 & 0.4821 & -2.2504 & -1.1784 \\
             \addlinespace[0.2cm]
\textbf{2} & \textbf{True values} & \textbf{-1.9000} & \textbf{1.0000} & \textbf{1.0000} & \textbf{1.0000} \\
 & Correct & -1.9558 (0.3573) & 1.0461 (0.1775) & \ 0.9880 (0.2527) & \ 1.1051 (0.4033) \\
 & Misspecified & -5.5997 (0.8459) & 1.8943 (0.2463) & -1.7383 (0.4201)  & -0.0256 (0.327) \\
 & Bias &  -2.4794 & 0.4477 & -2.7987 & -1.1142 \\
             \addlinespace[0.2cm]
\textbf{3} & \textbf{True values} & \textbf{-1.9000} & \textbf{1.0000} & \textbf{1.0000} & \textbf{2.0000} \\
 & Correct &  -1.9521 (0.3397) & 1.0416 (0.1953) & \ 0.9912 (0.2159) & \ 2.0869 (0.5427) \\
 & Misspecified & -5.5421 (0.9863) & 1.4382 (0.3297)& -1.5748 (0.4795) &  \ 0.2851 (0.7541)  \\
 & Bias & -2.2535 & 0.0458 & -2.5610 & -1.9033 \\
             \addlinespace[0.2cm]
\textbf{4} & \textbf{True values} & \textbf{-1.9000} & \textbf{1.0000} & \textbf{0.0000} & \textbf{1.0000} \\
 & Correct & -1.9443 (0.4427) & 0.9993 (0.1939) & \  0.0068 (0.3146) & \ 0.9524 (0.4843) \\
 & Misspecified & -8.6037 (1.1178) & 2.7119 (0.4770) & -4.3979 (0.6835) & -0.4729 (0.6099)   \\
 & Bias & -3.9897 & 0.8670 & -4.9044 & -1.3647 \\
             \addlinespace[0.2cm]
\textbf{5} & \textbf{True values} & \textbf{-1.9000} & \textbf{1.0000} & \textbf{2.0000} & \textbf{1.0000} \\
 & Correct &  -1.9285 (0.3105) & 1.0209 (0.1811) & \ 1.9961 (0.2065) & \ 1.0190 (0.3562) \\
 & Misspecified & -3.4500 (0.4940)	& 1.2746 (0.1470) &	\ 0.7175 (0.2127) & \ 0.3706 (0.4182)	 \\
 & Bias & -1.7883 & 0.3181 & -1.3200 & -0.7196 \\
             \addlinespace[0.2cm]
\textbf{6} & \textbf{True values} & \textbf{-1.9000} & \textbf{1.0000} & \textbf{4.0000} & \textbf{1.0000} \\
 & Correct &  -1.9921 (0.2694) & 1.0570 (0.1675) & \ 4.0473 (0.2497) & \ 1.0799 (0.2934) \\ 
 & Misspecified & -2.2741 (0.2431) & 0.9331 (0.1407) &  \ 3.5844 (0.2987)   & \ 0.5177 (0.1965)\\
 & Bias &   -2.0902     &    0.6627  &        0.1001     &    -0.4752 \\
\bottomrule
\end{tabular}
\begin{tablenotes}
  \footnotesize
\item Displayed values are means with standard deviations in parentheses.
\end{tablenotes}
\end{threeparttable}
\end{small}
\end{table}

\section{Multiplicity and Stability}
\label{sec:psd08}

To investigate multiplicity of equilibria and stability of CNTPL fixed points, we considered
a continuous-time version of the simple entry-exit model of \cite{PS08}.
This is a discrete time model that was also considered by
\cite{Dearing-blevins-2021} and \cite{Aguirregabiria2019}.
In the model there are $N=2$ players with $J=2$ actions.
The only state variable in the model is the incumbency status of the firms.
Let $x_{k,i} = 1$ denote firm $i$'s activity in the market and $x_{k,i} = 0$ denote that firm $i$ is inactive.
The state of the model can equivalently be represented by a single integer state with $K=4$ values $\mathcal{K} = \lbrace 1,2,3,4 \rbrace$ representing the previous action tuples $x_k = (x_{k,1},x_{k,2}) \in \mathcal{X} = \lbrace (0,0), (0,1), (1,0), (1,1) \rbrace$.
Exit is not permanent and firms can re-enter after exiting.
Monopoly firms in the market earn a constant flow payoff $\theta_\text{M}$ while firms in duopoly earn $\theta_{\text{M}} + \theta_{\text{C}}$, where $\theta_{\text{C}}$ is negative.
Firms pay a cost $\theta_{\text{EC}}$ to enter and exiting firms receive a scrap value of $\theta_{\text{SV}}$.

In terms of our general model notation we have
\begin{equation*}
  u_{ik} = \begin{cases}
             0 & \text{if } x_{k,i} = 0, \\
             \theta_{\text{M}} & \text{if } x_{k,i} = 1 \text{ and } x_{k,2-i} = 0,\\
             \theta_{\text{M}} + \theta_{\text{C}} & \text{if } x_{k,i} = 1 \text{ and } x_{k,2-i} = 1,
           \end{cases}
           \quad
  \psi_{ijk} =
  \begin{cases}
    0 & \text{if } j=0\\
    \theta_{\text{EC}} & \text{if } j=1 \text{ and } x_{k,i} = 0 \\
    \theta_{\text{SV}} & \text{if } j=1 \text{ and } x_{k,i} = 1 \\
  \end{cases}.
\end{equation*}

We searched for equilibria under several parameter vectors.
In all cases we specified $\lambda = 1$ and $\rho = 0.05$.
In each case we solved the system of equilibrium conditions with 10,000 random initial choice probability vectors.
Here, we summarize two specifications.\\
\textbf{Specification 1:} $(\theta_{\text{M}},\theta_{\text{C}},\theta_{\text{EC}},\theta_{\text{SV}}) = (1.2, -2.4, -0.2, 0.1)$.
Here, we found a unique stable, symmetric equilibrium with
\begin{align*}
\sigma_{11\cdot} &= ( 0.8326, 0.3990, 0.1436, 0.5619 ), \\
\sigma_{21\cdot} &= ( 0.8326, 0.1436, 0.3990, 0.5619 ).
\end{align*}\\
\textbf{Specification 2:} $(\theta_{\text{M}},\theta_{\text{C}},\theta_{\text{EC}},\theta_{\text{SV}}) = (2.0, -4.0, -1.0, 0.1)$.
In this case, we found one unstable, symmetric equilibrium with
\begin{align*}
\sigma_{11\cdot} &= ( 0.9577, 0.4469, 0.0131, 0.3570 ),\\
\sigma_{21\cdot} &= ( 0.9577, 0.0131, 0.4469, 0.3570 ),
\end{align*}
as well as two stable, asymmetric equilibria with
\begin{align*}
\sigma_{11\cdot} &= ( 0.6677, 0.1020, 0.1753, 0.7794 ),\\
\sigma_{21\cdot} &= ( 0.9952, 0.0010, 0.8526, 0.0608 ),
\end{align*}
and vice versa.

Therefore, we should not expect that in general the continuous time model will have a unique equilibrium or that the CTNPL mapping will be stable.
Although we did not find evidence of multiple equilibria in the five-firm heterogeneous entry-exit model used for the Monte Carlo section, this is a much higher-dimensional model and so a similar equilibrium search is more difficult in that setting.

\end{appendix}
\end{document}